\documentclass[a4paper]{article}
\usepackage[top=3cm, left=2cm, bottom=3cm, right=2cm]{geometry}
\usepackage{graphicx}
\usepackage{amsmath, amsthm, amssymb} 
\usepackage{overpic}
\usepackage{dsfont}
\usepackage[T1]{fontenc}
\usepackage{color}
\usepackage{natbib}
\usepackage[T1]{fontenc}
\usepackage{setspace}

\allowdisplaybreaks

\DeclareMathOperator\sgn{sgn}

\DeclareMathOperator\diag{diag}
\DeclareMathOperator\arcoth{arcoth}
\DeclareMathOperator\trace{tr}

\DeclareMathOperator\Cov{Cov}

\newtheorem{theorem}{Theorem}[section]

\newtheorem{prop}[theorem]{Proposition}
\newtheorem{corollary}[theorem]{Corollary}
\newtheorem{cor}[theorem]{Corollary}

\newtheorem{lem}[theorem]{Lemma}
\newtheorem{remark}[theorem]{Remark}

\newtheorem{assumption}[theorem]{Assumption}
\newtheorem{nota}[theorem]{Notation}

\newtheorem*{theorem*}{Theorem}
\newtheorem*{proposition*}{Proposition}
\newtheorem*{lemma*}{Lemma}
\newtheorem*{corollary*}{Corollary}
\newtheorem*{remark*}{Remark}

\newtheorem*{observation*}{Observation}
\newtheorem*{example*}{Example}
\newtheorem*{assumption*}{Assumption}

\theoremstyle{definition}
\newtheorem{definition}[theorem]{Definition}
\newtheorem*{definition*}{Definition}

\title{\begin{Huge}Portfolio liquidation in dark pools in continuous time\end{Huge}\footnote{We wish to thank Ulrich Horst, Werner Kratz and Nicholas Westray for useful discussions and comments. We are also grateful to seminar participants at the University of Bonn and Humboldt University Berlin.
This research was supported by Deutsche Bank through the Quantitative Products Laboratory.}}

\author{Peter Kratz\footnote{Humboldt Universit\"at zu Berlin, Unter den Linden 6,
10099 Berlin, Germany. Email: kratz@mathematik.hu-berlin.de}, 
Torsten Sch\"oneborn\footnote{Deutsche Bank AG, London, United Kingdom. Email: schoeneborn@math.tu-berlin.de}}

\date{July 02, 2012
}
\begin{document}

\maketitle

\begin{abstract}
\noindent
We consider an illiquid financial market where a risk averse investor has to liquidate a portfolio within a finite time horizon $[0,T]$ and can trade continuously at a traditional exchange (the ``primary venue'') and in a dark pool. At the primary venue, trading yields a linear price impact. In the dark pool, no price impact costs arise but order execution is uncertain, modeled by a multi-dimensional Poisson process. We characterize the costs of trading by a linear-quadratic functional which incorporates both the price impact costs of trading at the primary exchange and the market risk of the position.
The liquidation constraint implies a singularity of the value function of the resulting minimization problem at the terminal time $T$. Via the HJB equation and a quadratic ansatz, we obtain a candidate for the value function which is the limit of a sequence of solutions of initial value problems for a matrix differential equation. We show that this limit exists by using an appropriate matrix inequality and a comparison result for Riccati matrix equations. Additionally, we obtain upper and lower bounds of the solutions of the initial value problems, which allow us to prove a verification theorem. If a single asset position is to be liquidated, the investor slowly trades out of her position at the primary venue, with the remainder being placed in the dark pool at any point in time. For multi-asset liquidations this is generally not the case; it can, e.g., be optimal to oversize orders in the dark pool in order to turn a poorly balanced portfolio into a portfolio bearing less risk.
\end{abstract}

\section{Introduction}

In the last years, equity trading has been transformed by the advent of so called dark pools. These alternative trading venues differ significantly from classical exchanges and have gained a considerable market share, especially in the US. Dark pools vary in a number of properties such as crossing procedure, ownership and accessibility (see~\cite{Mittal2008} and \cite{Degryse2009a} for further 
details and a typology of dark pools). However, they generally share the following 
two stylized facts. First, the liquidity available in dark pools 
is not quoted, hence making trade execution uncertain 
and unpredictable. Second, dark pools do not determine prices. Instead, 
they monitor the prices determined by the classical 
exchanges and settle trades in the dark pool only if possible at these prices. 
Thus, trades in the dark pool have no or less price impact.\footnote{
For empirical evidence of lower 
transaction costs or price impact of dark pools compared to classical exchanges see, e.g., 
\cite{Conrad2003} and~\cite{Fong2004}.}

In this paper, we are concerned with the solution of a stochastic optimal control problem
in continuous time arising in the context of optimal portfolio liquidation if an investor
has access both to a classical exchange (also called ``primary venue''
or ``primary exchange'') and to a dark pool.\footnote{
The overall liquidity traded in dark pools in the US is strongly fragmented 
among approximately 40 different venues, see e.g.,~\cite{Carrie2008}. Therefore, 
liquidity aggregation 
is a major issue. \cite{Ganchev2009} and \cite{Laruelle2009} establish learning algorithms 
to achieve optimal order split \emph{between} dark pools. Instead of analyzing the simultaneous 
use of several dark pools, 
we consider such an ``aggregated'' dark pool in our model which we call ``the dark pool''.}
We study a continuous-time model for 
optimal liquidation of a portfolio consisting of $n$ assets within a finite time horizon $[0,T]$
reflecting the trade-off between execution uncertainty of dark pool orders and price impact costs of trading at the primary venue. It complements the model by~\cite{Kratz2010}, where (in particular) multi asset liquidation in dark pools is studied in discrete time; for single asset liquidation, the two models are connected via a convergence result. The mathematical analysis of the continuous-time model is substantially different from the discrete-time model. In discrete time, the optimal liquidation strategy and its costs are given by a backward recursion obtained by standard dynamic programming; this recursion is rather unhandy and it is difficult to deduce properties of the optimal strategy from it. In continuous time, the liquidation constraint imposes a singular boundary condition for the value function and therefore the problem cannot be solved by standard methods of stochastic control. Instead, these methods had to be modified and extended. The optimal strategy and the value function are characterized via the ``principal solution'' of a matrix differential equation. We believe that the analysis of the differential equation and the resulting verification theorem for the optimization problem is mathematically interesting in its own right. Additionally, the differential equation is easier to analyze than the (discrete-time) backward recursion and we can deduce additional properties of the value function and the optimal strategy from it; in particular, we are able to show that for a two asset portfolio, the optimal strategy is monotone in the correlation of the assets. To our best knowledge, the mathematical framework established here and in~\cite{Kratz2010} is the first within which optimal trade execution is analyzed for a multi asset portfolio using a classical exchange and a dark pool simultaneously.

Our model for trading and price formation at the classical exchange is a linear price impact model. Trade execution can be enforced by selling aggressively, which however results in quadratic execution costs due to a stronger price impact. We model order execution in the dark pool by an $n$-dimensional Poisson process. Orders submitted to the dark pool are executed at the jump times of the respective components of the process. The split of orders between dark pool and exchange is thus driven by the trade-off between execution uncertainty and price impact costs. 

The investor aims to maximize her expected proceeds minus a quadratic risk term. For appropriate assumptions on the price dynamics, this yields a linear-quadratic cost functional. However, the liquidation constraint implies a singularity of the value function at the terminal time $T$, which renders the solution of the optimization problem difficult.
We approximate the constraint by a sequence of modified unconstrained optimization problems with increasing finite end-costs. Via a quadratic ansatz, the corresponding Hamilton-Jacobi-Bellman (HJB) equation suggests that the value functions of the modified optimization problems are quadratic forms for matrix-valued functions which are the solutions of initial value problems for a specific matrix differential equation. Explicit solutions for these initial value problems are not known. We establish a matrix inequality which allows us to apply a known comparison result for Riccati matrix differential
equations in order to obtain closed form upper and lower bounds for the solutions
of the initial value problems. The bounds are important for several reasons.
Firstly, they enable us to prove a verification theorem for the modified unconstrained optimization problems
with finite end-costs.
Secondly, they imply the existence of the limit of the solutions of these optimization problems
and thus yield a well-defined candidate for the value function of the original optimization problem with
liquidation constraint.
Thirdly, the limit of the bounds transfers to bounds for this candidate value function and we deduce
a verification theorem for the solution of the optimal liquidation problem.

Hence, the value function of the problem is a quadratic form for a matrix-valued function which is the ``principal solution'' of the matrix differential equation. The optimal strategy is linear in the portfolio position; it is determined by the value function and can be computed easily from a numerical solution of the ``principal solution'' of the matrix differential equation. This makes the model tractable for practicable applications and allows us to investigate the properties of the solution. When a position in a single asset is to be liquidated, the current asset position is at all times being offered in the dark pool, while it is liquidated in parallel at the primary exchange. The opportunity to trade in the dark pool leads to a slower liquidation at the primary exchange compared to a market without a dark pool. Traders hence need to fundamentally adjust their trade execution algorithms when a dark pool is introduced. It is not sufficient to use an algorithm that was optimal for execution at the exchange and to add a component that also places trades in the dark pool; instead, trading at the exchange needs to be adjusted during the entire trade execution time interval. 
While in a single asset setting the entire asset position is placed in the dark pool, this is not true if a multi asset portfolio is to be liquidated. If, e.g., the portfolio is balanced and thus only exposed to little market risk, then a complete liquidation of the position in one of the assets is unfavorable and thus only a fraction of the entire portfolio should be placed in the dark pool. This highlights again that overly simple adjustments to existing trade execution algorithms are exposed to potential pitfalls. For dark pools, the reluctance of traders to place balanced portfolios in a dark pool is an incentive to offer balanced executions in order to attract more liquidity. 

Our paper is connected  to several lines of research. First, it builds on research on optimal trade execution strategies for a single trader in models with exogenously given liquidity effects. Several such models have been proposed for classical trading venues, e.g., 
\cite{Bertsimas1998}, \cite{Almgren2001}, \cite{Almgren2003}, \cite{Obizhaeva2006},  \cite{Schied2008},
\cite{Alfonsi2007} and~\cite{Schied2010}. We follow~\cite{Almgren2001} and assume a linear temporary price impact model for the primary venue. This choice yields a tractable model which nevertheless captures price impact effects. Linear price impact models have become the basis of several theoretical studies, e.g.,  \cite{Almgren2007}, \cite{Carlin2007}, \cite{Schoneborn2007} and~\cite{Rogers2010}. Furthermore, they demonstrated reasonable properties in real world applications and serve as the basis of many optimal execution algorithms run by practitioners (see e.g., \cite{Kissell2003}; \cite{Schack2004}; \cite{Abramowitz2006}; \cite{Leinweber2007}). 
The models above are limited to trading on one venue; they are concerned with the trade-off of execution costs versus market risk. By using a dark pool, the investor additionally faces order execution risk in our model. The trade-off of execution costs, market risk and execution risk is also apparent in the models of~\cite{Bayraktar2012} and~\cite{Gueant2012} who consider optimal liquidation with limit orders.\footnote{In our model all orders in the primary venue are market orders whose execution is guaranteed; hence, the execution risk only applies to the dark pool orders.} In these models, the investor choses a limit price for her orders; the ``probability'' of execution is then dependent on this price and the left-over position can be liquidated at fixed costs at the exchange. In contrast to this, the dark pool orders in our model do not involve a limit price; they are executed at the price of the primary venue. 
The continuous-time model studied in our paper complements the discrete-time analysis of~\cite{Kratz2010} for the specific case of linear price impact and no adverse selection. 
\cite{Kloeck2011} also work in a market model with exogenously specified liquidity characteristics for both a public exchange as well as a dark pool. Their focus is on the circumstances that can lead to price manipulation, while we focus on the quantitative and qualitative features of optimal trade execution strategies in markets without price manipulation. Therefore they investigate a more general class of price impact relationships, while our framework incorporates several aspects of market and investor behavior that shape optimal trade execution, such as risk aversion, multiple correlated assets and a dynamic usage of the dark pool.

A second line of research connected to our paper focuses on the underlying mechanisms for illiquidity effects. Early examples include \cite{Kyle1985}, \cite{Glosten1985} and \cite{Grossman1988}. In these models, price impact arises endogenously through the interplay of market participants. More recently, such models have been proposed to analyze the competition between classical trading venues and dark pools. \cite{Hendershott2000} analyze the interaction of dealer markets and a crossing network\footnote{Crossing networks are specific types of dark pools that offer periodic matching of buy and sell orders.} in a static one period framework where each investor buys or sells a \emph{single} share. Their findings include that trading in a crossing network is cheaper than in dealer markets since the trader saves half the spread, but trade execution is uncertain in the crossing network. In a similar setting, \cite{Donges2006} focus on game-theoretic refinements in order to remove the multiplicity of equilibria in~\cite{Hendershott2000}. \cite{Degryse2009} introduce a dynamic multi-period framework and compare the effect of different levels of transparency of the dark pool. In order to focus on the optimal order execution of an individual trader, we exogenously specify the liquidity properties of the market. Hence, all the models cited above, can shed more light on how the liquidity properties of the dark pool come about than our model can. On the other hand, by defining the model parameters exogenously, we can allow the investor to split her orders over time as well as between the two trading venues. While executing her order over time, she can dynamically react to the existence or absence of liquidity in the dark pool and adjust her trading strategy accordingly. Furthermore, we can take a multi asset perspective and investigate how the composition of basket portfolios influences optimal trade execution strategies. 

Methods of stochastic control is a third line of research that this paper is connected to. A standard reference for stochastic control with jumps is the book by~\cite{OksendalSulem}. The liquidation constraint yields a singularity of the value function at the terminal time
$T$, and thus, the resulting stochastic control
problems require non-standard considerations. For single asset liquidation and finite end-costs (hence no singularity
of the value function at time $T$), the solution of the linear-quadratic control problem is well-known,
see, e.g., \cite{Oksendal2008}. Multi-dimensional linear-quadratic control with jumps is treated in the book by~\cite{Hanson2007}. However, our setting (even without the singularity of the value function) is not covered therein.  
The difficulty in our setting stems from the \emph{combination} of the singularity of the value function
and the fact that we consider multi-dimensional portfolios; thus the solution of the
optimization problem involves the detailed
analysis of a specific non-Riccati-type matrix differential equation, for which
we establish existence results and upper and lower bounds of the solution by means of
a novel matrix inequality.
We are aware of only few other texts dealing with singular boundary constraints in liquidation problems. \cite{Schied2010} study optimal liquidation without dark pools
for CARA investors. In this case the optimal control problem does not include jumps. They carry out a verification argument for a candidate value function given in closed form. 
\cite{Naujokat2010} and \cite{hoschler2010} treat similar control problems with jumps. The focus of both texts is on trading with limit orders rather than with dark pools; they only treat single asset trading and obtain the single asset case of this paper as special cases of their respective settings. \cite{Kratz2012} generalizes the single asset case of this paper by allowing for adverse selection; this renders the liquidation problem non-linear quadratic; he shows that in this case the value function is a ``quasi-polynomial'' of degree two instead of being quadratic; the solution for the multi-dimensional case it not known yet.

The remainder of this paper is structured as follows. We specify the model setup and the optimization problem in Section~\ref{ChaptPortfolioSecModel} and obtain a candidate for the value function of the unconstrained optimization problem with finite end-costs via an initial value problem for a matrix differential equation. In Section~\ref{SecMain}, we state the main theoretical results of the article: the solution of the initial value problem (Section~\ref{SubSecMainIVP}), the solution of the unconstrained optimization problem (Section~\ref{SubSecMainOPTl}) and finally the solution of the constrained optimization problem (Section~\ref{SubSecMainOPT}). We discuss the properties of the value function of the optimal strategy extensively for a single asset position and a portfolio consisting of two assets in Section~\ref{SecPortfolioProp}. The main results of Section~\ref{SecMain} are proven in Section~\ref{SecProofs} and the results of Section~\ref{SecPortfolioProp} are proven in Section~\ref{SecAppendixProp}. 

\section{Model description} \label{ChaptPortfolioSecModel}

For a fixed time interval $[0,T]$, we consider the stochastic basis $(\Omega, \mathcal{F}, \mathbb{P}, \mathbb{F}=(\mathcal{F}_t)_{t \in[0,T]})$\footnote{The filtration is generated by the involved random processes and is specified at the end of Section~\ref{SubSecModelDarkPool}}. We investigate a market model where a risk averse trader with a personal risk aversion parameter $\alpha \geq 0$ has to liquidate a portfolio $x \in \mathds{R}^n$ of $n$ assets within a finite trading horizon $[0,T]$. The investor has the possibility to trade simultaneously at a traditional exchange and in a dark pool, facing the trade-off of paying market impact costs in the traditional exchange against uncertain execution in the dark pool. 

We specify the price dynamics at the primary exchange in Section~\ref{SubSecModelPrimaryEx} and the dynamics of order execution in the dark pool in Section~\ref{SubSecModelDarkPool}. In Section~\ref{SubSecModelStrategies} we define admissibility of trading strategies. In Section~\ref{SubSecModelCosts}, we specify the trading objective and show that the resulting optimization problem is equivalent to a linear-quadratic optimization problem. Heuristic arguments suggest that the value function of the optimization problem is singular at terminal time $T$ because of the liquidation constraint. Hence, we introduce a modified optimization problem where we drop the liquidation constraint and approximate it by finite end-costs for a portfolio not liquidated by time $T$ as an intermediate step. In Section~\ref{SubSecModelHJB} we derive a candidate for the value function of the modified optimization problem via a quadratic ansatz and the corresponding HJB Equation.

\subsection{The primary exchange} \label{SubSecModelPrimaryEx}

In absence of transactions of the investor, the fundamental asset price at the primary exchange is given by an $n$-dimensional stochastic process $\tilde{P}$.

\begin{assumption} \label{AssPrimary}
We assume that $\tilde{P}$ satisfies the following properties.
\begin{enumerate}
\item[(i)]
$\tilde{P}$ is a square-integrable c\`{a}dl\`{a}g martingale.
\item[(ii)]
The 
covariance matrix of $\tilde{P}$ is constant in time, i.e., for all $t \in[0,T]$, $i,j=1,\dots,n$,
\[
\Cov\big( \tilde{P}_i(t), \tilde{P}_j(t)\big) =t \Sigma_{i,j}
\]
and  $\Sigma=(\Sigma_{i,j})_{i,j=1,\dots,n} \in \mathds{R}^{n\times n}$.
\end{enumerate}
\end{assumption}

As the time-horizon for portfolio liquidation is usually short, i.e., several hours or a few days, the martingale
property in Assumption 2.1 (i) does not constitute a major divergence from reality.\footnote{We do not model the period in which the investor has held the assets before time zero. For this period, the martingale assumption is naturally a worse approximation of reality.} Once the trader becomes active on the primary exchange, she influences the market price $P$. We assume that the trader at the primary exchange can only execute trades continuously, i.e.~that her trading activity on the primary exchange is absolutely continuous and can hence be described by her trading intensity $\xi(t)$ with $t \in [0,T)$.
This trading in the traditional exchange generates price impact, which we assume to be temporary and linear in the trading rate $\xi(t)$. Given a strategy $(\xi(t))_{t\in [0,T]}$, the transaction price at time $t \in [0,T]$ is given by 
\[
P(t) = \tilde{P}(t) - \Lambda \xi(t), 
\]
where $\Lambda \in \mathds{R}^{n \times n}$ is a positive definite matrix constant in time. Execution of the trades at the primary exchange is \emph{certain}; we hence consider only market orders and no limit orders.

By assuming linear price impact for the primary venue, we follow~\cite{Almgren2001}. 
This choice yields a tractable model which nevertheless captures price impact effects. Linear price impact models have become the basis of several theoretical studies, e.g.,  \cite{Almgren2007}, \cite{Carlin2007}, \cite{Schoneborn2007} and~\cite{Rogers2010}. Furthermore, they demonstrated reasonable properties in real world applications and serve as the basis of many optimal execution algorithms run by practitioners (see e.g., \cite{Kissell2003}; \cite{Schack2004}; \cite{Abramowitz2006}; \cite{Leinweber2007}). In contrast to~\cite{Almgren2001}, the price impact in our model is purely temporary; \cite{Kloeck2011} analyze the influence of \emph{permanent} price impact on the existence of market manipulation strategies. Such an analysis is not in the scope of our article; hence we allow only for temporary price impact.\footnote{As the price is not influenced permanently by the investor, the term \emph{price impact} might be misleading. Alternatively, we could rename the quadratic costs caused by trading in the primary venue as the \emph{execution costs} of the investor; indeed, the costs of trading can be interpreted to include such different effects as price impact and transaction costs. However, the term \emph{execution costs} does not capture the possible impact of trading in one asset on the price of another asset which we include by allowing the off-diagonal elements of $\Lambda$ to be non-zero (cf.~the discussion of cross price impact in Section~\ref{SubSecPortfolioLambda}).}

\subsection{Order transaction in the dark pool} \label{SubSecModelDarkPool}

In addition to the primary exchange, the trader can also use a dark pool.
Dark pools often have rather complex order allocation mechanisms; most of them use some sort of a pro-rata or time-priority rule for matching orders from opposite sides of the market. Here, we consider a dark pool with a time-priority matching rule: the investor's order $\eta$ enters a queue and is matched with liquidity from the opposite side of the market (if there is any) once it has reached the front of the queue.

We allow for continuous updating of the orders $\eta(t)$ in the dark pool at any time $t \in [0,T]$. Orders for the $i^{\text{th}}$ asset in the dark pool are executed fully at the jump times of the $i^{\text{th}}$ component of an $n$-dimensional Poisson process
\[
\pi=(\pi_1,\dots,\pi_n) \text{ with intensities  } 
\theta_1,\dots,\theta_n \geq 0, \text{ respectively.}
\] 
Else, the orders are not executed at all. This mechanism implies two main  simplifications of reality which allow a thorough mathematical analysis of the model. First, we exclude partial execution; the probability of execution does not depend on the size of the order. Second, we assume independence of the increments of the dark pool liquidity. On the other hand, the resulting model captures the stylized facts outlined above and in the introduction; we believe that it constitutes a sufficiently well approximation of reality for our purposes.

\begin{assumption} \label{AssDP}
We assume that $\pi$ satisfies the following conditions.
\begin{enumerate}
\item[(i)]
$\pi_1,\dots,\pi_n$ are independent.
\item[(ii)]
$\pi$ and $\tilde{P}$ are independent.
\end{enumerate}
\end{assumption}

By Assumption~\ref{AssDP}~(ii) we exclude correlations between dark pool liquidity and the fundamental asset price in the market. In particular, this rules out adverse selection\footnote{Here, adverse selection refers to the phenomenon that liquidity seeking traders find that their trades in the
dark pool are usually executed just before a favorable price move, i.e., exactly when they \emph{do not} want them
to be executed since they miss out on the price improvement.}.
For single asset liquidation, adverse selection was incorporated by~\cite{Kratz2012}. This renders the liquidation problem non-linear quadratic; he shows that in this case the value function is a ``quasi-polynomial'' of degree two instead of being quadratic. A solution for the multi-dimensional case it not known yet.

While the dark pool has no impact on prices at the primary venue, it is less clear to which extent the price
impact of the primary venue $\Lambda \xi(t)$ is reflected in the trade price of the dark pool. If for example the price impact is realized predominantly in the form of a widening spread, then the impact on dark pools that monitor the
mid quote can be much smaller than $\Lambda \xi(t)$. We will make the simplifying assumption that trades in the dark pool are not influenced by the price impact at all, i.e., that they are executed at the fundamental
price $\tilde{P}$.
If alternatively the transaction price in the dark pool is the price $P$ at the primary exchange including the trader's price impact, market manipulation strategies can become profitable unless the parameters are chosen with great care, as has been shown by~\cite{Kratz2010} for the discrete-time case. For a detailed discussion see also~\cite{Kloeck2011} who analyze the circumstances that can lead to price manipulation in dark pools.

We are now ready to specify the filtration $(\mathcal{F}_t)_t$ as the completion of $\big(\sigma\big( \tilde{P}(s), \pi(r) | 0 \leq s \leq t, 0 \leq r < t \big)\big)_t.$

\subsection{Admissible trading strategies} \label{SubSecModelStrategies}

Let $t\in[0,T)$ be a given point in time and $x \in \mathds{R}^n$ be the portfolio position of the trader at time $t$.
The trader has the possibility to trade asset $k$ in the traditional exchange with trading intensity $\xi_k(s)$ 
at time $s \in [t,T)$ and to place orders $\eta_k(s)$ in the dark pool at time $s$. 

We call a $2 n$~-~dimensional stochastic process 
\[
(u(s))_{s \in [t,T)}=(\xi(s),\eta(s))_{s \in [t,T)}
\]
a \textit{trading strategy} if $\xi$ is progressively measurable and $\eta$ is predictable. 
Given a trading strategy $u$, the portfolio position at time $s \in [t,T)$ is given by the following 
controlled stochastic differential equation:
\begin{equation} \label{EqCSDE}
\begin{split}
dX^u(s)&=-\xi(s) ds - \eta(s) d\pi(s) \\
X^u(t)&=x
\end{split}
\end{equation}
such that the left hand side in~\eqref{EqCSDE}
is well-defined.

For technical reasons we require all trading strategies to fulfill
the following conditions.

\begin{definition}  \label{DefAdmStr}
Let $t \in [0,T)$ and $x \in \mathds{R}^n$ be fixed. Let
$
u=(u(s))_{s \in [t,T)}=((\xi(s),\eta(s)))_{s \in [0,T)}
$
be a trading strategy, i.e.,
$\xi$ is progressively measurable and $\eta$ is predictable. \\
(a) We call $u$
an \emph{admissible trading strategy }if it fulfills the following conditions.
\begin{enumerate}
\item[(i)]
The Stochastic Differential Equation~\eqref{EqCSDE} possesses a unique solution on $[t,T)$.
\item[(ii)]
\begin{equation*}
\mathbb{E}  \Big[\int_t^T\ \|\xi(s)\|_2^4 ds \Big] < \infty, \quad
\mathbb{E} \Big[\int_t^T\ \|\eta(s)\|_2^8 ds \Big]< \infty. 
\end{equation*}
\item[(iii)]
If $\theta_i=0$, then $\eta_i(s)=0$ for all $s\in [t,T)$.
\end{enumerate}
We denote the set of admissible trading strategies by $\tilde{\mathbb{A}}(t)$. \\

\noindent (b) We call $u\in \tilde{\mathbb{A}}(t)$ an
\emph{admissible liquidation strategy} or just \emph{liquidation strategy} if additionally
\begin{enumerate}
\item[(iv)]
$\lim_{s\rightarrow T-} X^u(s)=  $ a.s.
\end{enumerate}
and denote the 
set of admissible liquidation strategies by $\mathbb{A}(t,x)$.
\end{definition}

Let us shortly comment on Definition~\ref{DefAdmStr}.
Condition~(ii) is required for the moment bounds in Lemma~\ref{LemXuintegrierbar} which is in turn needed for the verification later. Condition~(iii) is needed in order to ensure uniqueness of optimal trading strategies: if $\theta_i=0$, no additional gain can be achieved by non-zero dark pool orders. If the portfolio is liquidated at constant speed at the primary exchange only, i.e., 
$\xi(s)=\frac{x}{T-t}$, $\eta(s)=0$ for $s \in [t,T)$,
Definition~\ref{DefAdmStr} (in particular the liquidation constraint~(iv)) is satisfied and hence $\mathbb{A}(t,x)\not= \emptyset$.

\begin{remark} \label{RemarkInterJump}
We expect that the stochastic control problems we solve are such that the optimal control is of 
\emph{Markovian form} (see, e.g., the book by \cite{Oksendal2007}, Theorem 11.2.3):
\[
u(s)=(\xi(s),\eta(s))=(\xi(s,X(s)),\eta(s,X(s-))
\]
for deterministic functions $\xi,\eta: [t,T) \times \mathds{R}^n \rightarrow \mathds{R}^n$.
The deterministic initial value problem
\begin{equation}
 \label{EqPortfolioInitialSimple}
X' =-\xi(\cdot,X), \quad
X(t)=x 
\end{equation}
possesses a unique solution on $[t,T)$ if 
$\|\xi(s,y)\| \leq f(s) \|y\| + g(s)$ on $[t,T) \times \mathds{R}^n$
for $f,g \in C([t,T))$ 
and $\xi(s,\cdot)$ is locally Lipschitz (e.g., $C^1$).\footnote{This follows, e.g., by Peano's existence theorem and Gronwall's inequality.} Let $\xi:[t,T) \times \mathds{R}^n \rightarrow \mathds{R}^n$ fulfill these conditions and let $\eta:[t,T) \times \mathds{R}^n \rightarrow \mathds{R}^n$. We can pathwise construct the
solution of the Stochastic Differential Equation~\eqref{EqCSDE} inductively by interlacing the jumps (see, 
e.g.,~\cite{Applebaum2004}, Example 1.3.13): as the $n$ Poisson processes
are independent, they
jump at distinct times almost surely. 
Let $(\tau_i)_{i\geq 1}$ be the jump times of $\pi$ such that
$t=:\tau_0 < \tau_1 < \dots$ almost surely,
and let $X$ be the solution of the Initial Value Problem~\eqref{EqPortfolioInitialSimple}
on $[\tau_i, \tau_{i+1} \wedge T)$ with initial value $x=X(\tau_i)$
for $i \in \mathds{N}$ such that $\tau_i \leq T$. For $\tau_{i+1} \leq T$ 
and $\Delta \pi_k(\tau_{i+1})>0$, we set
\[
X(\tau_{i+1}):= X(\tau_{i+1}-) - \eta_k (\tau_{i+1},X(\tau_{i+1}-)) e_k,
\]
where $e_k$ is the $k^\text{th}$ unit vector.
\end{remark}

\subsection{Cost functional} \label{SubSecModelCosts}

The proceeds of selling the portfolio $x \in \mathds{R}^n$ during $[t,T]$ according to the strategy
$(u(s))_s=(\xi(s),\eta(s))_s \in \mathbb{A}(t,x)$ are given by
\[
\phi(t,x,u) 
	 := \int_t^T \xi(s)^\top(\tilde{P}(s) - \Lambda \xi(s))ds + \int_t^T \eta(s)^\top \tilde{P}(s)  d\pi(s).
\]
The first term in the above equation represents the proceeds of selling at the primary exchange at a price of $P(t) = \tilde{P}(s) - \Lambda \xi(s)$, while the second term accounts for the proceeds of selling in the dark pool at the unaffected price $\tilde{P}(s)$. Applying integration by parts and using $X^u(s) = x - \int_t^s \xi(r)dr -\int_t^s \eta(r)d\pi(r)$, Assumption~\ref{AssPrimary}~(i), the fact that $\tilde{P}$ and $\pi$ are independent (Assumption~\ref{AssDP}~(ii)) and the liquidation constraint (Definition~\ref{DefAdmStr}~(iv)), we obtain 
\begin{align} \notag
\phi(t,x,u) &= -\int_t^T \xi(s)^\top \Lambda \xi(s)ds + x^\top \tilde{P}(t) + \int_t^T X^u(s-) d\tilde{P}(s).
\end{align}
This yields (cf.~Assumption~\ref{AssPrimary})
\[
\mathds{E} \big[\phi(t,x,u)\big] = x^\top \tilde{P}(t) -\mathds{E} \Big[
\int_t^T \xi(s)^\top\Lambda \xi(s)ds\Big].
\]
Instead of maximizing expected proceeds, we can thus equivalently minimize expected price impact costs. We assume that the trader is not only interested in expected liquidation proceeds, but in addition also wants to minimize risk during liquidation. We incorporate both aspects in the following cost functional:\footnote{Both mean and variance of execution costs are often used as measures of execution quality. The cost functional $J$ is inspired by such mean variance measures. An alternative approach is the maximization of expected utility: $\tilde{J}(t,x,u) = E[U(\phi(t,x,u))]$ for some utility function $U$. In this alternative set-up analytical solutions are unfortunately not directly available through the methods presented in this paper and are hence left for future research.}
\begin{equation} \notag
J(t,x,u):=x^\top \tilde{P}(t) - \mathds{E} \big[\phi(t,x,u)\big] + \mathds{E}\Big[ \alpha\int_t^T X^u(s)^\top \Sigma X^u(s)ds \Big] =  \mathbb{E} \Big[ \int_t^T f(\xi(s),X^u(s)) ds \Big],
\end{equation}
where
$f:\mathds{R}^{n } \times \mathds{R}^{n} \rightarrow \mathds{R}$
is given by
$
f(\xi ,x) :=  \xi^\top\Lambda \xi  + \alpha x^\top\Sigma x
$.
The first two terms in the cost functional capture the expected liquidation shortfall, while the last term is an additive penalty function $\alpha\int_t^T X^u(s)^\top \Sigma X^u(s)ds$ which reflects the market risk of the portfolio (recall that $\alpha\geq 0$ is the risk aversion parameter of the investor); it penalizes slow liquidation and poorly balanced portfolios. It does not incorporate execution risk; hence the investor is only risk averse with respect to market risk but not with respect to execution risk. In discrete time, \cite{Kratz2010} argue that for realistic parameters, market risk outweighs execution risk;\footnote{The illustration and argument provided in \cite{Kratz2010} applies for the continuous-time case in the same way as it applies to the discrete-time setting in which it is presented.} we hence expect that the inclusion of risk aversion with respect to execution risk would not change the optimal strategy in essence. For deterministic liquidation strategies without dark pools, the risk term reflects the variance of the liquidation costs (see~\cite{Almgren2001}). In this case, minimizing a mean-variance functional of the liquidation costs over all deterministic strategies is equivalent to maximizing the expected utility of the proceeds of an investor with CARA preferences over all strategies (see~\cite{Schied2010}). 

We assume that the trader aims to minimize the cost functional 
and hence considers the following optimization problem:\footnote{An alternative interesting set-up is to consider trade execution under minimum proceeds constraints. Such problems have recently been addressed using theory about stochastic target problems (see, e.g., \cite{Bouchard2009} and~\cite{Bouchard2012}).}
\begin{equation} \label{EqValueFct}
v(t,x):=\inf\limits_{u \in \mathbb{A}(t,x)}  J(t,x,u). \tag{OPT}
\end{equation}
Note that the optimization problem is well-defined and the value function satisfies
$v(t,x) < \infty$ for $t<T$
(consider, e.g., constant liquidation exclusively in the primary exchange).
Because of the liquidation constraint (cf.~Definition~\ref{DefAdmStr}~(iv)), we expect the value function to fulfill
\begin{equation} \notag
\lim\limits_{s \rightarrow T-} v(s,x) =
\begin{cases}
0 &\text{if } x= 0 \\
\infty &\text{else}, 
\end{cases}
\end{equation} 
i.e., $v$ has a singularity at the terminal time $T$. Because of this singularity,
non-standard considerations are necessary for solving the Optimization 
Problem~\eqref{EqValueFct} via a verification argument using
the HJB equation.

As an intermediate step, we hence weaken the liquidation constraint by
allowing for all strategies $u \in \tilde{\mathbb{A}}(t)$
and by penalizing non-liquidation by finite end-costs. More precisely,
for $l >0$ and $u=(\xi,\eta) \in \tilde{\mathbb{A}}(t)$, we define the following cost functional 
\begin{equation} \notag
\tilde{J}(l,t,x,u):=\mathbb{E} \Big[ \int_t^T f(\xi(s),X^u(s)) ds + l \cdot X^u(T)^\top I X^u(T) \Big].
\end{equation} 
The resulting optimization problem is
\begin{equation} \label{EqValueFctl}
\tilde{v}(l,t,x):=\inf\limits_{u \in \tilde{\mathbb{A}}(t)}  \tilde{J}(l,t,x,u). \tag{$\widetilde{\text{OPT}}$}
\end{equation}
The Optimization Problem~\eqref{EqValueFctl} mainly serves as an approximation of the optimization Problem~\eqref{EqValueFct}. However, it is also of interest itself: the penalization term $l \cdot X^u(T)^\top I X^u(T)$ can be considered as the liquidation cost of the left-over position. Note that the identity matrix in the term can be replaced by any positive definite matrix reflecting this interpretation (e.g., $\Lambda$) without changing any of the proofs significantly.

In the following, we solve the unconstrained Optimization Problem~\eqref{EqValueFctl} first (Section~\ref{SubSecMainOPTl}).
Then, we show that the solution of the Optimization Problem~\eqref{EqValueFctl} converges to
the solution of the original constrained Optimization Problem~\eqref{EqValueFct} as $l \rightarrow \infty$
(Section~\ref{SubSecMainOPT}). 

\subsection{Hamilton-Jacobi-Bellman equation} \label{SubSecModelHJB}

In this section we derive a candidate for the value function of the
Optimization Problem~\eqref{EqValueFctl}. Heuristic considerations suggest 
that it should satisfy the following HJB equation
(see, e.g., \cite{OksendalSulem}):
\begin{equation} \tag{$\widetilde{\text{HJB}}$} \label{EqHJB}
\begin{split}
\frac{\partial w}{\partial t} (t,x)& = \!\! \sup\limits_{u=(\xi,\eta) \in \mathds{R}^n\times\mathds{R}^n}
	\! \Big[ \sum\limits_{i=1}^n \theta_i \big( w(t,x)-w(t, x - \eta^\top e_i) \big) 
	\!+ \!\nabla_x w(t,x) \xi \! - \! f(\xi ,x)\Big] \\
w(T,x)&=lx^\top  x.	
\end{split}
\end{equation}
The linear-quadratic form of the cost functional suggests that the value function is quadratic. Assuming that the above guesses are correct, the following proposition provides candidates both for the value function and for the optimal strategy.

\begin{prop} \label{PropSolHJB}
Let $l>0$ and assume that the initial value problem for a matrix differential equation
\begin{equation*} 
C'=C^\top\Lambda^{-1}C+C^\top\tilde{C}C-\alpha \Sigma, \quad
C(T)=l I,
\end{equation*}
where
$
\tilde{C}(l,t):=\diag\Big(\tfrac{\theta_i}{c_{i,i}(l,t)}\Big)
$
possesses a positive definite solution $C(l,t)=(c_{i,j}(l,t))_{i,j=1,\dots,n}$ on $[0,T]$. Then
\begin{equation} \notag
w(l,t,x):= x^\top C(l,t) x
\end{equation}
satisfies the HJB Equation~\eqref{EqHJB}
with maximizer $u^*=(\xi^*,\eta^*)$ for
\[
\xi^*:=\xi^*(l,t,x):=\Lambda^{-1}C(l,t)x, \quad
\eta^*:=\eta^*(l,t,x):=\bar{C}(l,t)C(l,t)x, 
\]
where
$
\bar{C}(l,t):=\diag\Big(\tfrac{1}{c_{i,i}(l,t)}\Big)
$.

If there exist $i_1,\dots, i_k \in \{1,\dots, n\}$ 
such that $\theta_{i_j}=0$ ($j=1,\dots, k$), then $\eta_{i_j}$ can be chosen arbitrarily. Up to arbitrary choices of $\eta_{i_j}$, the maximizer is unique.
\end{prop}

\begin{proof}
The assertion follows directly from plugging the quadratic ansatz $w(l,t,x)= x^\top C(l,t) x$ into the HJB Equation~\eqref{EqHJB}; the resulting function can be maximized by standard calculus.
\end{proof}

\section{Main results} \label{SecMain}

Proposition~\ref{PropSolHJB} suggests that the solution of the Optimization Problem~\eqref{EqValueFctl} solves the initial value problem for the matrix differential equation
\begin{equation} \label{EqODE1}
\begin{split}
C'&=C^\top\Lambda^{-1}C+C^\top\tilde{C}C-\alpha \Sigma \\
C(T)&=l I,
\end{split}
\end{equation}
where 
\begin{equation} \label{EqTildeC}
\tilde{C}:=\diag\Big(\frac{\theta_i}{c_{i,i}}\Big).
\end{equation}
In the remainder of the section, we state the main results of the article. In Section~\ref{SubSecMainIVP}, we show that~\eqref{EqODE1} admits a unique solution $C$ on $[0,T]$ and establish appropriate upper and lower bounds for $C$. Subsequently, we deduce the solution of the Optimization Problem~\eqref{EqValueFctl} in Section~\ref{SubSecMainOPTl} and as a limit of this (as $l \rightarrow \infty$) the solution of the Optimization Problem~\eqref{EqValueFct} in Section~\ref{SubSecMainOPT}. Proofs of these results are presented in Section~\ref{SecProofs}.

Before we proceed, we introduce the following notations.
\begin{nota} \label{NotaMatrix}
\begin{enumerate}
\item[(i)]
For symmetric matrices $A$ and $B$ we say 
$A > B$ ($A\geq B$) if $A-B$ is positive (nonnegative) definite. 
\item[(ii)]
We denote the smallest and the largest eigenvalues of a real-symmetric matrix $A$ by $a_{\min}$ and $a_{\max}$, respectively. 
\item[(iii)]
We define the following nonnegative definite matrix: 
$D:=\sqrt{\Lambda^{-1}} \Sigma \sqrt{\Lambda^{-1}}$.
\end{enumerate}
\end{nota}

\subsection{Solution of the Initial Value Problem~\eqref{EqODE1}} \label{SubSecMainIVP}

It is not immediately clear that the Initial Value Problem~\eqref{EqODE1} possesses a 
positive definite solution on the whole interval $[0,T]$ for $n\geq 2$. For $n=1$, it reduces to
$
C'=\tfrac{C^2}{\Lambda} + \theta_1 C - \alpha \Sigma$, $C(T)= l$.
This is an initial value problem for a scalar Riccati differential equation
with constant coefficients, whose unique solution 
is explicitly known and exists on the whole interval $[0,T]$ (cf.~Section~\ref{SecContPropOne}). For $n\geq 2$, the following theorem establishes the existence and uniqueness of the solution of~\eqref{EqODE1}.

\begin{theorem} \label{TheoremDiffBounds}
Let $\theta_i \geq 0$ for $i=1,\dots ,n$, $\theta=\sum\limits_{i=1}^n \theta_i$ and $l>l_0$,
where
\begin{equation} \label{lNull}
l_0:=\max\Big\{\lambda_{\max} \Big( \sqrt{\tfrac{\theta^2}{4}+\alpha d_{\min}} - \tfrac{\theta}{2} \Big),
	\lambda_{\min} \big(\sqrt{\alpha d_{\max}}\big) \Big\}
\end{equation}
(cf.~Notation~\ref{NotaMatrix}). 
Then the Initial Value Problem~\eqref{EqODE1} possesses a unique solution $C(l,\cdot)$ on $(-\infty,T]$. The solution is symmetric for all 
$t \in (-\infty,T]$ and
\begin{equation*} 
0<P(l,t) \leq \sqrt{\Lambda^{-1}}C(l,t)\sqrt{\Lambda^{-1}} \leq Q(l,t),
\end{equation*}
where $P$ and $Q$ are the solutions of the initial value problems
\begin{equation} \label{IVPPQ}
P'=P^2+ \theta P-\alpha d_{\min}I, \quad
P(T)=\frac{l}{\lambda_{\max}}I  \quad \text{ and } \quad
Q'=Q^2 - \alpha d_{\max} I, \quad  
Q(T)=\frac{l}{\lambda_{\min}}I,
\end{equation}
respectively.
\end{theorem}

\begin{remark}
\begin{enumerate}
\item[(i)]
The solutions of the initial value problems for Riccati matrix differential equations in~\eqref{IVPPQ} exist on the whole interval $[0,T]$ and can be computed in closed form (cf.~Equations~\eqref{PQlt} - \eqref{pqlttheta0}). 
For technical reasons, we prefer to establish bounds for $\sqrt{\Lambda^{-1}} C \sqrt{\Lambda^{-1}}$ instead of bounds for $C$. 
$P$ and $Q$ are constructed in terms of multiples of the identity matrix and hence commute with all matrices. Therefore, they transfer to bounds of $C$ directly by multiplying them with $\Lambda$. 
\item[(ii)]
The bounds of $C$ are an essential component for the proof of Theorem~\ref{TheoremDiffBounds}. Additionally, they are required for all key steps of the solution of the Optimization Problem~\eqref{EqValueFctl} (Proposition~\ref{PropAdml} and Theorem~\ref{TheoremOptStrl}) and of the solution of the Optimization Problem~\eqref{EqValueFct} (Theorem~\ref{TheoremLimitValueFunc}, Theorem~\ref{TheoremAdm} and Theorem~\ref{TheoremOptLiq}).
\end{enumerate}
\end{remark}

\subsection{Solution of the Optimization Problem~\eqref{EqValueFctl}} \label{SubSecMainOPTl}

Combining Proposition~\ref{PropSolHJB} and Theorem~\ref{TheoremDiffBounds}, we obtain well-defined candidates both for the value function ($x^\top C(l,t) x$) and for the optimal strategy $u^*=(\xi^*,\eta^*)$ of the Optimization Problem~\eqref{EqValueFctl}. The latter is given by
\begin{align}
\xi^*(l):=\xi^*(l,t,x)&:=\Lambda^{-1}C(l,t)x,  \label{OptStrlxi}\\
\eta^*(l):=\eta^*(l,t,x)&:=\tilde{I} \bar{C}(l,t)C(l,t)x, \label{OptStrleta}
\end{align}
where
$\tilde{I}=(e_{i,j})_{i,j=1,\dots,n}$ is the diagonal matrix with
\begin{equation*}
e_{i,i}=
\begin{cases}
1 &\text{ if } \theta_i>0 \\
0 &\text{ else}
\end{cases}
\qquad \text{and} \qquad
\bar{C}(l,t):=\diag\Big(\frac{1}{c_{i,i}(l,t)}\Big).
\end{equation*}
The following Proposition confirms that $u^*$ is admissible. 

\begin{prop} \label{PropAdml}
Let $l>l_0$ for $l_0$ as in Equation~\eqref{lNull} and $(t,x) \in [0,T) \times \mathds{R}^n$. Then $u^*(l)=(\xi^*(l),\eta^*(l)) \in \tilde{\mathbb{A}}(t)$, where $\xi^*(l)$ and $\eta^*(l)$ are as in Equations~\eqref{OptStrlxi} and \eqref{OptStrleta}, respectively.
\end{prop}

Finally, we obtain the solution of the Optimization Problem~\eqref{EqValueFctl}.

\begin{theorem} \label{TheoremOptStrl}
Let $l \geq l_0$ for $l_0$ as in Equation~\eqref{lNull} and let $C(l,t)$ be the unique 
solution of the Initial Value Problem~\eqref{EqODE1}.
Then the value function of the Optimization Problem~\eqref{EqValueFctl} is given by
\begin{equation} \notag
\tilde{v}(l,t,x)=x^\top C(l,t) x
\end{equation}
and the $\mathbb{P}\otimes \lambda$ - almost surely unique optimal strategy is given by 
$u^*(l)$ as in Equations~\eqref{OptStrlxi} and~\eqref{OptStrleta}.
\end{theorem}

\subsection{Solution of the Optimization Problem~\eqref{EqValueFct}} \label{SubSecMainOPT}

Intuitively, infinite end-costs should force the controlled process $X^u(s)$ to approach zero as $s \rightarrow T-$; furthermore the solution of the Optimization Problem~\eqref{EqValueFct} should be the limit of the solution of
the Optimization Problem~\eqref{EqValueFctl} as $l \rightarrow \infty$. The following theorem confirms (in particular) that this limit is well-defined. 

\begin{theorem} \label{TheoremLimitValueFunc}
Let $t \in [0,T)$.
\begin{enumerate}
\item[(i)]
The element-wise limit of the value function matrix
\begin{equation} \notag
C(t):=\lim\limits_{l\rightarrow \infty }C(l,t)
\end{equation}
exists on $[0,T)$, and $C(l,\cdot)$ converges compactly to $C$ on $[0,T)$. Furthermore,
$\lim_{l\rightarrow \infty }c_{\min}(l,T)=\infty.$
\item[(ii)]
$C$ solves the matrix differential equation
\begin{equation} \label{EqDiffEqLimit}
C'=C^\top\Lambda^{-1}C+C^\top\tilde{C}C-\alpha \Sigma 
\end{equation}
on $[0,T)$ with boundary condition $\lim_{s \rightarrow T-} c_{\min}(s) =\infty$. Moreover, the following inequalities hold.
\begin{equation} \label{EqBoundsLimitC}
0<P(t) \leq \sqrt{\Lambda^{-1}} C(t) \sqrt{\Lambda^{-1}} \leq Q(t),
\end{equation}
where $P(t):=\lim_{l\rightarrow \infty} P(l,t)$ and $Q(t):=\lim_{l\rightarrow \infty} Q(l,t)$.
\end{enumerate}
\end{theorem}

\begin{remark} \label{RemarkPrincipalSolution}
For Riccati matrix differential equations, there exists a unique solution $F$ with
$\lim_{s \rightarrow T-} f_{\min}(s) = \infty.$
This solution is called the \textnormal{principal solution} (see, e.g., \cite{Coppel1971}). In this
spirit, $C$ is the principal solution of the Matrix Differential Equation~\eqref{EqDiffEqLimit}. Note
however that it is not entirely clear that $C$ is the \emph{only} solution of~\eqref{EqDiffEqLimit}
satisfying
$\lim_{s \rightarrow T-} c_{\min}(s) = \infty$
since~\eqref{EqDiffEqLimit} is \emph{not} a Riccati matrix differential equation.
\end{remark}

By Theorem~\ref{TheoremLimitValueFunc}, we also obtain the existence of the limits of 
the optimal strategy:
\begin{equation}
\xi^*:=\xi^*(t,x):=\lim\limits_{l \rightarrow \infty} \xi^*(l,t,x) = \Lambda^{-1} C(t) x, \quad
\eta^*:=\eta^*(t,x):=\lim\limits_{l \rightarrow \infty} \eta^*(l,t,x) = \tilde{I} \bar{C}(t)C(t)x. \label{Eqxieta*}
\end{equation}
It turns out that $u^*:=(\xi^*,\eta^*)$ is an admissible liquidation strategy, in particular that it satisfies the liquidation constraint 
\begin{equation} \notag
\lim\limits_{s \rightarrow T-} X^*(s) =0 \quad \text{for} \quad  X^*(s):=X^{u^*}(s).
\end{equation}

\begin{theorem} \label{TheoremAdm}
Let $t\in [0,T)$, $x \in \mathds{R}^n$ and $u^*=(\xi^*,\eta^*)$ for 
$\xi^*$ and $\eta^*$ as in~\eqref{Eqxieta*}. Then $u^*\in \mathbb{A}(t,x)$.
\end{theorem}

We are now ready to present the main result of this article: the solution of the Optimization Problem~\eqref{EqValueFct}. 

\begin{theorem} \label{TheoremOptLiq}
The value function of the Optimization Problem~\eqref{EqValueFct} is given by
\begin{equation*}
v(t,x)=x^\top C(t) x 
\end{equation*}
for all $t \in [0,T)$, $x \in \mathds{R}^n$ and
\begin{equation*}
\lim\limits_{s\rightarrow T-} v(s,x)=
\begin{cases}
0 &\text{if } x=0 \\
\infty &\text{else.}
\end{cases}
\end{equation*}
The $\mathbb{P}\otimes \mathbb{\lambda}$ - almost surely unique optimal strategy is given by 
$u^*=(\xi^*,\eta^*)$ as in~\eqref{Eqxieta*} .
\end{theorem}

\section{Properties of the value function and the optimal strategy} \label{SecPortfolioProp}

The characterization of the solution of the Optimization Problem~\eqref{EqValueFct} enables us to analyze the properties of the optimal strategy and the value function in detail.\footnote{We limit the analysis to the Optimization Problem~\eqref{EqValueFct}. Most of the properties transfer directly to the Optimization Problem~\eqref{EqValueFctl} with the same or similar proofs.} For single asset liquidation ($n=1$) the Initial Value Problem~\eqref{EqODE1} can be solved in closed form as the differential equation is a scalar Riccati equation with constant coefficients. This allows us to prove monotonicity properties of the value function and the optimal strategy in Section~\ref{SecContPropOne}. In Section~\ref{SubSecContMuliAsset} we discuss the multi asset case by analyzing a portfolio of two assets. Although a closed form solution of the Initial Value Problem~\eqref{EqODE1} is not known in general for $n\geq 2$, it is possible to derive analytical results about the dependence of the value function and the optimal strategy on the model parameters, in particular on the correlation of the assets. We illustrate that overly simple adjustments of existing trading algorithms for optimal liquidation without dark pools can have undesirable properties.
The proofs of the results of this section are presented in Section~\ref{SecAppendixProp}.

\subsection{Single asset liquidation} \label{SecContPropOne}

We let $n=1$ and set $\theta=\theta_1$.
The solution of the Initial Value Problem~\eqref{EqODE1} is given by
\begin{equation} \notag
C(l,t)= \frac{\Lambda \tilde{\theta}}{2} \coth \Big(\frac{\tilde{\theta}}{2} (T-t) +\kappa(l) \Big) 
-\frac{ \Lambda \theta}{2},
\end{equation}
where 
\[
\kappa(l) := \arcoth\Big( \frac{ \frac{2 l}{\Lambda} + \theta}{\tilde{\theta}}\Big), \quad
\tilde{\theta} :=\sqrt{\theta^2+\tfrac{4 \alpha \Sigma }{\Lambda}}
\]
for $\theta >0$ or $\alpha \Sigma >0$ and
$
C(l,t) = \frac{\Lambda}{T-t+\frac{\Lambda}{l}}
$
for $\theta=\alpha \Sigma=0$.
In order to highlight the dependence of the value function on the parameters $\theta$, $\Lambda$ and $\alpha \Sigma$, we define, e.g.,
\begin{equation} 
C(t;\theta):=C(t) = \lim\limits_{l \rightarrow \infty} C(l,t) 
	= \begin{cases}
		\frac{\Lambda\tilde{\theta}}{2} \coth \Big(\frac{\tilde{\theta}}{2} (T-t) \Big) -\frac{\Lambda \theta}{2}
	& \text{if } \theta>0 \text{ or } \alpha \Sigma >0 \\
		\frac{\Lambda}{T-t} & \text{if } \theta = \alpha \Sigma = 0,
		\end{cases} \notag
\end{equation}
in particular
\begin{equation} \notag
C(t;0) =  \begin{cases}
		\sqrt{\alpha \Sigma \Lambda}  \coth \Big(\sqrt{\frac{\alpha \Sigma}{\Lambda}} (T-t) \Big)
			& \text{if }  \alpha \Sigma >0 \\
		\frac{\Lambda}{T-t} & \text{if }  \alpha \Sigma = 0.
		\end{cases}
\end{equation}
We will apply similar notations throughout Section~\ref{SecPortfolioProp} to make the dependence of other model components (such as the optimal strategy) on the respective parameters explicit whenever this clarifies the exposition. For $\theta=0$, we obtain the special case of optimal liquidation without dark pool, see~\cite{Almgren2001} for the discrete-time case and~\cite{Schied2010} for the continuous-time version. 

No transaction costs must be paid in the dark pool; intuitively, the investor should hence try to liquidate as much as possible in the dark pool. Indeed, we have
\[
\eta^*(t,x;\theta) = \bar{C}(t;\theta)C(t;\theta) x =x
\]
for $\theta>0$ by Theorem~\ref{TheoremOptLiq}, i.e., it is optimal to always place the full remaining asset position in the dark pool;\footnote{We want to remark that this property is sensitive to the assumption that $\tilde{P}$ is a martingale (Assumption~\ref{AssPrimary}~(i)). If, e.g., the investor holds a long position in the asset and $\tilde{P}$ has a positive drift, she should be reluctant to sell her entire position too early. Therefore, we expect that a drift changes this property. Similarly, if adverse selection is included into the model, the property $\eta^*(t,x)=x$ does not hold; this was shown by~\cite{Kratz2012}.} note that the execution of the dark pool order immediately stops the trading activity by linearity of the optimal strategy in the position.

In Section~\ref{SubSubSecPropSingleTheta}, we discuss the dependence of the optimal strategy and the value function on $\theta$. Subsequently, we analyze the dependence on the price impact $\Lambda$ and the risk parameter $\alpha \Sigma$ in Section~\ref{SubSubSecPropSingleLambda}.

\subsubsection{Dependence on $\theta$} \label{SubSubSecPropSingleTheta}

We expect it to be optimal to slow down trading in the primary venue initially as the trader hopes to trade cheaper in the dark pool. This intuition is confirmed by Proposition~\ref{PropPropertiesn=1Cont}~(ii) and~(iii) below.
In order to state the property rigorously, we first denote the optimal trading trajectory until execution in the dark pool by $\tilde{X}(\cdot;\theta)$,
i.e., $\tilde{X}$ is
the solution of the linear initial value problem
$X' = - \xi^*(\cdot,X;\theta)$, $X(0)=x$,
where $x$ is the initial asset position at time zero. Then for $t\in [0,T)$,
\begin{equation} \label{EqSingleAssetTrajectory}
\tilde{X}(t;\theta)= x \exp \Big(- \int_0^t \frac{C(s;\theta)}{\Lambda} ds \Big) 
	=\frac{ \sinh \big(\frac{\tilde{\theta}}{2}(T-t)\big) 
	\exp \big(\frac{\theta}{2} t\big) }{ \sinh \big(\frac{\tilde{\theta}}{2}T \big) } x.
\end{equation}
We obtain the following monotonicity properties.
For simplicity of exposition we assume $\alpha \Sigma >0$. Similar results hold for the simpler case $\alpha \Sigma =0$ (cf.~also the right hand graph of Figure~\ref{FigureSingle}). 

\begin{prop} \label{PropPropertiesn=1Cont}
\begin{enumerate}
\item[(i)]
For $t\in [0,T)$, $C(t;\theta)$ is strictly decreasing in $\theta$.
\item[(ii)]
For $x>0$ and $t \in (0,T]$, $\xi^*(t,x;\theta)$
is strictly decreasing in $\theta$.
\item[(iii)]
For $x>0$ and $t \in (0,T)$, $\tilde{X}(t;\theta)$ is strictly increasing in $\theta$.
\item[(iv)]
For $x>0$ and $t \in (0,T)$, the expected asset position if the 
optimal strategy is applied,
$
\mathbb{E} [X^*(t;\theta)],
$
is strictly decreasing in $\theta$.
\item[(v)]
For $\alpha \Sigma>0$, the risk costs
$
\alpha \Sigma \cdot \mathbb{E} \big[\int_0^T X^*(t;\theta)^2 dt \big] 
$
are strictly decreasing in $\theta$.
\item[(vi)]
The impact costs
$
\Lambda \cdot \mathbb{E} \big[\int_0^T \xi^*(t,X^*(t;\theta);\theta)^2 dt \big] 
$
are strictly decreasing in $\theta$.
\end{enumerate}
\end{prop}

Let us shortly comment on Proposition~\ref{PropPropertiesn=1Cont}. The fact that the overall costs are decreasing in $\theta$ (part~(i)) is quite intuitive and can be deduced directly from the definition of the cost functional $J$. Proposition~\ref{PropPropertiesn=1Cont}~(iv) states that the introduction of a dark pool decreases the expected asset position despite slower initial trading in the primary exchange (parts~(ii) and~(iii)). Parts~(v)~and~(vi) of Proposition~\ref{PropPropertiesn=1Cont} confirm that the introduction of the dark pool decreases both the impact costs component and the risk costs component of the value function.\footnote{It is an interesting effect that this is not necessarily the case in the discrete-time setting of~\cite{Kratz2010} where the risk costs can be increasing for small $\theta$. This effect is lost by letting the length of the trading periods tend to zero (see~\cite{Kratz2011}).}

We illustrate these properties in Figure~\ref{FigureSingle}. In the left picture we consider a risk neutral investor and in the right picture a risk averse investor. The optimal trading trajectories for trading with dark pool are displayed by the thick solid lines. In the displayed scenario the dark pool order is executed at time $\tau$. After execution in the dark pool, the liquidation task is finished. The dotted lines denote the scenario where the dark pool order is not executed during the entire trading horizon and the thin solid lines refer to the expected asset positions. We contrast the optimal strategy with dark pool by the optimal strategy without dark pool (dashed lines). As shown in Proposition~\ref{PropPropertiesn=1Cont}, the dark pool slows down trading in the primary venue initially. Nevertheless, the expected position is smaller than the trading trajectory without dark pool.

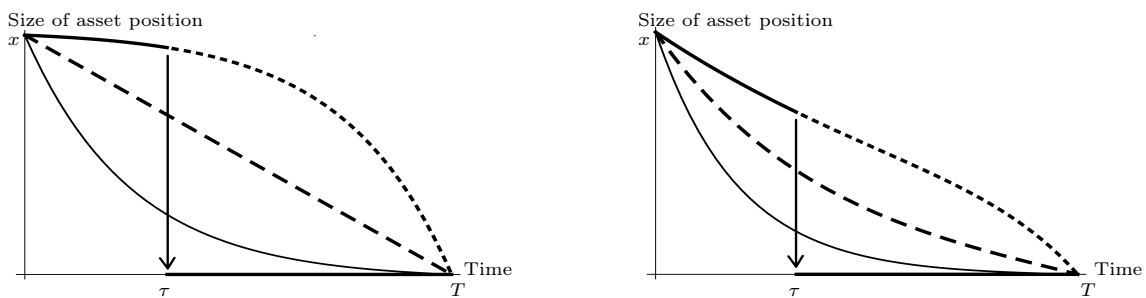
\begin{figure}[ht!]
\vspace{4ex}
\centering
\begin{tabular}{lll}
\begin{tabular}{l}\begin{overpic}[height=3.5cm, width=6cm]{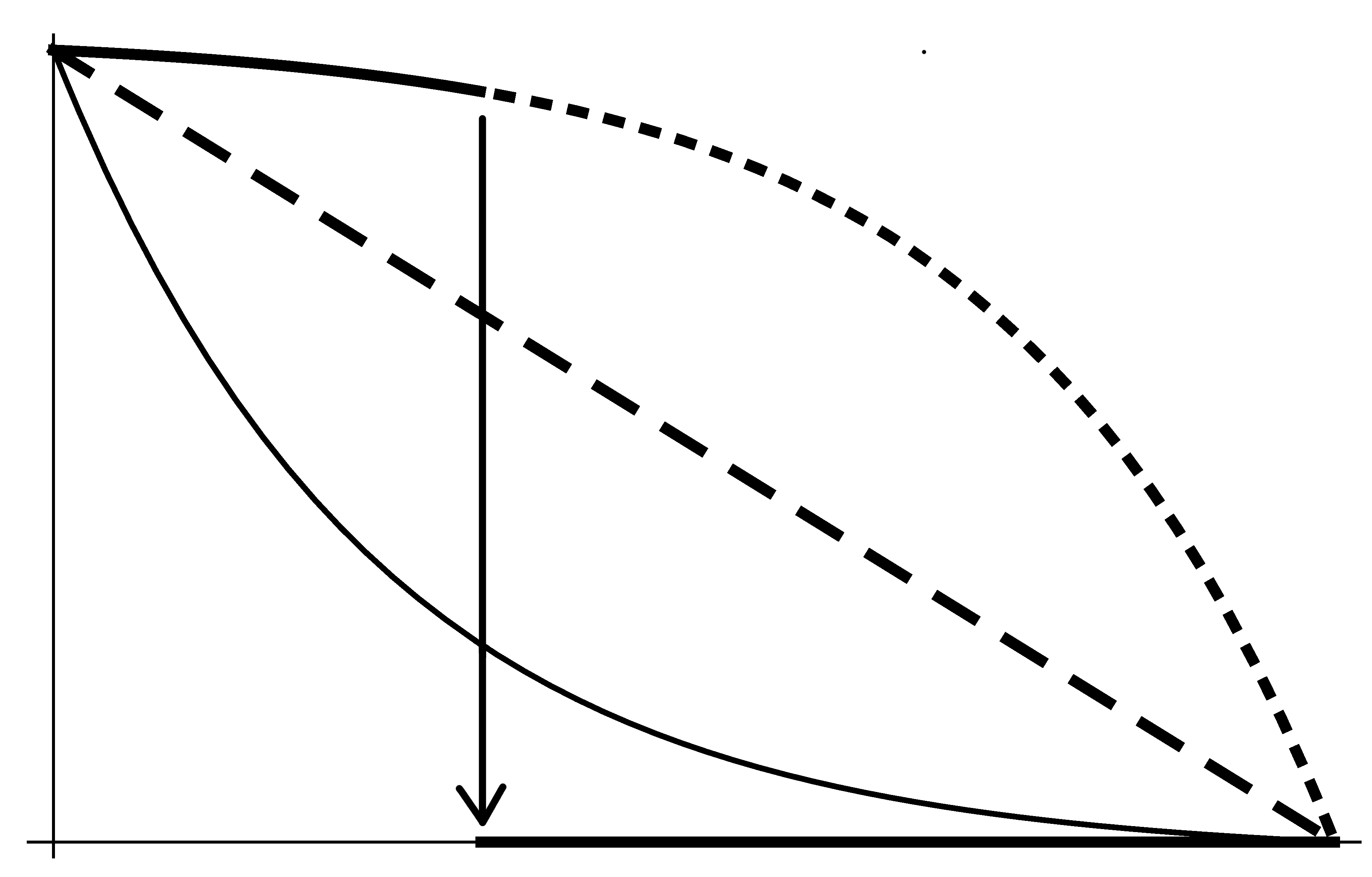}
\put(97,-2.5){\scriptsize{$T$}}
\put(33,-2.5){\scriptsize{$\tau$}}
\put(0,53){\scriptsize{$x$}}
\put(100,2){\scriptsize{Time}}
\put(0,57){\scriptsize{Size of asset position}}
\end{overpic} \end{tabular} 
&\begin{tabular}{l}
\hspace{4ex}
\end{tabular}
&
\begin{tabular}{l}\begin{overpic}[height=3.5cm, width=6cm]{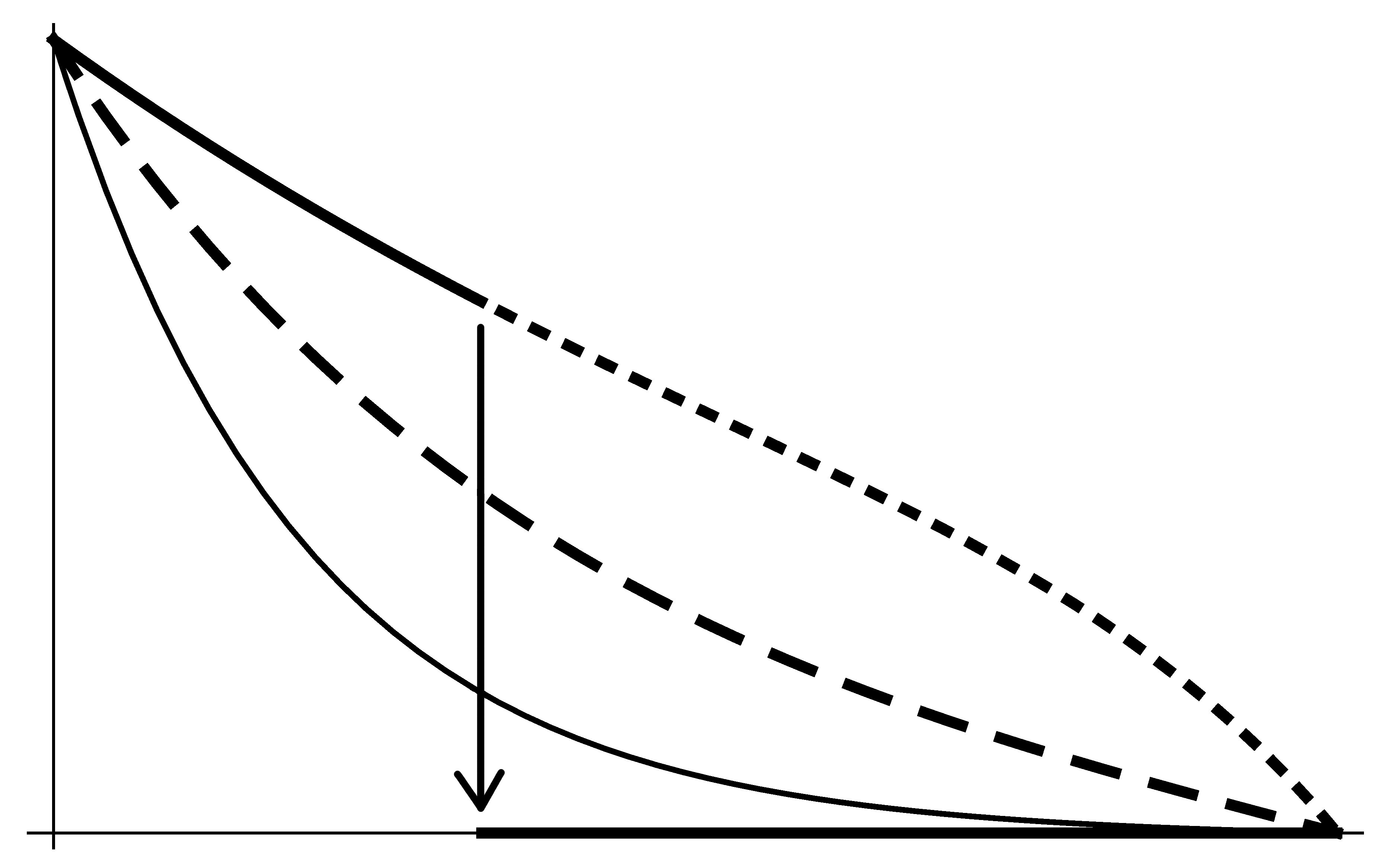}
\put(97,-2.5){\scriptsize{$T$}}
\put(33,-2.5){\scriptsize{$\tau$}}
\put(0,53){\scriptsize{$x$}}
\put(100,2){\scriptsize{Time}}
\put(0,57){\scriptsize{Size of asset position}}
\end{overpic} \end{tabular} 
\end{tabular}
\caption{Optimal trading trajectories for a risk neutral (left picture) respectively a risk averse trader (right picture). The thick solid lines denote the optimal trading trajectory with dark pool in the scenario where the dark pool order is executed at time $\tau$. The dotted lines denote the scenario where the dark pool order is not executed during the entire trading horizon and the thin solid lines refer to the expected asset positions. The dashed lines denote optimal liquidation without dark pools. $x=1$, $T=1$, $\Lambda=1$, $\theta=4$ and for the right picture $\alpha=6$, $\Sigma=1$.}
\label{FigureSingle}
\end{figure}

\subsubsection{Dependence on $\Lambda$ and $\alpha \Sigma$} \label{SubSubSecPropSingleLambda}

It follows directly from the definition of the cost functional $J$ that the value function is strictly increasing both in the impact costs parameter $\Lambda$ and in the risk costs parameter $\alpha \Sigma$. For the optimal strategy, impact costs and risk costs have conflicting influences: while larger impact costs yield a reduction of the trading intensity, larger risk costs yield faster trading (cf.~also the difference of the left and the right picture of Figure~\ref{FigureSingle}). We summarize these findings in the following proposition.

\begin{prop} \label{PropPropertiesn=1ContImpact}
\begin{enumerate}
\item[(i)]
For $t\in [0,T)$, $C(t;\Lambda,\alpha \Sigma)$ is strictly increasing in $\Lambda$ and in $\alpha \Sigma$.
\item[(ii)]
Let $t \in [0,T)$ and $x >0$ be fixed. Then 
$\xi^*(t,x;\Lambda)$ 
is strictly decreasing in $\Lambda$.
Consequently, $\tilde{X}(t;\Lambda)$ is strictly increasing in $\Lambda$ for $t \in (0,T)$.
\item[(iii)]
Let $t \in [0,T)$ and $x >0$ be fixed. Then 
$\xi^*(t,x;\alpha \Sigma)$ 
is strictly increasing in $\alpha \Sigma$.
Consequently, $\tilde{X}(t;\alpha \Sigma)$ is strictly decreasing in $\alpha \Sigma$ for $t \in (0,T)$.
\end{enumerate}
\end{prop}

\subsection{Portfolio liquidation} \label{SubSecContMuliAsset}

If a risk averse investor has to liquidate a portfolio of multiple assets ($n\geq2$), 
then correlation between the assets comes into play. Depending on whether the portfolio is well (poorly) balanced, it is intuitively optimal to place orders in the dark pool which are smaller (larger) than the remaining portfolio position. In the first case, the risk costs of the portfolio are small and therefore, the trader will not risk to lose her balanced position by the full execution of her dark pool order for only one of the assets; hence her orders are smaller than the remainder of the position. In the second case, it might be optimal to place orders in the dark pool which are larger than the remainder of the position for risk mitigation reasons. This illustration suggests that overly simple adjustments of existing trading algorithms for optimal liquidation without dark pools can have undesirable results. 

In Section~\ref{SubSecDpendenceRhon=2}, we verify the above intuition by analytical results about the dependence of the value function and the optimal strategy on the correlation of a portfolio of two assets; we then deduce the general structure of the optimal strategy dependent on whether the portfolio is well or poorly balanced. We also discuss the dependence of the value function and the optimal strategy on the price impact parameter and the execution intensities in Section~\ref{SubSecPortfolioLambda} and~\ref{SubSecPortfolioTheta}, respectively.

As a prerequisite, we introduce a characterization of the optimal dark pool order which exploits that the jump times of the Poisson processes are almost surely distinct. Intuitively, an execution of the optimal dark pool order for asset $i$ should bring the position in asset $i$ to its optimal value given unchanged positions in all other assets $j \not=i$. The following Proposition confirms this intuition (see also~\cite{Naujokat2010} for a similar result). 

\begin{prop} \label{PropOptDPOrder}
Let $t\in [0,T]$, $x\in \mathds{R}^n$ be the portfolio position at time $t$ and $i=1,\dots,n$. Then, 
\[
v(t,x-\eta^*_i(t,x) e_i) = \min\limits_{\eta \in \mathds{R}} v(t,x-\eta e_i).
\]
\end{prop}

\subsubsection{Dependence on correlation} \label{SubSecDpendenceRhon=2}

We will see that dark pool trading is sensitive to the correlation of price increments. In the following, we discuss the case $n=2$.\footnote{If $n>2$, the situation is more complicated: positive correlation is not transitive in general. Hence, we cannot use Definition~\ref{DefWellDiv} below for the characterization of well versus poorly diversified portfolios. If the correlations of the price processes satisfy transitivity, some of the results can be generalized.}  In order to simplify the exposition, we assume that there is no cross asset price impact:\footnote{Cross price impact and correlation can have conflicting influences on the value function and the optimal strategy. We discuss cross asset impact at the end of Section~\ref{SubSecPortfolioLambda}.}
\begin{equation} \label{EqLambdaDiagonal}
\Lambda = 
\begin{pmatrix}
\lambda_1 & 0 \\
0 & \lambda_2
\end{pmatrix},
\quad
\Sigma = 
\begin{pmatrix}
\sigma_1^2 & \rho \sigma_1 \sigma_2 \\
\rho \sigma_1 \sigma_2 & \sigma_2^2
\end{pmatrix}.
\end{equation} 
For the purposes of this section, we assume that the variances $\sigma_1$ and $\sigma_2$ of the two assets as well as the risk aversion parameter $\alpha$ are strictly positive.

If the correlation of the two assets is positive ($\rho>0$), a portfolio consisting of a long position in one asset and a short position in the other asset is more desirable than long positions (or short positions) in both assets; in the former case, a part of the risk of each asset is hedged by the other asset. Conversely, if $\rho <0$, it is more desirable to have long (or short) positions in both assets. 

\begin{definition} \label{DefWellDiv}
A portfolio $x=(x_1,x_2)^\top$ ($x_1,x_2\not=0$) is \emph{well diversified} if either the signs of the positions are equal ($\sgn(x_1) = \sgn(x_2)$) and $\rho<0$ or if the signs of the positions are different and $\rho>0$. Otherwise, the portfolio is \emph{poorly diversified}.
\end{definition}

\begin{prop} \label{PropPropertiesCharactWell}
Let $t\in [0,T)$ and $x_1,x_2\not=0$. Then $v(t,(x_1,x_2)^\top) <  v(t,(x_1,-x_2)^\top)$ if and only if the portfolio $x$ is well diversified; $v(t,(x_1,x_2)^\top) >  v(t,(x_1,-x_2)^\top)$ if and only if the portfolio $x$ is poorly diversified. 
\end{prop}

We can further specify the dependence of the value function on the correlation if the portfolio is well diversified.

\begin{prop} \label{PropPropertiesValueWellMonotone}
Let $t\in [0,T)$ and $x$ be well diversified. Then $v(t,x)$ is strictly decreasing in $|\rho|$.
\end{prop}

The left picture of Figure~\ref{FigTwoMonotonValue} illustrates the dependence of the value function on the correlation $\rho$ for a portfolio that is long in both assets. For~$\rho <0$, this portfolio is well diversified and the value function is increasing in $\rho$ (i.e., decreasing in $|\rho|$) in line with Proposition~\ref{PropPropertiesValueWellMonotone}. For $\rho>0$, the portfolio is poorly diversified. This leads to elevated liquidation costs for small positive~$\rho$. For large positive~$\rho$, the increased risk costs of the current portfolio are outweighed by the (projected) smaller risk costs of a future well diversified portfolio (e.g., after the execution of an order in the dark pool; cf.~also Proposition~\ref{PropPropertiesNeverShort}~(ii) below). The opportunity of risk reduction results in a decrease of the value function for large values of $\rho$ in the displayed example.

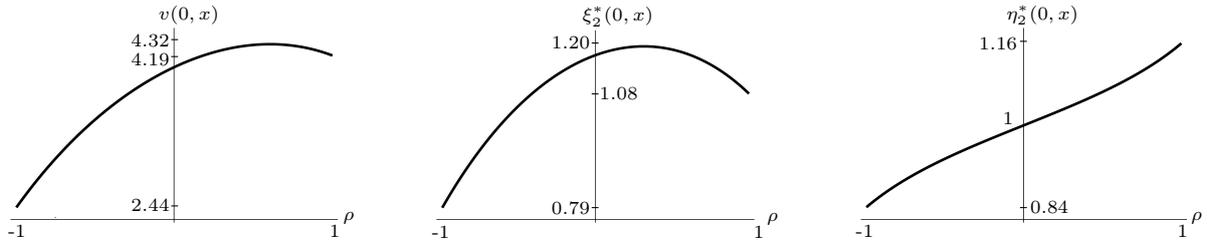
\begin{figure}[!ht]
\centering\vspace{4ex}
\begin{tabular}{lllll}
\begin{tabular}{l}\begin{overpic}[height=2.7cm, width=4.5cm]{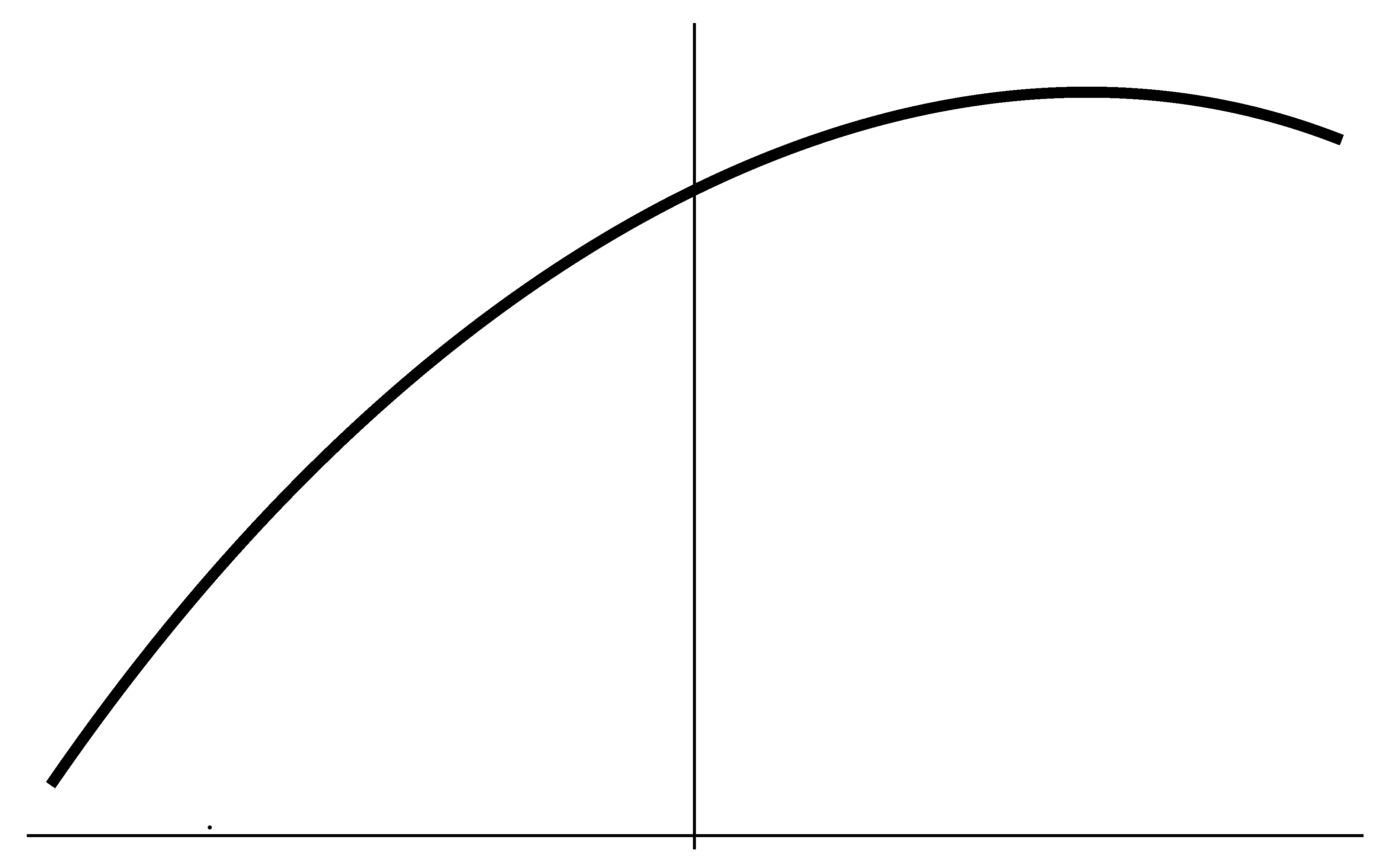}
\put(1,-3){\scriptsize{-1}}
\put(37,53){\scriptsize{$4.32$}}
\put(48.7,55){\line(1,0){2}}
\put(37,47.5){\scriptsize{$4.19$}}
\put(48.7,50){\line(1,0){2}}
\put(37,4.5){\scriptsize{$2.44$}}
\put(48.7,6){\line(1,0){2}}
\put(99,2){\scriptsize{$\rho$}}
\put(45,61){\scriptsize{$v(0,x)$}}
\put(95,-3){\scriptsize{1}}
\end{overpic} 
\end{tabular} 
\begin{tabular}{l}
\vspace{1.5ex}
\end{tabular}
\begin{tabular}{l}\begin{overpic}[height=2.7cm, width=4.5cm]{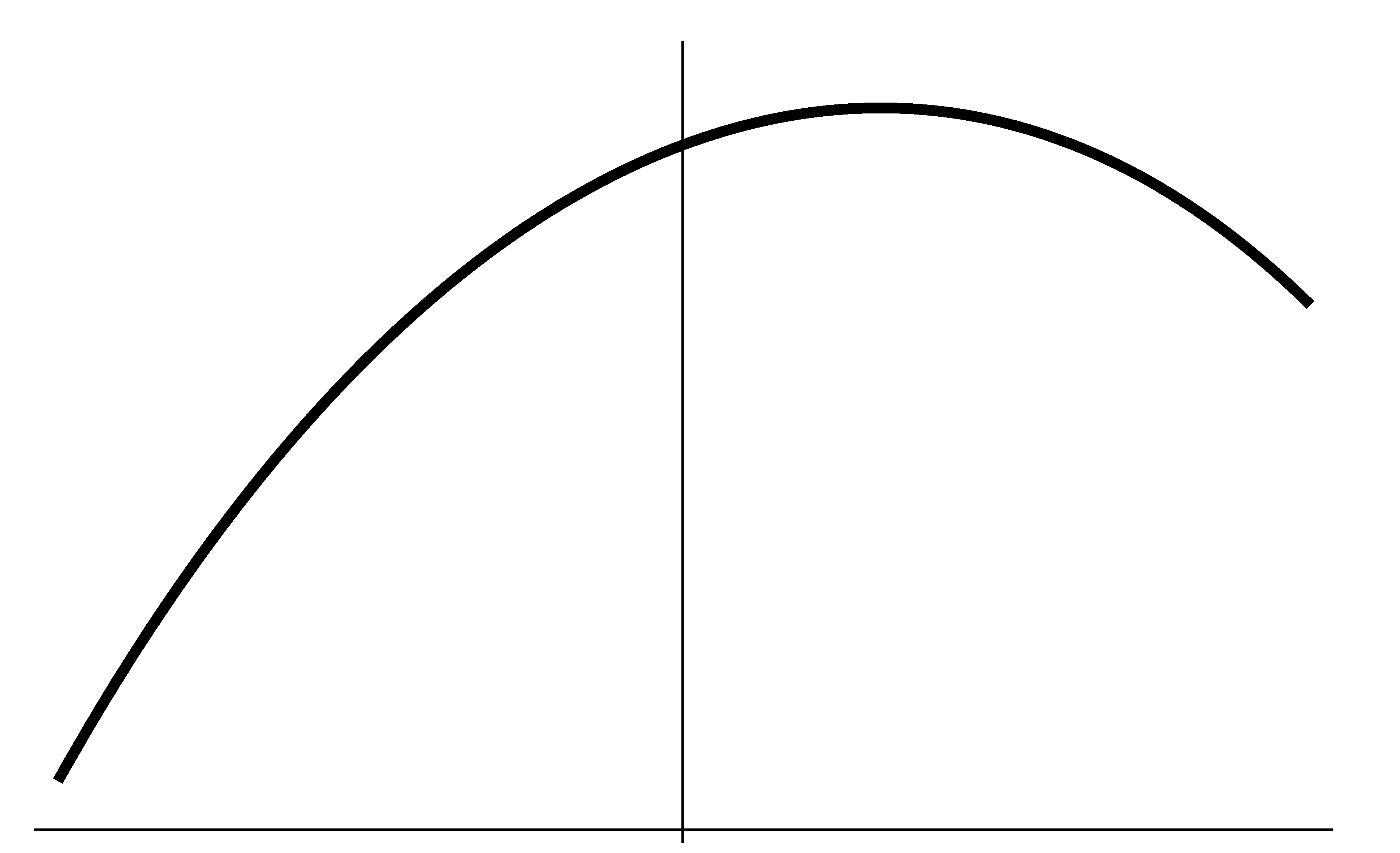}
\put(1,-4){\scriptsize{-1}}
\put(36,52){\scriptsize{$1.20$}}
\put(50,37){\scriptsize{$1.08$}}
\put(36,3){\scriptsize{$0.79$}}
\put(47.7,53.5){\line(1,0){2}}
\put(47.7,38.5){\line(1,0){2}}
\put(47.7,5){\line(1,0){2}}
\put(99,1){\scriptsize{$\rho$}}
\put(45,60){\scriptsize{$\xi_2^*(0,x)$}}
\put(95,-4){\scriptsize{1}}
\end{overpic} 
\end{tabular} 
\begin{tabular}{l}
\vspace{1.5ex}
\end{tabular}
\begin{tabular}{l}\begin{overpic}[height=2.7cm, width=4.5cm]{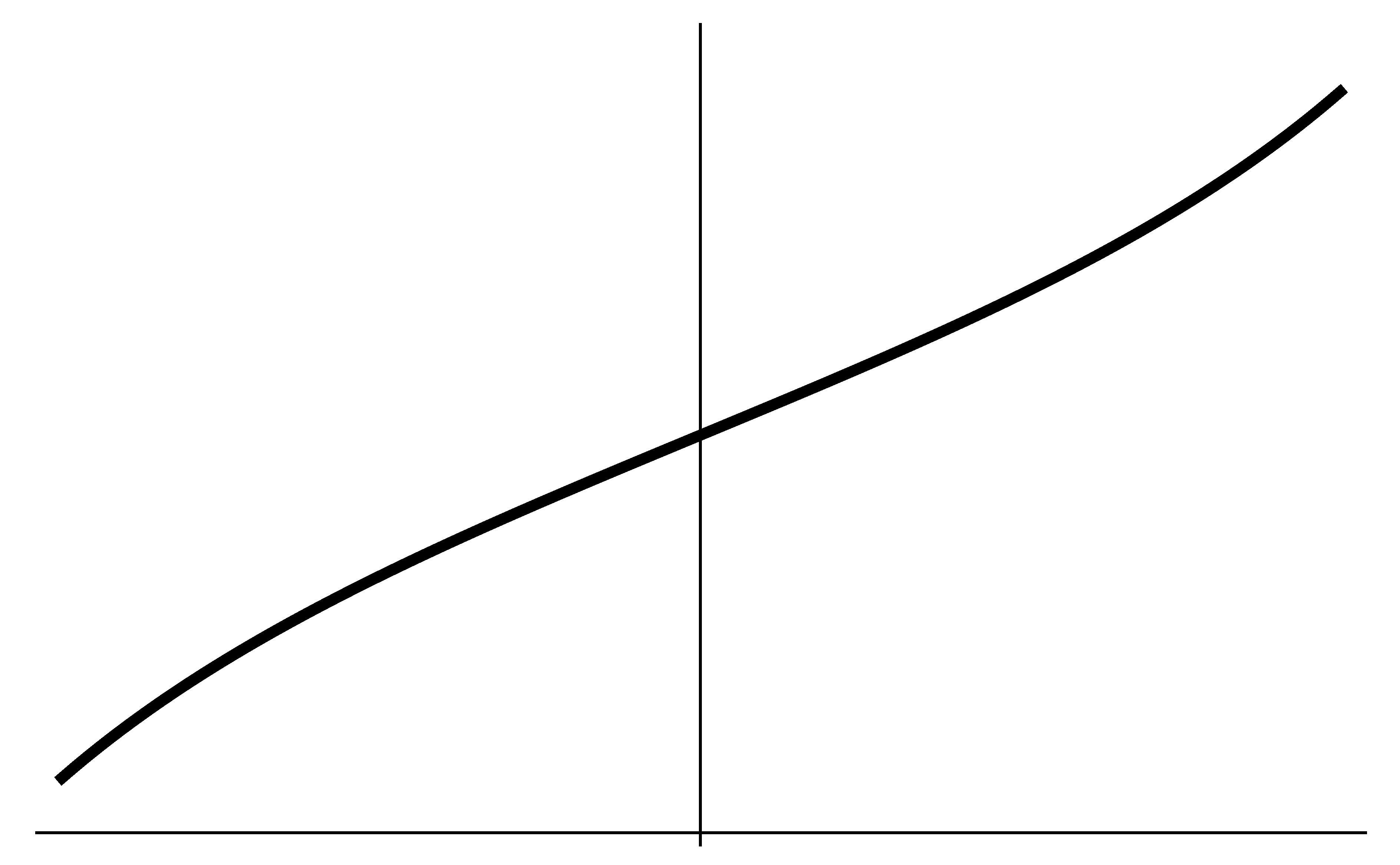}
\put(1,-3){\scriptsize{-1}}
\put(37,52){\scriptsize{$1.16$}}
\put(44,30){\scriptsize{$1$}}
\put(52,4){\scriptsize{$0.84$}}
\put(49,54.5){\line(1,0){2}}
\put(49,5.5){\line(1,0){2}}
\put(99,2){\scriptsize{$\rho$}}
\put(45,61){\scriptsize{$\eta_2^*(0,x)$}}
\put(95,-3){\scriptsize{1}}
\end{overpic} 
\end{tabular} 
\end{tabular}
\vspace{1ex}
\caption{Dependence of the value function (left picture), the optimal trading intensity for the second asset (middle picture) and the optimal dark pool order for the second asset (right picture) on the correlation $\rho$. $x_1=x_2=1$, $T=1$, $\lambda_1=3$, $\lambda_2=0.2$, $\theta_1=0.5$, $\theta_2=5$, $\alpha=4$ and $\sigma_1=\sigma_2=1$.}
\label{FigTwoMonotonValue}
\end{figure}

We have the following symmetry and monotonicity properties for the entries of the value function matrix $C$.

\begin{prop} \label{PropMatrixEntries}
Let $t \in [0,T)$ and denote the entries of the value function matrix by $c_{i,j}(t;\rho)$, $i,j=1,2$.
\begin{enumerate}
\item[(i)]
$c_{1,1}(t;\rho), c_{2,2}(t;\rho)>0$ and $c_{1,1}(t;\rho) =c_{1,1}(t;-\rho)$, $c_{2,2}(t;\rho) =c_{2,2}(t;-\rho)$, $\sgn(c_{1,2}(t;\rho))=\sgn(\rho)$ and $c_{1,2}(t;\rho)=-c_{1,2}(t;-\rho)$.
\item[(ii)]
$c_{1,1}(t;\rho)$ and $c_{2,2}(t;\rho)$ are strictly increasing in $\rho$ on $[-1,0)$ and strictly decreasing in $\rho$ on $(0,1]$. $c_{1,2}(t;\rho)$ is increasing in $\rho$ on $[-1,1]$.
\end{enumerate}
\end{prop}

The risk mitigation opportunity created by a strong correlation of the price increments becomes apparent again in Proposition~\ref{PropMatrixEntries}. Liquidating a single asset position $x=(x_1,0)^\top$ results in the cost $v(t,(x,0)) = c_{1,1}(t;\rho) x_1^2$, which exhibits a strict local maximum at $\rho=0$ and decreases as the correlation between the two assets becomes stronger (irrespective of the sign of the correlation). This implies in particular that it is optimal for the investor to trade in both assets (unless $\rho=0$) even if the current position in one asset is zero.

Proposition~\ref{PropPropertiesCharactWell} suggests that an optimal liquidation strategy never changes the sign of the asset positions of a well diversified portfolio and always seeks to turn a poorly diversified portfolio into a well diversified portfolio. The following proposition confirms this conjecture. 

\begin{prop} \label{PropPropertiesNeverShort}
Let $t \in [0,T)$ and $x\in \mathds{R}^2$ be the portfolio position at time $t$.
\begin{enumerate}
\item[(i)]
If $x$ is well diversified, then $X^*(s)$ is well diversified for all $s \in [t,T)$ with $\sgn(X_i^*(s)) = \sgn(x_i)$ for $i=1,2$. 
\item[(ii)]
If $x$ is poorly diversified, then 
$
\sgn(X^*_i(s-) - \eta_i^*(s,X^*(s-))) \neq \sgn(x_i)
$ for $s \in [t,\tau)$, where 
\[
\tau=\inf\{s\geq t| \sgn(X_i^*(s))\not= \sgn(x_i) \text{ or } X^*_i(s)=0 \text{ for some } i=1,2\} \wedge T >t\quad \text{a.s.}
\]
\end{enumerate}
\end{prop}

By Proposition~\ref{PropPropertiesNeverShort}, the investor trades in both assets during the entire trading horizon $[0,T]$ if the portfolio is well diversified. If the portfolio is poorly diversified the execution of the dark pool order in one of the assets \emph{always} changes the sign of the position. If dark pool orders are never executed, it can be optimal to decrease the risk costs by changing the sign of the position in one of the assets by only trading at the exchange (as it is the case  in the numerical example underlying the left picture of Figure~\ref{WellVsBadlyHedged}); in general, this is not the case.\footnote{Consider, e.g., a portfolio where $x_1=x_2$, $\lambda_1=\lambda_2$, $\sigma_1=\sigma_2$ and $\theta_1=\theta_2$. Then the optimal trading intensities for the two assets must be equal until a dark pool order is executed. In particular, if the orders in the dark pool are never executed, both position must become zero at the same time after which further trading is not optimal.}


In the following, we specify the dependence of the optimal strategy on $\rho$. By Theorem~\ref{TheoremOptLiq}, we have
\begin{align}
\xi^*_1(t,x)& =\frac{1}{\lambda_1} \big(c_{1,1}(t;\rho) x_1 + c_{1,2}(t;\rho) x_2 \big),
&\xi^*_2(t,x)&=\frac{1}{\lambda_2} \big(c_{2,2}(t;\rho) x_2 + c_{1,2}(t;\rho) x_2 \big), 
\label{EqOptXin=2}\\
\eta^*_1(t,x)&= x_1 + \frac{ c_{1,2}(t;\rho)}{ c_{1,1}(t;\rho)} x_2 ,
&\eta^*_2(t,x)&= x_2 + \frac{ c_{1,2}(t;\rho)}{ c_{2,2}(t;\rho)} x_1. \label{EqOptEtan=2}
\end{align}

\begin{prop} \label{PropStratMon}
Let $t \in [0,T)$, $x=(x_1,x_2)^\top$, $x_1,x_2>0$ and $i=1,2$.
\begin{enumerate}
\item[(i)]
$\eta_i(t,x)$ is strictly increasing in $\rho$ for $\rho\in [-1,1]$.
\item[(ii)]
$\xi_i(t,x)$ is  strictly increasing in $\rho$ for $\rho\in [-1,0)$.
\end{enumerate}
Analog results hold for $x_1,x_2<0$ and $\sgn(x_1)\not=\sgn(x_2)$.
\end{prop}

For a well diversified portfolio, the profit from diversification is increasing in $|\rho|$ (cf.~Proposition~\ref{PropPropertiesValueWellMonotone}). Therefore, the investor decreases her trading activity both at the exchange and in the dark pool for larger $\rho$. For a poorly diversified portfolio, this is not necessarily the case (cf.~the left picture of Figure \ref{FigTwoMonotonValue} and the corresponding discussion); in contrast to the optimal dark pool orders, the trader might decrease her trading intensity for large positive $\rho$ in order to save price impact costs while waiting for the execution of an order in the dark pool (which yields a well diversified position by Proposition~\ref{PropPropertiesNeverShort}~(ii) below). We illustrate the dependence of $u^*$ on $\rho$ in the middle and the right picture of Figure~\ref{FigTwoMonotonValue}. In the displayed case, the trading intensity of the second asset is not increasing in $\rho$ if $x$ is poorly diversified (i.e., $\rho>0$). On the other hand, the optimal dark pool order is strictly increasing for $\rho \in [-1,1]$. The symmetry of the graph in the right picture follows directly from the symmetries of $c_{1,2}$ and $c_{2,2}$. If $\rho=0$, the optimal strategies for the two assets are independent and follow from the formulae of Section~\ref{SecContPropOne}; in particular, the optimal dark pool order equals $x_2=1$.

We deduce the general structure of the optimal strategy for initial positions $x_1,x_2>0$ from the above results. If $x$ is poorly diversified ($\rho >0$), we have $\xi_i(t,x), \eta_i(t,x)>0$ by Equations~\eqref{EqOptXin=2} and~\eqref{EqOptEtan=2}; hence it is optimal to decrease the position. This is not necessarily the case for a well diversified portfolio ($\rho <0$); in this case both  $\xi_i(t,x)$ and $\eta_i(t,x)$ can be negative as $c_{1,2}(t;\rho)<0$ (cf.~also the middle picture of Figure~\ref{FigNonMon} below). For the first asset, this is the case if and only if $x_1 < - \tfrac{c_{1,2}(t;\rho)}{c_{1,1}(t;\rho)} x_2$. It can hence be optimal to \emph{increase} the position in order to further reduce the risk costs of the portfolio. Note that as $\sgn(\xi_i(t,x))=\sgn(\eta_i(t,x))$ the direction of trading in the dark pool and at the primary exchange is always the same. Furthermore, it can be optimal to neither trade at the exchange nor in the dark pool (in the first asset) at time $t$ if $x_1 = - \tfrac{c_{1,2}(t;\rho)}{c_{1,1}(t;\rho)} x_2$.

As the execution of an order in the dark pool balances the trade-off of price impact costs against risk costs, it should intuitively never be optimal to increase a positive position or to decrease a negative position once the dark pool order of one of the assets has been executed. The following proposition confirms this conjecture.

\begin{prop} \label{PropWrongDirection}
Let $t\in [0,T)$, $x\in \mathds{R}^2$ be the portfolio position at time $t$ and $\tau_1$ be the first jump time of $\pi$. Then for all $s \geq \tau_1$, $i=1,2$, 
\[
\sgn(\xi_i(s,X^*(s-)))=\sgn(\eta_i(s,X^*(s-)))=\sgn(X_i^*(s-)) \quad \text{or} \quad 
\xi_i(s,X^*(s-))=\eta_i(s,X^*(s-))=0.
\]
\end{prop}

We close the section by illustrating the structure of the optimal strategy and its dependence on $\rho$ by a numerical example. To this end, we consider two strongly positively correlated assets with $\lambda_1=3$, $\lambda_2=0.2$, i.e., the second asset is more liquid than the first asset. We also model the dark pool in such a way that the execution of orders for the second asset is more probable than for the first asset: 
$\theta_1=0.5$, $\theta_2=3$.\footnote{Our choice of the parameters reflects the intuition that the asset which is more liquid at the exchange (smaller $\lambda$) is also more liquid in the dark pool (larger $\theta$). Theoretical findings of \cite{Ye2011} support this choice. However, we are not aware of any empirical evidence for this; in some cases, the opposite parameter choice can also be plausible.} We consider a poorly diversified portfolio $x=(1,1)^\top$ and a well diversified portfolio $x=(1,-1)^\top$ of the two stocks.

\begin{figure}
\centering
\begin{tabular}{lll}
\begin{tabular}{l}
\begin{overpic}[height=3.5cm, width=6cm]{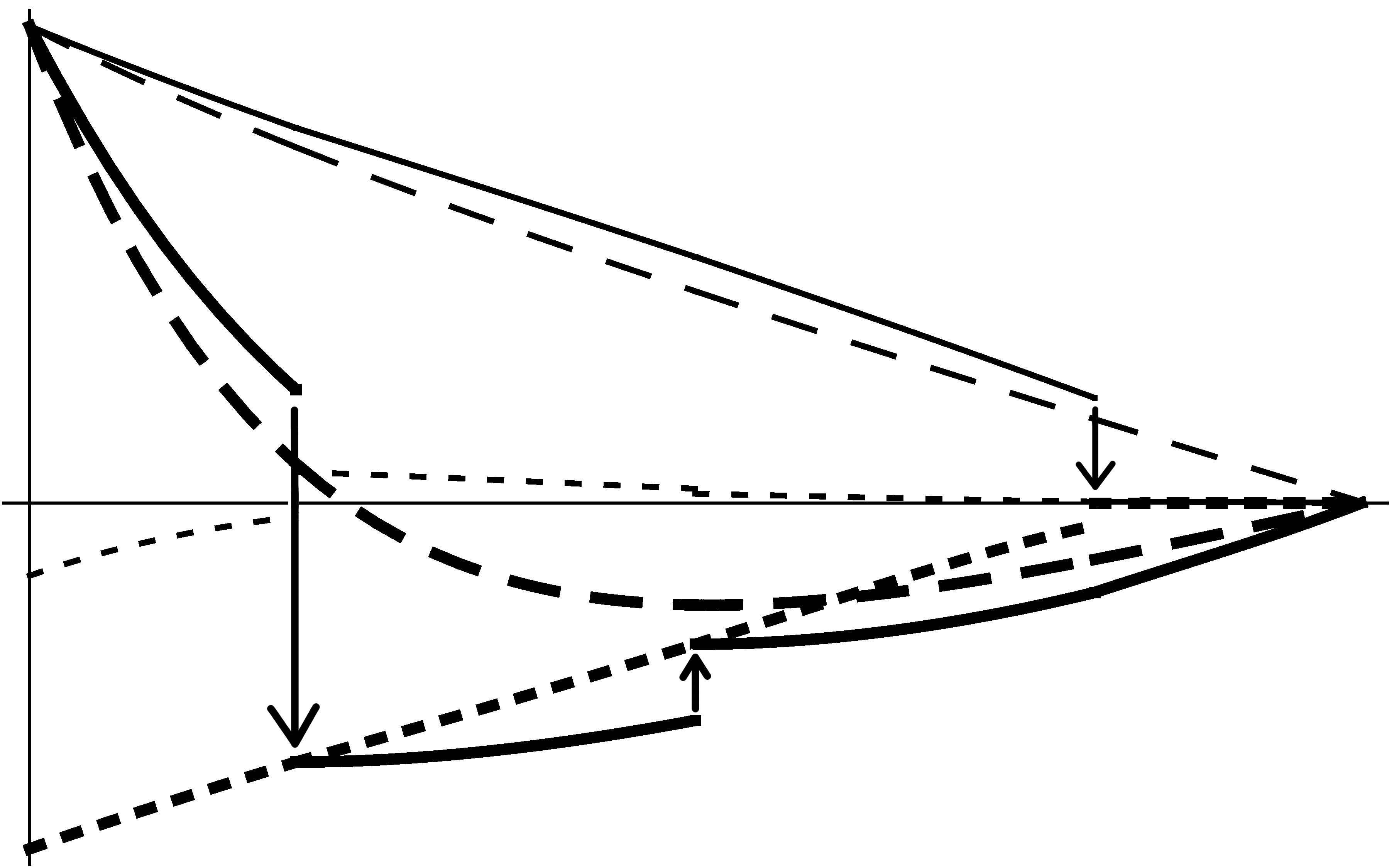}
\put(97,20){\scriptsize{$T$}}
\put(22,21){\scriptsize{$\tau_1$}}
\put(47,22){\scriptsize{$\tau_2$}}
\put(72,27){\scriptsize{$\tau_3$}}
\put(-15,55){\scriptsize{$x_1=x_2$}}
\put(100,24){\scriptsize{Time}}
\put(0,60){\scriptsize{Size of asset position}}
\end{overpic}
\end{tabular} &
\begin{tabular}{l}
\hspace{2ex}
\end{tabular}
&
\begin{tabular}{l}
\begin{overpic}[height=3.5cm, width=6cm]{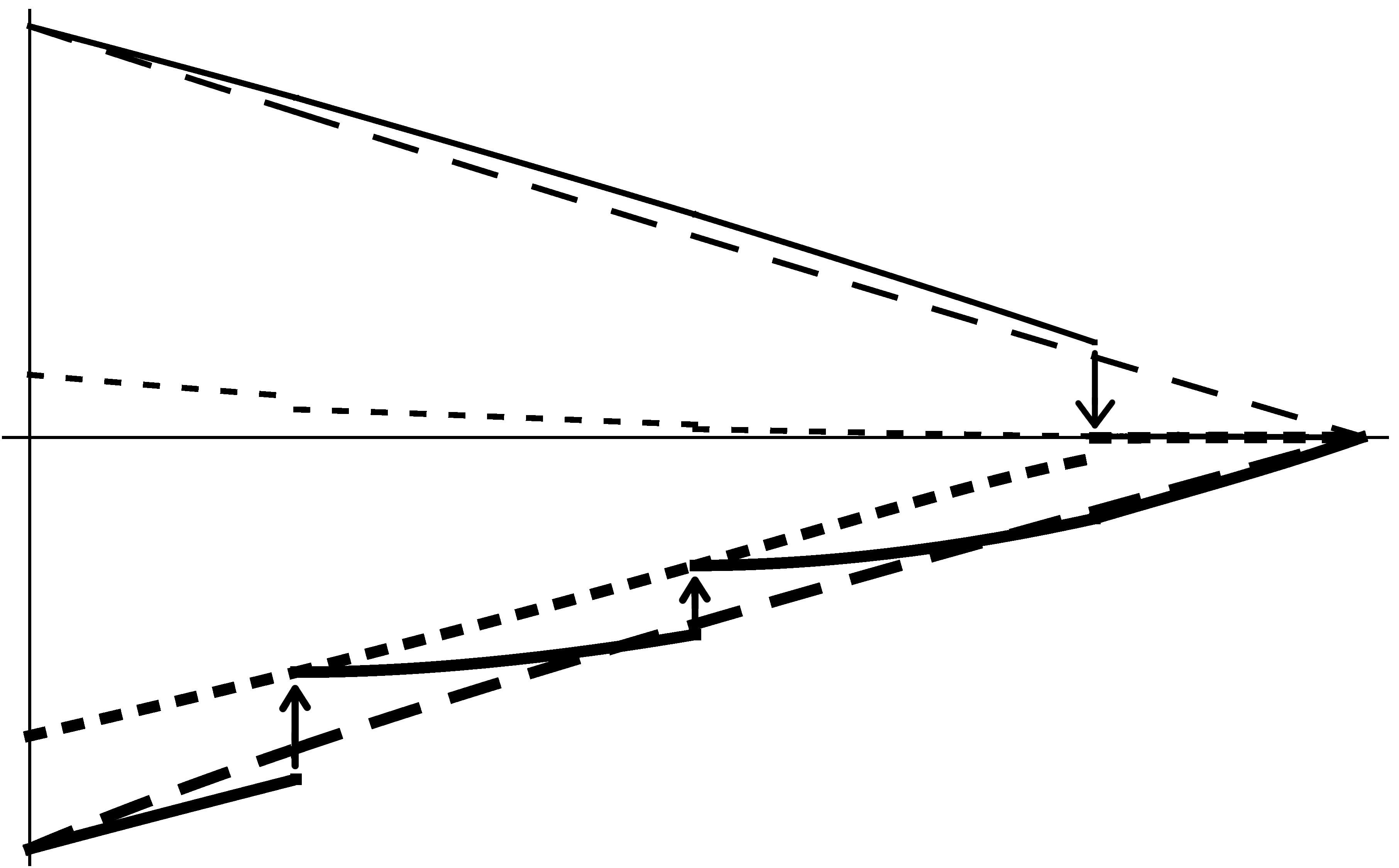}
\put(97,25){\scriptsize{$T$}}
\put(22,27){\scriptsize{$\tau_1$}}
\put(47,27){\scriptsize{$\tau_2$}}
\put(72,32){\scriptsize{$\tau_3$}}
\put(-4,56){\scriptsize{$x_1$}}
\put(-4,1){\scriptsize{$x_2$}}
\put(100,28){\scriptsize{Time}}
\put(0,60){\scriptsize{Size of asset position}}
\end{overpic}
\end{tabular}
\end{tabular}
\caption{Evolution of a portfolio consisting of two highly correlated stocks over time. The left picture illustrates a poorly diversified portfolio, the right picture a well diversified portfolio. In both pictures, thin lines are used for the less liquid first stock and thick lines for the more liquid second stock. Dashed lines correspond to trading without the
dark pool and solid lines correspond to a realization of the liquidation process using the dark pool, where dark pool orders for the second stock are executed at times $\tau_1$ and $\tau_2$ and for the first stock at time $\tau_3$. Dotted lines correspond to the position which the investor aims to reach by her dark pool order for the respective stock. $x_1=1$, $x_2=1$ (left picture), $x_2=-1$ (right picture),  $T=1$, $\theta_1=0.5$, $\theta_2=3$, $\lambda_1=3$, $\lambda_2=0.2$,  $\alpha=4$, $\sigma_1=\sigma_2=1$ and $\rho=0.9$. }
\label{WellVsBadlyHedged}
\end{figure}
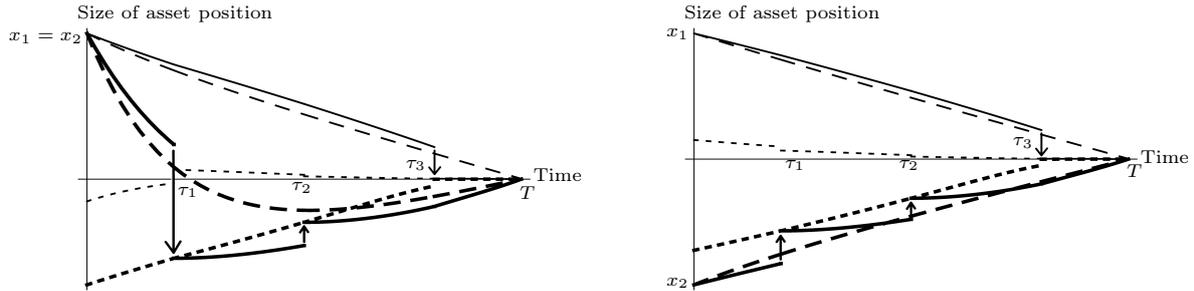

Figure~\ref{WellVsBadlyHedged} shows the evolution of the two portfolios if a risk averse investor applies the optimal strategy. The left picture corresponds to the poorly diversified portfolio, the right one to the well diversified portfolio. In both cases, thin lines are used for the first stock and thick lines for the second. Dashed lines correspond to trading without the dark pool and the solid lines correspond to a realization of the liquidation process using the dark pool, where the dark pool orders for the second stock are executed at times $\tau_1$ and $\tau_2$ and for the first stock only at time~$\tau_3$, i.e., dark pool orders for the more liquid stock are executed twice before any execution in the less liquid stock takes place. Dotted lines correspond to the position which the investor aims to reach by her dark pool order for the respective stock (cf.~Proposition~\ref{PropOptDPOrder}).

For the poorly diversified portfolio, the trader tries to improve her risky position by trading out of the second stock (cf.~Proposition~\ref{PropPropertiesNeverShort}~(ii) and the subsequent discussion). For this stock, trading in the primary venue is less expensive and being executed in the dark pool is more probable. 
For the well diversified portfolio, the portfolio position is decreasing almost linearly in time in all cases. 
Additionally, orders in the dark pool are very large for the poorly diversified portfolio and comparatively small for the well diversified portfolio, in line with Proposition~\ref{PropPropertiesNeverShort}. The reason is that dark pool orders are such that either the risk costs are decreased significantly by an execution (in the poorly diversified case) or they are only slightly increased (in the well diversified case).
Note that both for the poorly and for the well diversified portfolio, these effects are stronger for the liquid stock; for the illiquid stock, savings in price impact costs outweigh savings in risk costs. 

\subsubsection{Dependence on price impact} \label{SubSecPortfolioLambda}

In this section, we discuss the dependence of the value function and the optimal strategy on the price impact and on the cross price impact. It follows directly from the definition of the cost functional $J$ that the value function is increasing in the price impact matrix $\Lambda$. On the other hand, the increase of $v$ is in some sense bounded by the increase of $\Lambda$.

\begin{prop} \label{PropMontonValueLam}
Let $t\in [0,T)$ and $x \in \mathds{R}^n$. 
\begin{enumerate}
\item[(i)]
$v(t,x)$ is increasing in $\Lambda$.
\item[(ii)]
Let $\Lambda = \diag(\lambda_j,; j=1,\dots, n)$. For $i=1,\dots,n$, $\frac{v(t,x;\lambda_i)}{\lambda_i}$ is decreasing in $\lambda_i$.
\end{enumerate}
\end{prop}

The optimal trading strategy does not need to be monotone in the price impact parameter since risk mitigation and liquidation can be conflicting desires. We illustrate this situation in Figure~\ref{FigNonMon} by considering three different well diversified portfolios. In the first case ($x_1=0.7$, $x_2=-0.7$; left picture) further risk mitigation is not optimal; the trading intensity is decreasing in the price impact. In the second case ($x_1=0.7$, $x_2=-1.1$; middle picture) it is profitable to increase the position for small price impact in order to reach a position with even less risk costs; the optimal trading intensity is increasing and further risk mitigation is only profitable for small enough $\lambda_1$. In the third case ($x_1=0.7$, $x_2=-0.81$; right picture) the optimal intensity is increasing for very small $\lambda_1$ and then decreasing; the conflict apparent in the left and the middle picture destroys the monotonicity in this case.

\begin{figure}[!ht]
\centering\vspace{4ex}
\begin{tabular}{lllll}
\begin{tabular}{l}\begin{overpic}[height=2.7cm, width=4.5cm]{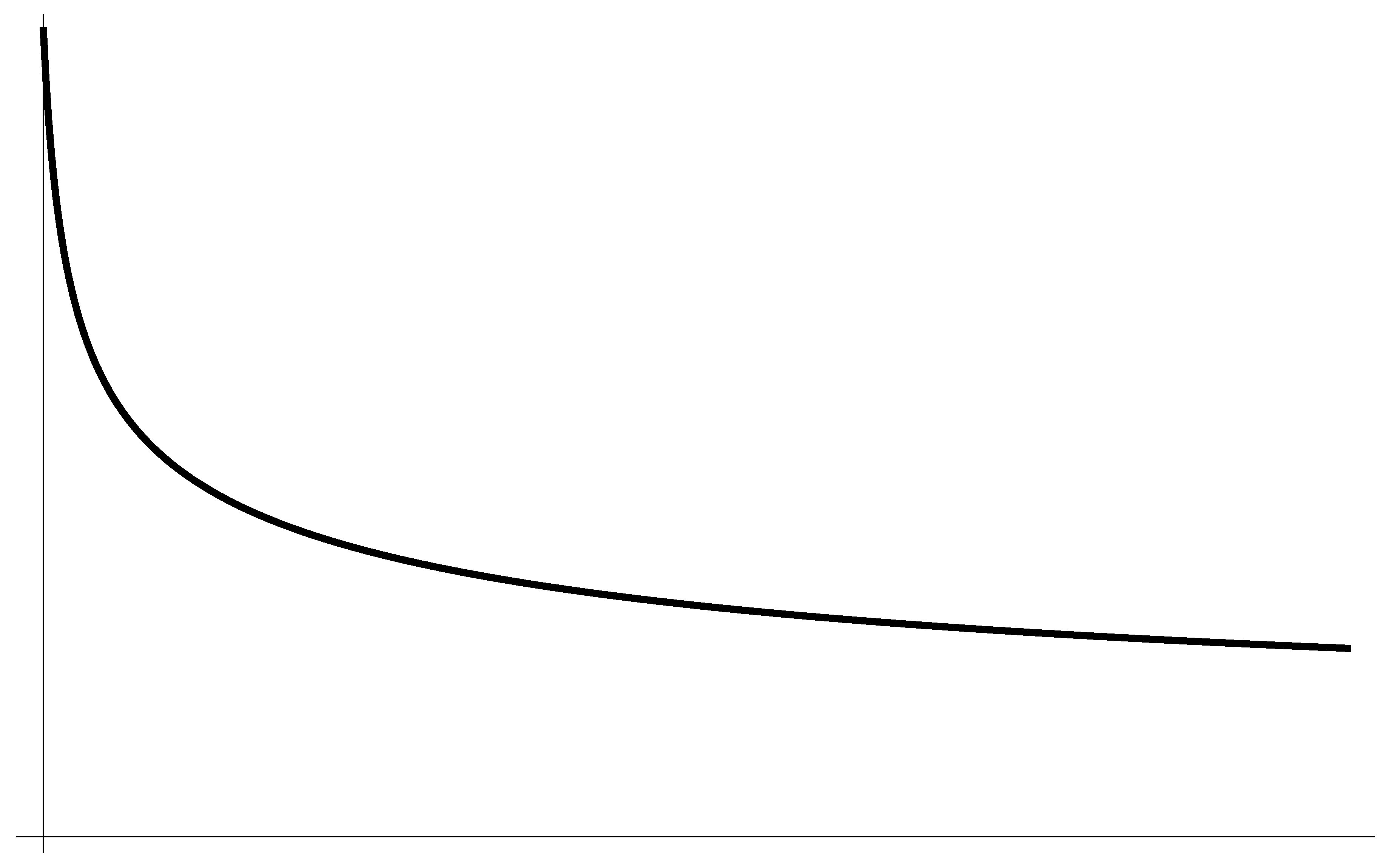}
\put(-10,56){\scriptsize{$2.94$}}
\put(-10,13){\scriptsize{$0.68$}}
\put(4,-4){\scriptsize{$0.005$}}
\put(100,0.5){\scriptsize{$\lambda_1$}}
\put(-2,63){\scriptsize{$\xi^*_1(0,x)$}}
\put(94,-4){\scriptsize{0.5}}
\put(2,58){\line(1,0){2}}
\put(2,15){\line(1,0){2}}
\put(97,0.5){\line(0,1){2}}
\end{overpic} 
\end{tabular} 
\begin{tabular}{l}
\hspace{1.5ex}
\end{tabular}
\begin{tabular}{l}\begin{overpic}[height=2.7cm, width=4.5cm]{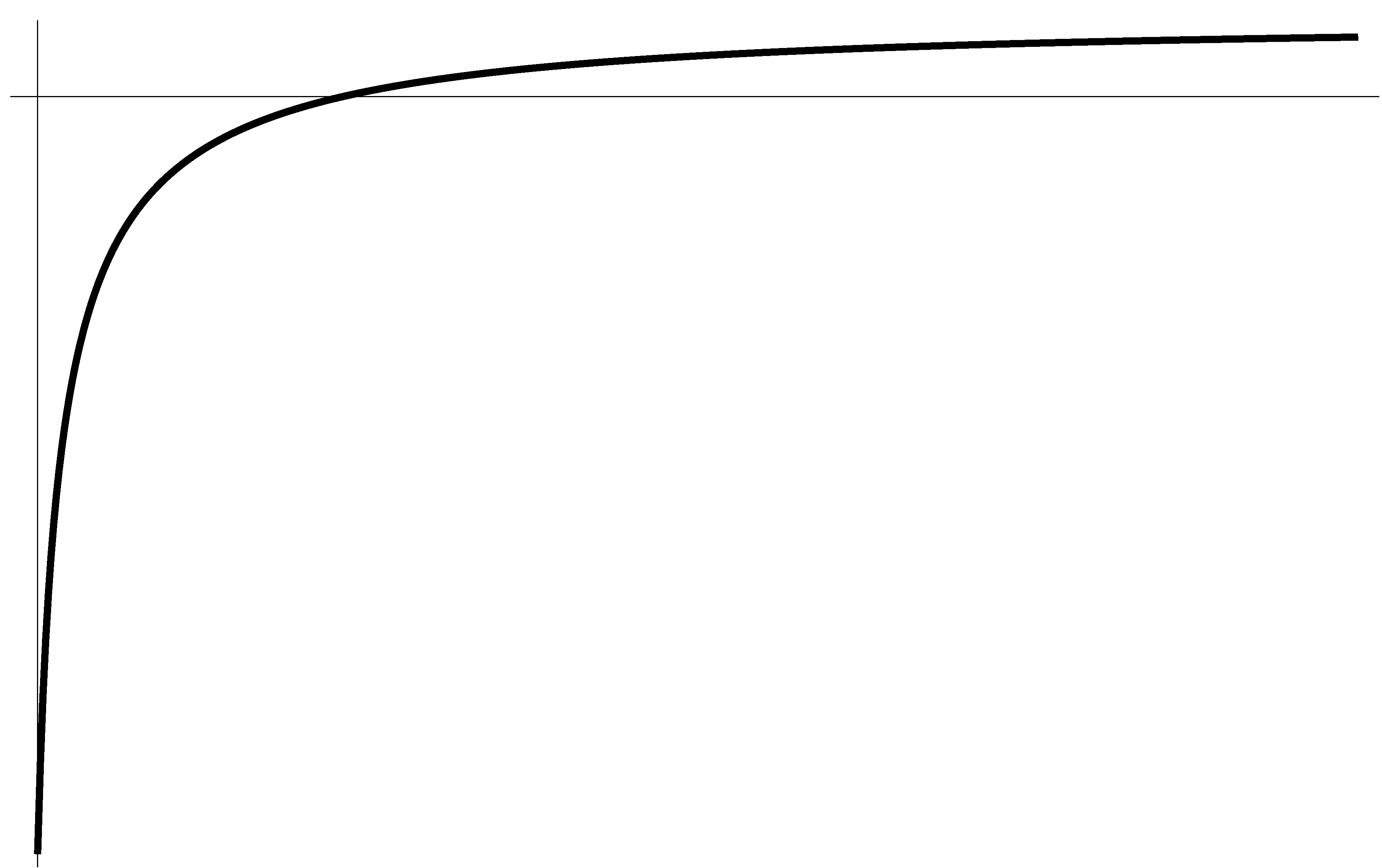}
\put(-10,56){\scriptsize{$0.32$}}
\put(5,0){\scriptsize{$-4.10$}}
\put(8,55){\scriptsize{$0.005$}}
\put(101,53){\scriptsize{$\lambda_1$}}
\put(-2,63){\scriptsize{$\xi^*_1(0,x)$}}
\put(94,47){\scriptsize{0.5}}
\put(1.7,57.6){\line(1,0){2}}
\put(1.7,1){\line(1,0){2}}
\put(98,52.3){\line(0,1){2}}
\end{overpic} 
\end{tabular} 
\begin{tabular}{l}
\hspace{1.5ex}
\end{tabular}
\begin{tabular}{l}\begin{overpic}[height=2.7cm, width=4.5cm]{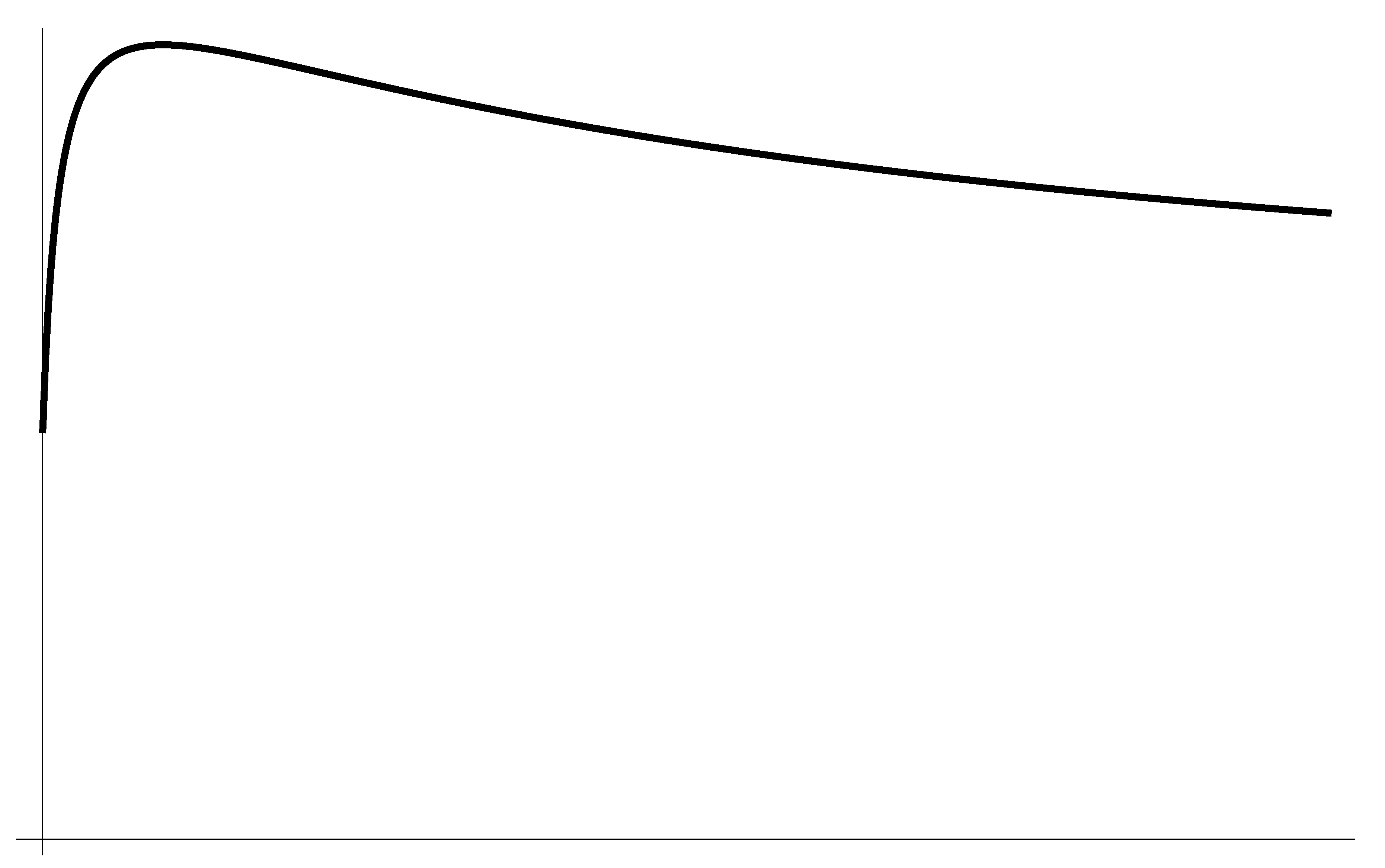}
\put(-10,44){\scriptsize{$0.55$}}
\put(-10,29){\scriptsize{$0.36$}}
\put(4,-4){\scriptsize{$0.005$}}
\put(100,0.5){\scriptsize{$\lambda_1$}}
\put(-2,63){\scriptsize{$\xi^*_1(0,x)$}}
\put(92,-4){\scriptsize{0.5}}
\put(2,46){\line(1,0){2}}
\put(2,30){\line(1,0){2}}
\put(96,0.8){\line(0,1){2}}
\end{overpic} 
\end{tabular} 
\end{tabular}
\vspace{1ex}
\caption{Dependence of the optimal trading intensity for the first asset on $\lambda_1 \in (0.005,0.5)$ for a well diversified portfolio $x$ with $x_1=0.7$, and different positions in the second asset: $x_2=-0.7$ (left picture), $x_2=-1$ (middle picture) and $x_2=-0.81$ (right picture). $T=1$, $\lambda_2=1$, $\theta_1=\theta_2=1$, $\alpha=4$, $\sigma_1=\sigma_2=1$ and $\rho=0.9$.}
\label{FigNonMon}
\end{figure}
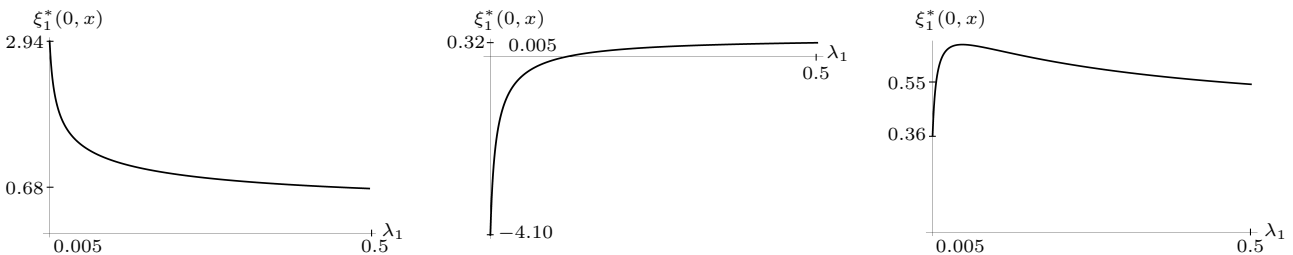

We close the section by analyzing the cross asset price impact $\lambda_{1,2}$ in the price impact matrix 
\begin{equation*}
\Lambda = 
\begin{pmatrix}
\lambda_1 & \lambda_{1,2} \\
\lambda_{1,2} & \lambda_2
\end{pmatrix} \quad \text{for} \quad 0 < |\lambda_{1,2}| < \sqrt{\lambda_{1}\lambda_{2}}.
\end{equation*}
We first consider the case $\sgn(\lambda_{1,2})=\sgn(\rho)$.\footnote{This case is more intuitive than the converse: if the assets are positively correlated, buying in the first asset should rather increase than decrease the price of the second asset. \label{FootCross}} We obtain the following analogs of Propositions~\ref{PropPropertiesCharactWell} and~\ref{PropPropertiesNeverShort}~(ii).

\begin{prop} \label{PropPropertiesCross}
Let $t\in [0,T)$, $\sgn(\lambda_{1,2})=\sgn(\rho)$ and $x_1,x_2\not=0$. 
\begin{enumerate}
\item[(i)]
$v(t,(x_1,x_2)^\top) <  v(t,(x_1,-x_2)^\top)$ if and only if the portfolio $x$ is well diversified;\footnote{Similarly as before (cf.~Definition~\ref{DefWellDiv}), we call a portfolio $x=(x_1,x_2)^\top$ ($x_1,x_2\not=0$) \emph{well diversified} if either the signs of the positions are equal ($\sgn(x_1) = \sgn(x_2)$) and $\rho,\lambda_{1,2}<0$ or if the signs of the positions are different and $\rho,\lambda_{1,2}>0$. Otherwise, the portfolio is \emph{poorly diversified}.} \\
$v(t,(x_1,x_2^\top)) >  v(t,(x_1,-x_2)^\top)$ if and only if the portfolio $x$ is poorly diversified.
\item[(ii)]
If $x$ is poorly diversified, then 
$
\sgn(X^*_i(s-) - \eta_i^*(s,X^*(s-))) \neq \sgn(x_i)
$ for $s \in [t,\tau)$, where 
\[
\tau=\inf\{s\geq t| \sgn(X_i^*(s))\not= \sgn(x_i) \text{ or } X^*_i(s)=0 \text{ for some } i=1,2\} \wedge T >t\quad \text{a.s.}
\]
\item[(iii)]
If $x$ is well diversified, then 
$
\sgn(X^*_i(s-) - \eta_i^*(s,X^*(s-))) = \sgn(x_i)
$ for $s \in [t,\tau)$, where $\tau$ is as above.
\end{enumerate}
\end{prop}

For well diversified portfolios, we can recover only a part of Proposition~\ref{PropPropertiesNeverShort}~(i): by (iii), it is not optimal to change the sign of the position by placing oversized orders in the dark pool. However, it can be optimal to turn a well diversified portfolio into a poorly diversified portfolio by trading in the primary exchange. We illustrate this in the left picture of Figure~\ref{FigCross}; in the displayed scenario, the trader holds a positive position in the second asset only. The trading intensity of the first asset is denoted by the solid line. If $\lambda_{1,2}=0$ it is optimal to sell stocks in the first asset which is positively correlated to the second asset (cf.~the results of Section~\ref{SubSecDpendenceRhon=2}, in particular the discussion following Proposition~\ref{PropMatrixEntries}). If $\lambda_{1,2}$ increases, it becomes profitable to buy stocks in the first asset: as the cross price impact is positive, this allows the trader to increase her trading intensity in the second asset (dashed line) without increasing the overall impact costs too much. In parallel, the trader places a sell order for the first asset in the dark pool as she does not want to sell the stocks at the exchange again (which would result in additional price impact costs due to the sign of $\lambda_{1,2}$).

\begin{figure}[!ht]
\centering\vspace{4ex}
\begin{tabular}{lll}
\begin{tabular}{l}\begin{overpic}[height=2.7cm, width=4.5cm]{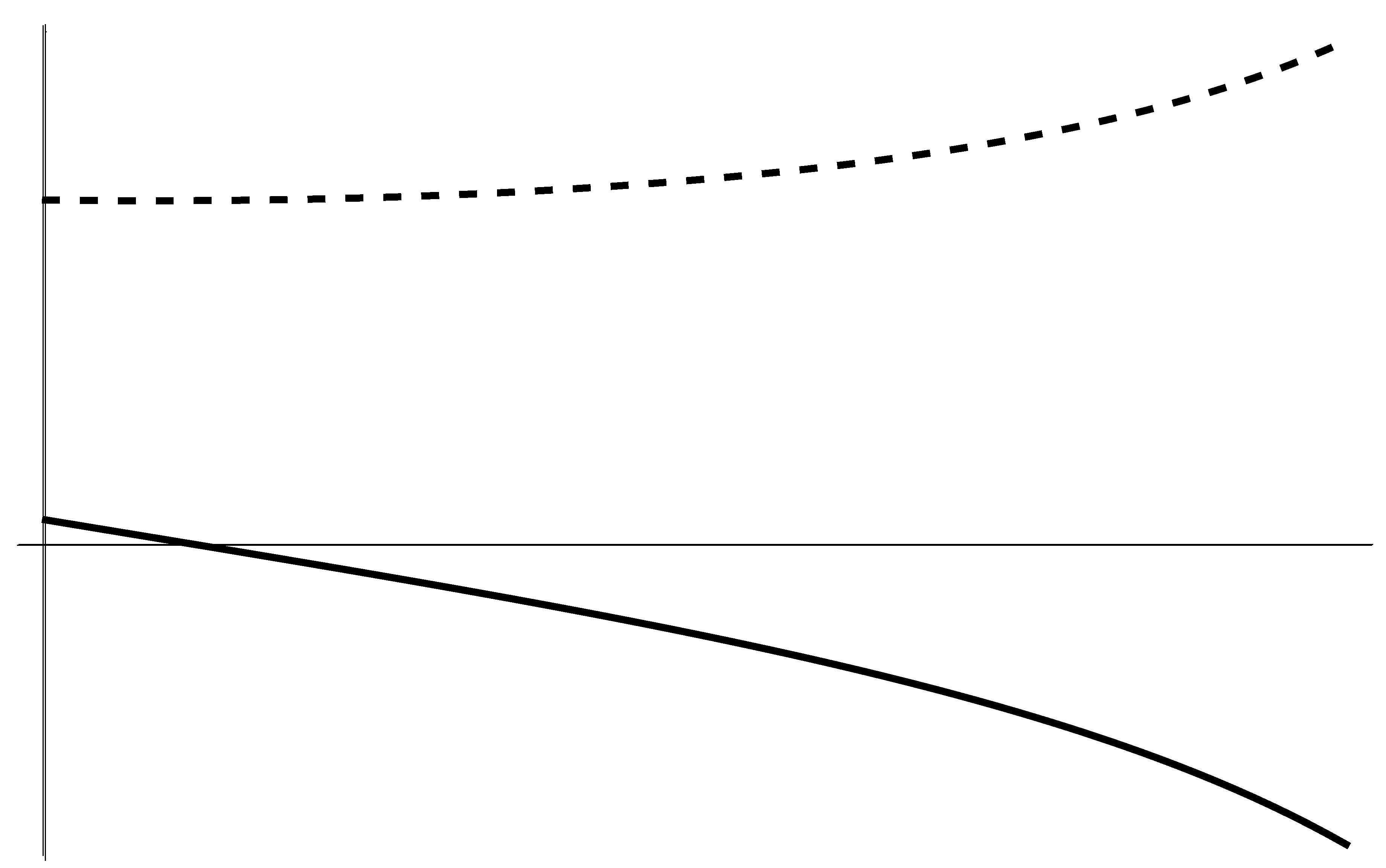}
\put(-9.5,55.5){\scriptsize{$0.59$}}
\put(-9.5,45){\scriptsize{$0.40$}}
\put(5,-0.5){\scriptsize{$-0.35$}}
\put(-9.5,22.5){\scriptsize{$0.03$}}
\put(91.5,16){\scriptsize{$0.7$}}
\put(100,20){\scriptsize{$\lambda_{1,2}$}}
\put(-2,63){\scriptsize{$\xi^*_i(0,x)$}}
\put(5,16){\scriptsize{0}}
\put(2,1){\line(1,0){2}}
\put(2,57){\line(1,0){2}}
\put(2,46.5){\line(1,0){2}}
\put(2,24){\line(1,0){2}}
\put(96,21.3){\line(0,1){2}}
\end{overpic} 
\end{tabular} 
\begin{tabular}{l}
\hspace{1.5ex}
\end{tabular}
\begin{tabular}{l}\begin{overpic}[height=2.7cm, width=4.5cm]{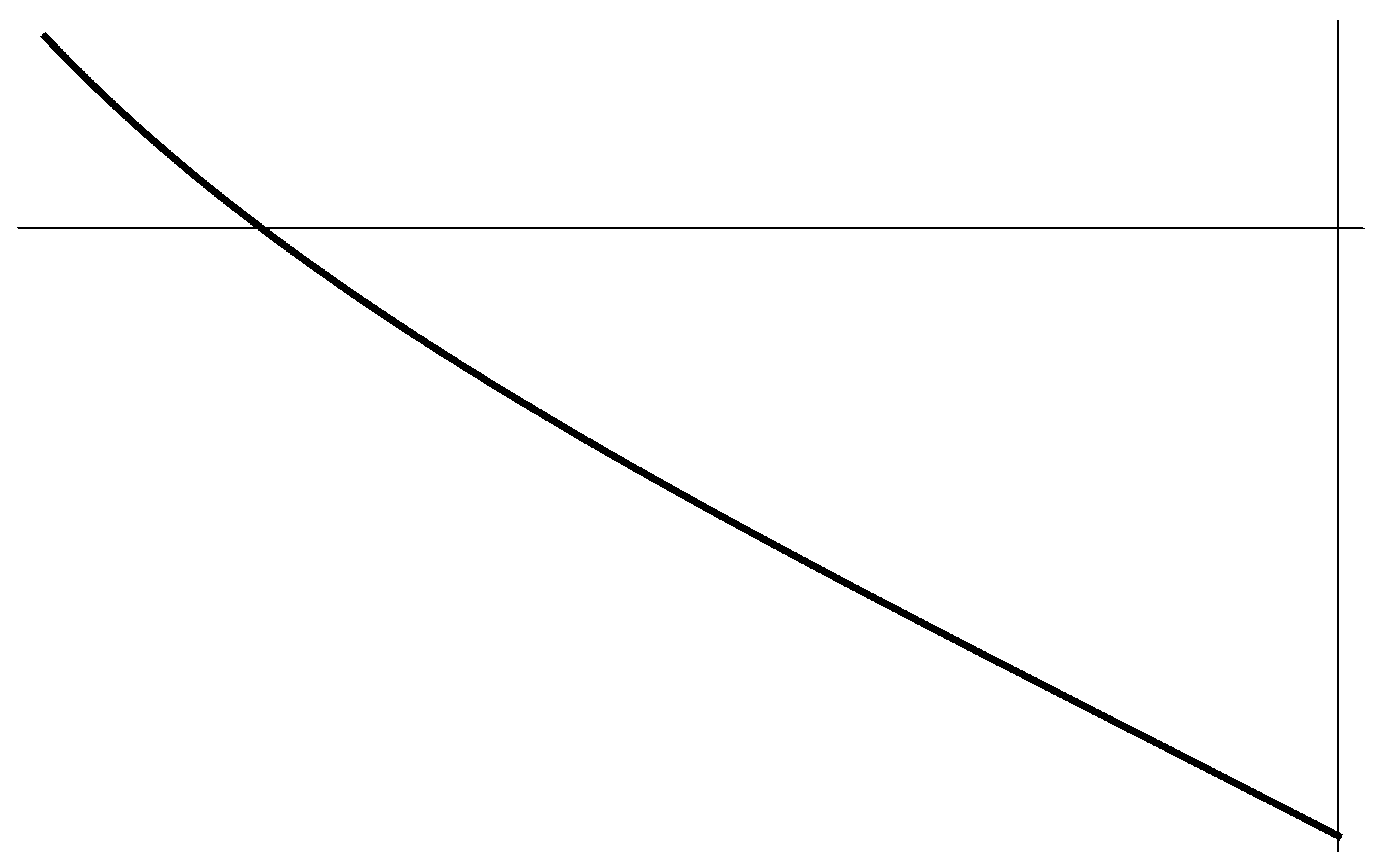}
\put(99.5,56.3){\scriptsize{$1.02$}}
\put(101,43){\scriptsize{$x_1\!\!=\!\!1$}}
\put(99.5,0.3){\scriptsize{$0.93$}}
\put(93,38.5){\scriptsize{$0$}}
\put(-3,38){\scriptsize{$-0.7$}}
\put(-10,43.3){\scriptsize{$\lambda_{1,2}$}}
\put(91,63){\scriptsize{$\eta^*_1(0,x)$}}
\put(96,58){\line(1,0){2}}
\put(96,2){\line(1,0){2}}
\put(3,43.5){\line(0,1){2}}
\end{overpic} 
\end{tabular} 
\begin{tabular}{l}
\hspace{1.5ex}
\end{tabular}
\begin{tabular}{l}\begin{overpic}[height=2.7cm, width=4.5cm]{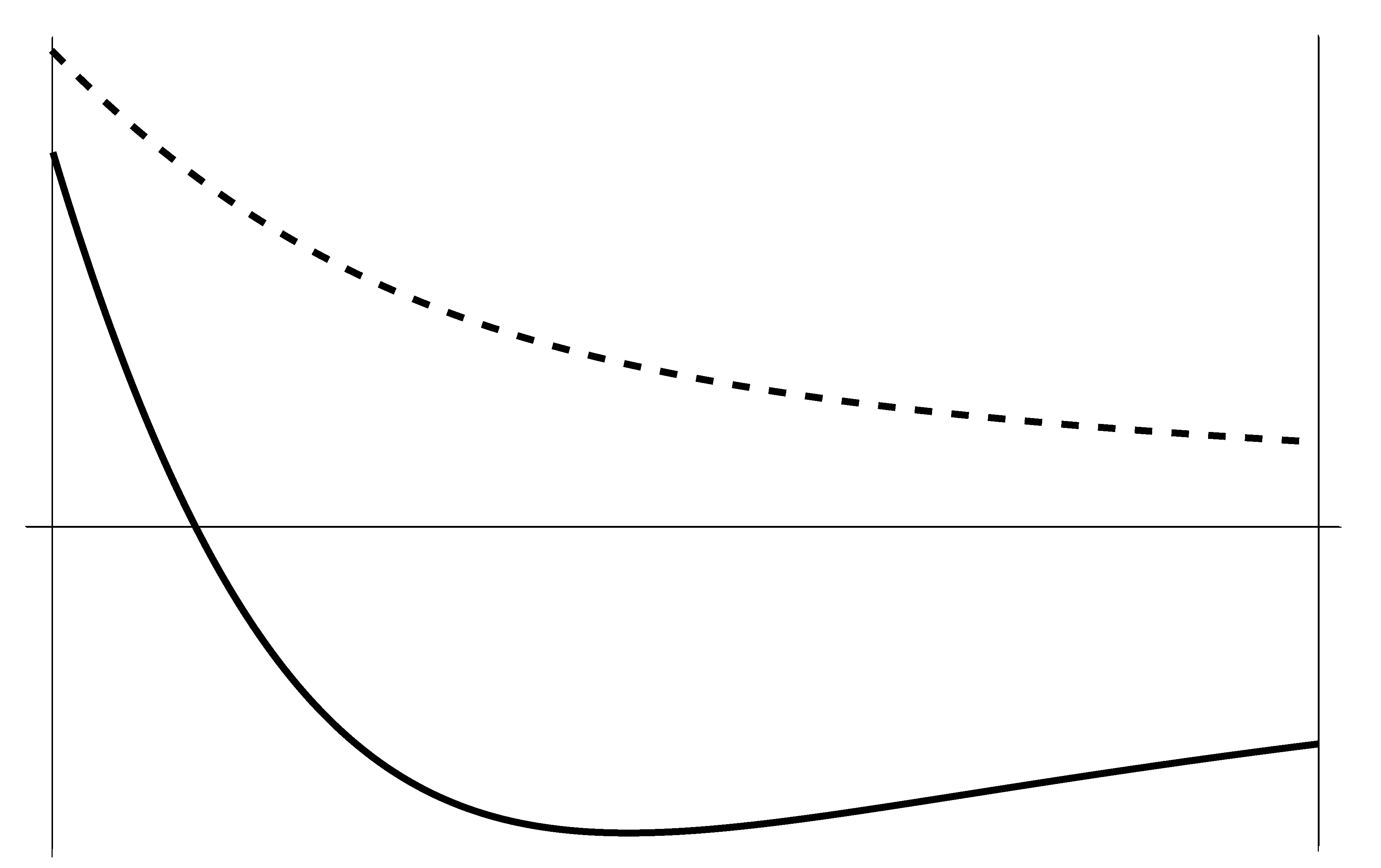}
\put(-10,48){\scriptsize{$0.09$}}
\put(-15,8){\scriptsize{$-0.05$}}
\put(-15,0.5){\scriptsize{$-0.07$}}
\put(5,18.5){\scriptsize{$0$}}
\put(89,18.5){\scriptsize{$10$}}
\put(99,55){\scriptsize{$0.56$}}
\put(99,28.5){\scriptsize{$0.01$}}
\put(100,22.5){\scriptsize{$\theta_1$}}
\put(-2,63){\scriptsize{$\xi^*_1(0,x)$}}
\put(91,63){\scriptsize{$\xi^*_1(0,\tilde{x})$}}
\put(2.5,50){\line(1,0){2}}
\put(2.5,9.5){\line(1,0){2}}
\put(2.5,2){\line(1,0){2}}
\put(95,30){\line(1,0){2}}
\put(95,57){\line(1,0){2}}
\end{overpic} 
\end{tabular} 
\end{tabular}
\vspace{1ex}
\caption{The left picture illustrates the dependence of the optimal trading intensity on $\lambda_{1,2} \in (0,0.7)$ for a portfolio with $x_1=0$, $x_2=1$. The solid line refers to the first asset, the dashed line to the second asset. The middle picture illustrates the dependence of the optimal dark pool order on $\lambda_{1,2} \in (-0.7,0)$ for a portfolio $x_1=1$, $x_2=-1$. in both pictures, $T=1$, $\lambda_1=\lambda_2=\alpha=\sigma_1=\sigma_2=1$, $\theta_1=\theta_2=3$ and $\rho=0.2$; in particular, $\sgn(\rho)=\sgn(\lambda_{1,2})$ for the left picture and $\sgn(\rho)\not=\sgn(\lambda_{1,2})$ for the middle picture. 
The right picture illustrates the dependence of the optimal trading intensity on $\theta_1 \in(0,10)$ for a well diversified portfolio $x$ ($x_1=0.25$, $x_2=-1$, $\rho=0.9$; solid line) and a poorly diversified portfolio $\tilde{x}$ ($\tilde{x}_1=x_1$, $\tilde{x}_2=1$; dashed line). $\lambda_{1,2}=0$; the remaining parameters are as above.}
\label{FigCross}
\end{figure}
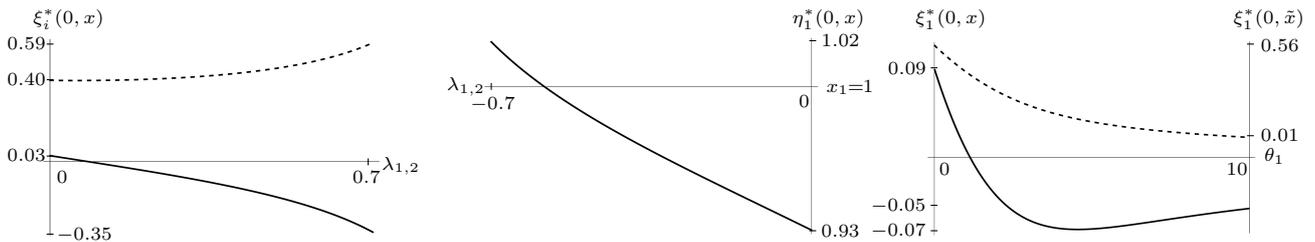

The case $\sgn(\lambda_{1,2})\not=\sgn(\rho)$ (cf.~Footnote~\ref{FootCross}) is more complicated. In this case, it can be optimal to change the sign of a well diversified position (in the sense of Definition~\ref{DefWellDiv}) by placing oversized orders in the dark pool; for large cross price impact, the savings in impact costs resulting from the change of the position can outweigh the increase risk costs. We illustrate this by a numerical example in the middle picture of Figure~\ref{FigCross}.

\subsubsection{Dependence on the execution intensities} \label{SubSecPortfolioTheta}

We conclude this section by analyzing  the dependence of the value function and the optimal strategy on dark pool liquidity (more precisely on the intensities of the Poisson process $\pi$). Similarly as in the single asset case (cf.~Proposition~\ref{PropPropertiesn=1Cont}~(i)), the costs are decreasing in the intensities $\theta_i$. 

\begin{prop} \label{PropMontonIntensity}
Let $t\in [0,T)$, $x \in \mathds{R}^n$ and $i=1,\dots,n$. Then
$v(t,x;\theta_i)$ is decreasing in $\theta_i$.
\end{prop}


The same trade-offs that can cause a non-monotone dependence of the optimal strategy on the price impact parameters can also give rise to a non-monotone dependence on the dark pool liquidity parameters $\theta_i$. We illustrate the dependence of the optimal trading intensity on $\theta_1$ for a poorly (dashed line) and a well diversified portfolio (solid line) in the right picture of Figure~\ref{FigCross}. 

\section{Proofs of the main results} \label{SecProofs}

\subsection{Proof of Theorem~\ref{TheoremDiffBounds}}

For $n\geq 2$, the second summand in the matrix differential equation $C^\top \tilde{C} C$
is in general not linear (or quadratic), and~\eqref{EqODE1} is not a Riccati matrix differential equation. Furthermore,
a closed form solution for the corresponding initial value problem is not known, and the existing theory about Riccati
matrix differential equations is not applicable directly.

It turns out that appropriate upper and lower bounds for the non-linear term $C^\top \tilde{C} C$
transform to lower and upper bounds ($P$ respectively $Q$) for the solution of the Matrix Initial Value Problem~\eqref{EqODE1} and yield existence and positive definiteness of the solution on the whole interval $(-\infty,T]$
(Theorem~\ref{TheoremDiffBounds}). To this end, we require a version of a well-known comparison result for matrix Riccati differential equations, which we state in Appendix~\ref{SecAppRiccati}.
The main step is thus to obtain adequate matrix inequalities which enable us to transfer these results to the
Initial Value Problem~\eqref{EqODE1}. 

For $C>0$, we have
$
0 \leq C \tilde{C} C.
$
The desired upper bound of $C \tilde{C} C$ is a direct consequence of the matrix inequality stated in the following result. 

\begin{prop} \label{PropMatrixIneq}
Let $C=(c_{i,j})_{i,j=1\dots,n} \in \mathds{R}^{n\times n}$ be a positive definite matrix and 
$\theta_i>0$, $i=1,\dots,n$. Then
\begin{equation} \label{Ineq}
C \leq  \theta \diag\Big(\frac{c_{i,i}}{\theta_i}\Big)= \theta \tilde{C}^{-1} ,
\end{equation}
where 
$
\theta:=\theta(n):=\sum_{i=1}^n \theta_i.
$
\end{prop}

\begin{proof}
We prove the inequality by induction on $n$. It is clear for $n=1$ with equality in~\eqref{Ineq}.
Let now $n\geq 1$ and $C=(c_{i,j})_{i,j=1,\dots,n+1} \in \mathds{R}^{(n+1)\times(n+1)}$
be positive definite. Define 
$C_n\in \mathds{R}^{n \times n}$ and $c \in \mathds{R}^n$ such that
\begin{equation*}
C=\begin{pmatrix}
C_n & c \\
c^\top & c_{n+1,n+1}
\end{pmatrix}.
\end{equation*} 
For $z=(x,y)^\top \in \mathds{R}^n\times\mathds{R}$, $z \not= 0$, we have
\[
z^\top C z = x^\top C_n x + 2x^\top c y+ c_{n+1,n+1} y^2, \quad
z^\top \diag\Big(\frac{c_{i,i}}{\theta_i}\Big) z = 
x^\top \diag\Big(\frac{c_{i,i}}{\theta_i} \Big) x + \frac{c_{n+1,n+1}}{\theta_{n+1}}y^2. 
\]
By abuse of notation, $\diag\Big(\frac{c_{i,i}}{\theta_i} \Big)$ is used both for the diagonal $n\times n$ - matrix with $\frac{c_{1,1}}{\theta_1},\dots ,\frac{c_{n,n}}{\theta_n}$ in the diagonal and for the respective diagonal $(n+1)\times (n+1)$ - matrix. Which one we refer to is always clear from the context. 
\begin{align}
& z^\top \Big(\theta(n+1) \diag\Big(\frac{c_{i,i}}{\theta_i}\Big) - C\Big) z \notag \\*
&\qquad = x^\top \Big( \theta(n) 
	\diag\Big(\frac{c_{i,i}}{\theta_i}\Big)   
	- C_n \Big) x^\top + \theta_{n+1} x^\top \diag\Big(\frac{c_{i,i}}{\theta_i}\Big) x -2x^\top c y 
	+ \frac{\theta(n)}{\theta_{n+1}} c_{n+1,n+1} y^2 \label{EqMatrixIneq1}. 
\end{align}
The first summand in Equation~\eqref{EqMatrixIneq1} is non-negative by the induction hypothesis. The remainder equals
\begin{align}
& \Big(x-\frac{1}{\theta_{n+1}} \diag\Big(\frac{\theta_i}{c_{i,i}}\Big)c y \Big)^\top
	\theta_{n+1}\diag\Big(\frac{c_{i,i}}{\theta_i} \Big) \Big(x-\frac{1}{\theta_{n+1}} 								\diag\Big(\frac{\theta_i}{c_{i,i}} \Big)c y \Big) \notag \\*
& \quad  + \Big( c_{n+1,n+1} \frac{\theta(n)}{\theta_{n+1}} 
	-c^\top \frac{\diag\big(\frac{\theta_i}{c_{i,i}} \big)}{\theta_{n+1}} c \Big) y^2, \notag 
\end{align}
where the first summand is nonnegative as $C_n$ (and therefore 
$\diag\big(\frac{c_{i,i}}{\theta_i} \big)$) is positive definite and $\theta_{n+1}>0$. 
We have (see, e.g., the book by~\cite{HornJohnson1985}, Corollary 7.7.4\footnote{For matrices $A,B$ with $0<A<B$, we have $0<B^{-1}<A^{-1}$. \label{FootnoteHornJohnson}})
\begin{equation}
0< \frac{1}{\theta(n)} \diag\Big(\frac{\theta_i}{c_{i,i}} \Big) \leq C_n^{-1} \label{Ineq3}
\end{equation}
by the induction hypothesis. Moreover
\begin{equation}
C\begin{pmatrix}
I_{n\times n} & -C_n^{-1} c \\
0 & 1
\end{pmatrix}
= \begin{pmatrix}
C_n & 0 \\
c^\top & c_{n+1,n+1}-c^\top C_n^{-1}c
\end{pmatrix} \notag
\end{equation}
and hence 
\begin{equation}
c_{n+1,n+1}-c^\top C_n^{-1} c = \frac{\det C}{\det{C_n}} >0. \label{4}
\end{equation}
Finally,
\begin{equation*}
c_{n+1,n+1} -\frac{1}{\theta(n)} 
	c^\top \diag\Big(\frac{\theta_i}{c_{i,i}} \Big) c 
	\overset{\eqref{Ineq3}}{\geq} c_{n+1,n+1} -c^\top C_n^{-1} c \overset{\eqref{4}}{>} 0,
\end{equation*}
finishing the proof.
\end{proof}

Applications of~\cite{HornJohnson1985}, Corollary 7.7.4 (cf.~Footnote~\ref{FootnoteHornJohnson} again), imply the desired bound for $C \tilde{C}C$.  Additionally, we obtain two elementary matrix inequalities.

\begin{corollary} \label{CorMatIneq}
Let $C>0$ and $\theta_i \geq 0$, $i=1,\dots ,n$. Then
\begin{equation} \label{IneqMatrixIneq2}
C \tilde{C} C \leq \theta C,
\end{equation}
\[ 
C \leq n \diag \big( c_{i,i} \big) \quad \text{and} \quad
C \leq \trace(C) I.
\]
\end{corollary}

The Matrix Inequality~\eqref{IneqMatrixIneq2} enables us to apply Theorem~\ref{TheRiccIneq} to the Matrix Initial Value Problem~\eqref{EqODE1} such that we can prove existence of a solution $C$ of~\eqref{EqODE1} on the whole interval $(-\infty,T]$ and
at the same time construct upper and lower bounds for $C$ via the solutions of the initial value problems in~\eqref{IVPPQ}; these are given explicitly by
\begin{equation}
P(l,t) = p(l,t) I, \quad
Q(l,t) = q(l,t) I, \label{PQlt}
\end{equation}
where
\begin{align}
p(l,t) &:= \sqrt{\tfrac{\theta^2}{4}+ \alpha d_{\min}}
	\coth \Big( \sqrt{\tfrac{\theta^2}{4}+ \alpha d_{\min}}(T-t)+ \kappa_1(l)\Big)
	-\frac{\theta}{2}, \label{plt}\\
q(l,t) &:= \sqrt{\alpha d_{\max}}
	\coth\Big( \sqrt{ \alpha d_{\max}}(T-t)+ \kappa_2(l)\Big) \label{qlt}
\end{align}
for $\theta+ \alpha d_{\min} >0$ respectively $\alpha d_{\max}>0$ with
\[
\kappa_1(l):=\arcoth\Big( \frac{ \frac{l}{\lambda_{\max}}+ \frac{\theta}{2}}
	{\sqrt{\frac{\theta^2}{4}+ \alpha d_{\min}}} \Big)>0, \quad
\kappa_2(l):=\arcoth\Big( \frac{ \frac{l}{\lambda_{\min}}}
	{ \sqrt{\alpha d_{\max}}} \Big)>0
\]
and
\begin{equation}
p(l,t) := \frac{1}{T-t+\frac{\lambda_{\max}}{l}}, \quad
q(l,t) := \frac{1}{T-t+\frac{\lambda_{\min}}{l}} \label{pqlttheta0}
\end{equation}
for $\theta= \alpha d_{\min} =0$ respectively $\alpha d_{\max}=0$. 
Note also that
$0<p(l,t),q(l,t)<\infty$ for all $t\in(-\infty,T]$.

\begin{proof}[ Proof of Theorem~\ref{TheoremDiffBounds}]
Let $C(l,t)$ be a solution of~\eqref{EqODE1} on some interval
$(t_1,T]$; note that there exists a local solution by the Picard-Lindel\"{o}f theorem. The symmetry of 
$\Lambda$, $\Sigma$ and the initial value $C(l,T)=lI$ imply 
that $C(l,t)$ is symmetric on $(t_1,T]$.

Let now $\hat{P}:=\Lambda P$ for $P$ as in~\eqref{PQlt}. Then
$\hat{P}(l,t)$ solves
$
\hat{P}' = \hat{P} \Lambda^{-1} \hat{P}+ \theta \hat{P}-\alpha d_{\min} \Lambda$, 
$\hat{P}(T) = \frac{l}{\lambda_{\max}}\Lambda
$
on $(-\infty,T]$. 
As $P>0$ and $P$ commute with $\Lambda$, we have
$
\hat{P}(l,t) > 0.
$
Assume that
\begin{equation} \label{Asstau}
\{t \in (t_1,T] \,|\,C(l,t) \text{ is not positive definite}\} \not= \emptyset
\end{equation}
and define
$
\tau:=\sup \{t \in (t_1,T] \,|\, C(l,t) \text{ is not positive definite}\}.
$
As $C(l,T)=lI>0$ and $C(l,\cdot)$ is continuous, there exists an $\epsilon>0$ such that $C(l,t)>0$
for $t \in (T-\epsilon,T]$ and thus $\tau <T$. We apply Theorem~\ref{TheRiccIneq} to $\bar{P}:=-\hat{P}$
and $\bar{C}:=-C$ on $[\tau,T]$. We have
\begin{equation*}
\bar{P}(l,T)=-\frac{l}{\lambda_{\max}} \Lambda \geq -lI = \bar{C}(l,T)
\end{equation*}
and
\begin{align}
\bar{P}'\!=\!-\bar{P} \Lambda^{-1} \bar{P} +\theta \bar{P} +\alpha d_{\min} \Lambda, \quad
\bar{C}'\!=\!-\bar{C} \Lambda^{-1} \bar{C} +\bar{C}\tilde{\bar{C}}\bar{C}+\alpha \Sigma =-\bar{C} \Lambda^{-1} \bar{C} +\theta \bar{C} 
	+\big(\alpha \sqrt{\Lambda} D \sqrt{\Lambda} + \bar{C}\tilde{\bar{C}}\bar{C}-\theta \bar{C}\big). \notag
\end{align}
Let now $x\in \mathds{R}^n$. Applying Corollary~\ref{CorMatIneq} to
$-\bar{C}$, we obtain
\[
x^\top \Big( \alpha \sqrt{\Lambda} D \sqrt{\Lambda} + \bar{C}\tilde{\bar{C}}\bar{C}- \theta \bar{C} 
-\alpha d_{\min} \Lambda \Big) x 
 =\alpha x^\top\big( \sqrt{\Lambda}(D-d_{\min}I)\sqrt{\Lambda} \big) x +
x^\top\big(\bar{C}\tilde{\bar{C}}\bar{C}-\theta \bar{C}\big)(t) x\geq 0.
\]
As $\Lambda >0$, 
Theorem~\ref{TheRiccIneq} implies
$
\bar{C}(l,t) \leq \bar{P}(l,t)
$
and therefore
$
0 < \hat{P}(l,t) \leq C(l,t)
$
on $(\tau,T]$. By continuity of $C(l,\cdot)$, we have
$
0 < \hat{P}(l,\tau) \leq C(l,\tau)
$
and thus $C(l,t)>0$ in some neighborhood of
$\tau$, a contradiction to Assumption~\eqref{Asstau}. Hence, $C(l,t)$ is positive definite on the whole interval
$(t_1,T]$. Applying Theorem~\ref{TheRiccIneq}
in the same way as above again, yields that we may choose $t_1=-\infty$
and that
$
0 < \hat{P}(l,t) \leq C(l,t)
$
on $(-\infty,T]$. A similar argument establishes $Q$ as an upper bound by using $0 \leq C\tilde{C}C$ instead of Inequality~\eqref{IneqMatrixIneq2}. 
\end{proof}

\subsection{Proof of Proposition~\ref{PropAdml}} \label{ProofAdml}

Before we begin, we introduce the following notation. 

\begin{nota} \label{NotationTau}
\begin{enumerate}
\item[(i)]
We denote the jump times of $\pi$ by 
$(\tau_j)_{j \in \mathds{N}}$, where $\tau_j < \tau _{j+1}$ ($j \in \mathds{N}$) almost surely (with the convention $\tau_0=t$).
\item[(ii)]
Given the Markovian control $u^*(l)$, the Stochastic Differential Equation~\eqref{EqCSDE} possesses a
unique solution.
We denote the process controlled by $u^*(l)$ by 
\begin{equation} \notag
X^*(l,s):=X^{u^*(l)}(s).
\end{equation}
\end{enumerate}
\end{nota}

In order to prove admissibility of $u^*(l)$, we show that
$\|X^*(l,\cdot )\|_2$ is bounded by using Gronwall's inequality pathwise inductively on 
the time-intervals 
$[\tau_i \wedge T,\tau_{i+1}\wedge T)$
and interlacing the
jumps (cf.~Remark~\ref{RemarkInterJump}). This can be achieved by applying the upper and lower bounds of $C(l,s)$ from Theorem~\ref{TheoremDiffBounds}. 

\begin{lem} \label{LemXl*}
Let $l>l_0$ for $l_0$ as in Equation~\eqref{lNull}, $t \in [0,T)$, $x \in \mathds{R}^n$ be the portfolio position at time $t$ and $\theta=\sum_i\theta_i$ as before. Then the following hold.
\begin{enumerate}
\item[(i)]
For $s \in [t,T)$,
\[
X^*(l,s)^\top C(l,s) X^*(l,s) \leq \exp(\theta (s-t)) x^\top C(l,t) x
\quad \text{a.s.}
\]
\item[(ii)]
There exists a constant $K$ independent of $l$ such that for all $s \in [t,T)$,
$\| X^*(l,s) \|_2 \leq K$ a.s.

\end{enumerate}
\end{lem}

\begin{proof}
\begin{enumerate}
\item[(i)]
Let $i \in \mathds{N}$. On $\{\tau_i <T\}$, 
$X^*(l,\cdot)$ solves the initial value problem
\begin{equation*}
X'=-\Lambda^{-1} C(l) X, \quad
X(\tau_i)=X^*(l,\tau_i)
\end{equation*}
for $s \in [\tau_i,\tau_{i+1} \wedge T)$. Hence, as $C(l)$ solves the Initial Value Problem~\eqref{EqODE1},
\begin{align}
& \frac{\partial}{\partial s}\big( X^*(l,s)^\top C(l,s) X^*(l,s) \big) \notag \\*
& \qquad = \frac{\partial}{\partial s} X^*(l,s)^\top C(l,s) X^*(l,s) +
 X^*(l,s)^\top  \frac{\partial}{\partial s} C(l,s) X^*(l,s)+
  X^*(l,s)^\top C(l,s)  \frac{\partial}{\partial s} X^*(l,s) \notag \\
& \qquad= -X^*(l,s)^\top C(l,s)\Lambda^{-1} C(l,s) X^*(l,s) + X^*(l,s)^\top C(l,s)\Lambda^{-1} C(l,s) X^*(l,s)  \notag \\*
& \qquad \qquad +
 X^*(l,s)^\top C(l,s) \tilde{C}(l,s) C(l,s) X^*(l,s) - \alpha  X^*(l,s)^\top \Sigma X^*(l,s)  \notag \\*
& \qquad \qquad 
 -X^*(l,s)^\top C(l,s)\Lambda^{-1} C(l,s) X^*(l,s)  \notag \\
 & \qquad \leq \theta X^*(l,s)^\top C(l,s) X^*(l,s) \notag
\end{align}
by Corollary~\ref{CorMatIneq} and the fact that $C,\Lambda^{-1},\Sigma \geq 0$. By Gronwall's inequality, this implies
\begin{equation} \label{Gronwall1}
X^*(l,s)^\top C(l,s) X^*(l,s) \leq \exp(\theta (s-\tau_i)) x^\top C(l,t) x.
\end{equation}
Now, on $\{\tau_{i+1}<T\}$, there exits an (almost surely) unique $j=1,\dots,n$ (cf.~Assumption~\ref{AssDP}~(ii)) such that
\[
X^*(l,\tau_{i+1})=X^*(l,\tau_{i+1}-)-\eta_j^*(l,\tau_{i+1},X^*(l,\tau_{i+1}-)) e_j.
\]
Let $\eta \in \mathds{R}$. Then 
\begin{align*}
& (X^*(l,\tau_{i+1}-)-\eta e_j)^\top C(l,\tau_{i+1}) (X^*(l,\tau_{i+1}-)-\eta e_j) \notag \\*
& \qquad =X^*(l,\tau_{i+1}-)^\top C(l,\tau_{i+1}) X^*(l,\tau_{i+1}-) +\eta^2 e_j^\top C(l,\tau_{i+1}) e_j -2 \eta e_j^\top C(l,\tau_{i+1}) X^*(l,\tau_{i+1}-) \\�
	& \qquad = X^*(l,\tau_{i+1}-)^\top C(l,\tau_{i+1}) X^*(l,\tau_{i+1}-) + \Big( \eta \sqrt{c_{j,j}(l,\tau_{i+1})} - \frac{1}{\sqrt{c_{j,j}(l,\tau_{i+1})}} e_j^\top C(l,\tau_{i+1})X^*(l,\tau_{i+1}-)\Big)^2 \notag \\*
&\qquad \qquad 	-\frac{1}{c_{j,j}(l,\tau_{i+1})} (e_j^\top C(l,\tau_{i+1}) X^*(l,\tau_{i+1}-) )^2, 
\end{align*}
which attains its (unique) minimum in $\eta$ for $\eta=\eta_j^*(l,\tau_{i+1},X^*(l,\tau_{i+1}-))$. Hence,
\begin{align}
& \big(X^*(l,\tau_{i+1}-)-\eta_j^*(l,\tau_{i+1},X^*(l,\tau_{i+1}-)) e_j\big)^\top
C(l,\tau_{i+1}) \big( X^*(l,\tau_{i+1}-)-\eta_j^*(l,\tau_{i+1},X^*(l,\tau_{i+1}-)) e_j\big) \notag \\*
&\qquad = \min\limits_{\eta \in \mathds{R}} 
\big( X^*(l,\tau_{i+1}-)-\eta e_j\big)^\top
C(l,\tau_{i+1}) 
\big( X^*(l,\tau_{i+1}-)-\eta e_j\big) \notag
\end{align}
and therefore
\begin{equation}
X^*(l,\tau_{i+1})^\top C(l,\tau_{i+1}) X^*(l,\tau_{i+1}) \leq X^*(l,\tau_{i+1}-)^\top C(l,\tau_{i+1})X^*(l,\tau_{i+1}-).
	\label{IneqStattBellman} 
\end{equation}
Using Inequalities~\eqref{Gronwall1} and~\eqref{IneqStattBellman} inductively, we obtain the assertion.
\item[(ii)]
By Theorem~\ref{TheoremDiffBounds}, we have
$\Lambda \leq 1/p(l,s) C(l,s)$
and hence,
\begin{align}
X^*(l,s)^\top \Lambda X^*(l,s) &\leq  \frac{1}{p(l,s)} X^*(l,s)^\top C(l,s) X^*(l,s) \label{IneqAbsch} \\
& \leq  \frac{1}{p(l_0,s)} X^*(l,s)^\top C(l,s) X^*(l,s) \notag \\ 
&\leq \tilde{K} x^\top C(l,t) x \notag
\end{align}
for a constant $\tilde{K}$ independent of $s$
by~(i) and the fact that $p(l_0)$ attains its minimum in $[t,T)$. The assertion follows as $q(t):= \lim_{l \rightarrow \infty} q(l,t) <\infty$ exists (cf.~Equations~\eqref{qlt} and~\eqref{pqlttheta0}, cf.~also Lemma~\ref{LemBounds} below) and
$
C(l,t) \leq \lambda_{\max} q(t)
$
(cf.~Theorem~\ref{TheoremDiffBounds}).
\end{enumerate}
\end{proof}

The bound obtained in Lemma~\ref{LemXl*} enables us to prove that $u^*(l)$ fulfills the moment conditions 
of Definition~\ref{DefAdmStr}~(ii) and hence Proposition~\ref{PropAdml}.

\begin{proof}[Proof of Proposition~\ref{PropAdml}]
Definition~\ref{DefAdmStr}~(i) and~(iii) are clearly satisfied. 

Let $\|\cdot \|_{2,2}$ denote the matrix norm induced by the space 
$(\mathds{R}^n, \|\cdot \|_2)$. Note that $\|\cdot\|_{2,2}$ is the spectral norm on $\mathds{R}^{n \times n}$
and therefore (see, e.g., \cite{Bernstein}, Theorem 8.4.9)
$
\|A\|_{2,2} \leq \|B\|_{2,2}$ for $0\leq A\leq B$.  
Using Theorem~\ref{TheoremDiffBounds} and~\ref{LemXl*}~(ii), we deduce
\[
\mathbb{E} \Big[\int_t^T\ \|\xi^*(l,s,X^*(s))\|_2^4 ds \Big]
 \leq \mathbb{E} \Big[\int_t^T\ \|\Lambda^{-1}\|_{2,2}^4 \|C(l,s) \|_{2,2}^4 
	\| X^*(l,s)\|_2^4 ds \Big] \notag \\
\leq \mathbb{E}\Big[\int_t^T\ q(l,s)^4 K^4 ds \Big] < \infty, \notag 
\]
for $K$ as in Lemma~\ref{LemXl*}~(ii).
Similarly,
\begin{align*}
\mathbb{E} \Big[\int_t^T\ \|\eta^*(l,s,X^*(s-))\|_2^8 ds \Big]
& \leq \mathbb{E} \Big[\int_t^T\ 
	\|\!\!\!\!\!\!\!\!\!\!\!\!\!\!\!\!\!\!\!\!\!\!\!\underbrace{\bar{C}(l,s)}
	_{\quad \quad \quad\leq n C(l,s)^{-1} \text{ by Corollary~\ref{CorMatIneq}}\!\!\!\!\!\!\!\!\!\!\!\!\!\!\!\!\!\!\!\!\!\!\!}
	\!\!\!\!\!\!\!\!\!\!\!\!\!\!\!\!\!\!\!\!\!\!\|_{2,2}^8 \|C(l,s) \|_{2,2}^8
	\| X^*(l,s-)\|_2^8 ds \Big] \notag \\
&\leq \mathbb{E}\Big[\int_t^T\ \frac{n^8 \lambda_{max}^8 q(l,s)^8}{\lambda_{\min}^8 p(l,s)^8} K^8 ds \Big] < \infty.
\end{align*}
\end{proof}

\subsection{Proof of Theorem~\ref{TheoremOptStrl}}

We first require the following moment estimate for the controlled process.

\begin{lem} \label{LemXuintegrierbar}
Let $t\in [0,T)$, $x \in \mathds{R}^n$ and $u \in \tilde{\mathbb{A}}(t)$. Then
\begin{equation} \notag
\mathbb{E}\Big[ \sup\limits_{t \leq s \leq T} \|X^u(s)\|_2^4 \Big] < \infty, \quad \text{in particular} \quad
\mathbb{E}\Big[ \int_{t}^T \|X^u(s)\|_2^4 ds \Big] < \infty.
\end{equation}
\end{lem}

\begin{proof}
Let $s \in [t,T]$. Then by H\"older's inequality and a multi-dimensional 
version of Jensen's inequality 
(see, e.g., \cite{Kallenberg}, Lemma 3.5) we obtain
\begin{align}
\| X^u(s) \|_2^4 
& \leq \bigg( \|x\|_2 + \Big\| \int_t^s \xi(r)dr \Big\|_2
	+\Big\| \int_t^s \eta(r) d \pi(r) \Big\|_2 \bigg)^4\notag \\
& \leq 27 \bigg( \|x\|_2^4 + (s-t)^{3} \int_t^s \|\xi(r) \|_2^4 dr
	+\Big\| \int_t^s \eta(r) d \pi(r) \Big\|_2^4 \bigg). \label{IneqJensen1} 
\end{align}
By Definition~\ref{DefAdmStr}~(ii), it is hence sufficient to consider the last summand of Equation~\eqref{IneqJensen1}.
To this end we define the compensated Poisson processes 
$M_i(s):=\pi_i(s)-\theta_i s, \quad i=1,\dots,n$.
We note that Poisson distributed random variables have finite moments and apply It\^{o}'s isometry (note that $\langle M_i\rangle(s)=\theta_i s$) and H\"older's inequality to obtain
\begin{align} 
\mathbb{E}\bigg[ \Big| \int_t^s \eta_i(r) d \pi_i(r) \Big|^4 \bigg] 
&\leq \mathbb{E}\bigg[ \big(\pi_i(s)-\pi_i(t)\big)^3
	\Big( \int_t^s |\eta_i(r)|^4 d M_i(r) 
	+ \theta_i \int_t^s |\eta_i(r)|^4 d r \Big) \bigg] \notag \\
&\leq \mathbb{E} \bigg[ \big(\pi_i(s)-\pi_i(t)\big)^6\bigg]^{\frac{1}{2}} \bigg(
	\mathbb{E}\bigg[\Big(\int_t^s |\eta_i(r)|^4 d M_i(r)\Big)^2\bigg]^{\frac{1}{2}}
 	+\mathbb{E} \bigg[ \Big(\theta_i \int_t^s |\eta_i(r)|^4 d r\Big)^2 \bigg]^{\frac{1}{2}} \bigg)
 	\notag \\
& < K_i \notag
\end{align}
for a constant $K_i$ independent of $s$ by Definition~\ref{DefAdmStr}~(ii), which finishes the proof.
\end{proof}

\begin{proof}[Proof of Theorem~\ref{TheoremOptStrl}]
Let $l>l_0$, $(t,x) \in [0,T)\times \mathds{R}^n$ and $u=(\xi,\eta) \in \tilde{\mathbb{A}}(t)$. 
We apply It\^{o}'s formula (see, e.g., the book by~\cite{OksendalSulem}) to the function 
$w(l,t,X^u(t))=X^u(t)^\top C(l,t) X^u(t)$.
\begin{align}
w(l,t,x) 
 &= w(l,T,X^u(T)) + \int_t^T \nabla_x w(l,s,X^u(s)) \xi(s) 
		-\frac{\partial{w}}{\partial s} (l,s,X^u(s))ds \notag \\*
	& \qquad +\int_t^T \Big(\sum\limits_{i=1}^n  w(l,s,X^u(s-))
		-w\big(l,s, X^u(s-)- \eta_i(s) e_i\big) \Big) \pi_i(ds) \notag \\
& \leq w(l,T,X^u(T)) +\int_t^T f(X^u(s),\xi(s))ds \notag \\*
	& \qquad +\int_t^T \Big(\sum\limits_{i=1}^n w(l,s,X^u(s-))
		-w\big(l,s, X^u(s-)- \eta_i(s) e_i\big) \Big)\pi_i(ds)  \notag \\*
	&  \qquad -\int_t^T \Big( \sum\limits_{i=1}^n \theta_i \big( w(l,s,X^u(s-))
		-w\big(l,s, X^u(s-)- \eta_i(s) e_i \big) \Big) ds \label{InEqHJB} 	\\
& =  w(l,T,X^u(T)) + \!\int_t^T \!f(X^u(s),\xi(s))ds + \!\sum\limits_{i=1}^n \int_t^T 
		\!\!\Big(w(l,s,X^u(s-))-w\big(l,s, X^u(s-)\!-\!\eta_i(s) e_i\big) \Big) M_i(ds),  \notag
\end{align}
where $M_i$ is the compensated Poisson process $M_i(s):=\pi_i(s)-\theta_i s$ and Inequality~\eqref{InEqHJB} follows from Proposition~\ref{PropSolHJB}. Furthermore, we have (pathwise) equality in \eqref{InEqHJB} if and only if $u=u^*$ $\lambda$ - a.s. 

Taking expectations on both sides, we obtain
\begin{equation}
w(l,t,x) \leq \tilde{J}(l,t,x,u) + \!\sum\limits_{i=1}^n\mathbb{E}\Big[
	\int_t^T \Big(\! w(l,s,X^u(s-))-w\big(l,s, X^u(s-)-\eta_i(s) e_i\big)\!\Big) M_i(ds) \Big], 
	\label{IneqItoIntMart}
\end{equation}
with equality if and only if $u=u^*$ $\mathbb{P} \otimes \lambda$ - a.s. 

It remains to show that the stochastic integrals in Inequality~\eqref{IneqItoIntMart}
are martingales. To this end, we compute  
\begin{align}
& \mathbb{E}\Big[\int_t^T |w(l,s,X^u(s-))
	-w\big(l,s, X^u(s-)-\eta_i(s) e_i\big)|^2 ds \Big] \notag \\*
&\qquad =\mathbb{E}\Big[\int_t^T  |2X^u(s-)^\top C(l,s) e_i\eta_i(s) 
	-\eta_i(s)^2 c_{i,i}(s,l)|^2 ds\Big] \notag \\
&\qquad \leq \mathbb{E}\Big[\int_t^T  2|2X^u(s-)^\top C(l,s) e_i\eta_i(s)|^2 ds\Big] 
	+\mathbb{E}\Big[\int_t^T  2|\eta_i(s)^2 c_{i,i}(l,s)|^2ds\Big]  
	\notag \\
&\qquad \leq 8 \mathbb{E} \Big[ \int_t^T  \| X^u(s-)^\top \|_2^2 \| C(l,s) \|_{2,2}^2 
	|\eta_i(s) |^2 ds \Big] 
	+ 2 \mathbb{E}\Big[\int_t^T  | \eta_i(s)|^4 |c_{i,i}(l,s)|^2 ds\Big]  
	\notag \\
&\qquad \leq 8 \lambda_{\max}^2 \big(\max\limits_{s \in [t,T]} q(l,s)^2\big)
	\mathbb{E}\Big[\int_t^T  \| X^u(s-) \|_2^4  ds\Big]^{\frac{1}{2}}
	\mathbb{E}\Big[\int_t^T  \| \eta(s)\|_2^4 ds\Big]^{\frac{1}{2}} 
	\notag \\*
& \qquad \qquad  + 2 \lambda_{\max}^2 \big(\max\limits_{s \in [t,T]} q(l,s)^2\big)
	\mathbb{E}\Big[\int_t^T  \|\eta(s)\|_2^4 ds\Big] 
	\label{IneqItoMartHolder} \\*
&\qquad <\infty \notag
\end{align}
by Definition~\ref{DefAdmStr}~(ii) and Lemma~\ref{LemXuintegrierbar}, where Inequality~\eqref{IneqItoMartHolder} follows from H\"{o}lder's inequality. As $\langle M_i\rangle (s) = \theta_i s$, this finishes the proof.
\end{proof}

\subsection{Proof of Theorem~\ref{TheoremLimitValueFunc}}

We start by computing the limits of the functions $p(l)$ and $q(l)$ given
by Equations~\eqref{plt} and~\eqref{qlt} respectively~\eqref{pqlttheta0}.

\begin{lem} \label{LemBounds}
Let $t \in [0,T)$ and $p(l)$ and $q(l)$ as
in Equations~\eqref{plt} and~\eqref{qlt} respectively~\eqref{pqlttheta0}. Then
\begin{equation*}
\lim\limits_{l\rightarrow \infty} p(l,\cdot) =p(\cdot), \quad
\lim\limits_{l\rightarrow \infty} q(l,\cdot) =q(\cdot),
\end{equation*}
where $p$ and $q$ are given by
\begin{align}
p(t) &:= \sqrt{\tfrac{\theta^2}{4}+ \alpha d_{\min}}
	\coth \Big( \sqrt{\tfrac{\theta^2}{4}+ \alpha d_{\min}}(T-t)\Big)
	-\frac{\theta}{2} \label{pt},\\
q(t) &:= \sqrt{\alpha d_{\max}}
	\coth\Big( \sqrt{ \alpha d_{\max}}(T-t)\Big) \label{qt}
\end{align}
for $\theta+ \alpha d_{\min} >0$ respectively $\alpha d_{\max}>0$ and
\begin{equation}
p(t) := q(t):=\frac{1}{T-t}, 
\label{ptqtalpha0}
\end{equation}
for $\theta= \alpha d_{\min} =0$ respectively $\alpha d_{\max}=0$. The convergence is compact and strictly increasing on $[0,T)$.
Furthermore,
$p(l,T),q(l,T) \nearrow \infty$ as $l \rightarrow \infty$.
\end{lem}

\begin{proof}
Point-wise convergence and the formulae for the limits are straightforward
by Equations~\eqref{plt} and \eqref{qlt} respectively by Equation~\eqref{pqlttheta0}. Strict monotonicity follows from the fact that the initial values are strictly 
increasing in $l$. Finally, compact convergence follows from these observations
by Dini's theorem.
\end{proof}

\begin{proof}[Proof of Theorem~\ref{TheoremLimitValueFunc}]
\begin{enumerate}
\item[(i)]
Note first that Theorem~\ref{TheoremOptStrl} implies that for fixed $t\in (-\infty,T]$, 
$C(l,t)$ is increasing in $l$ on $(l_0,\infty)$ for $l_0$ as in Equation~\eqref{lNull} in the sense of Notation~\ref{NotaMatrix}~(i).
The existence of the element-wise limit of $(C(l,t))_{l>0}$ follows directly from this monotonicity
and the boundedness by $\Lambda q(t) I$ for $q(t)$ as in Equation~\eqref{qt} 
respectively~\eqref{ptqtalpha0}.
Compact convergence follows by Dini's theorem due to the monotonicity. Finally, $p(l,T) \nearrow \infty$ implies 
$\lim_{l \rightarrow \infty} c_{\min}(l,T)=\infty$.
\item[(ii)]
The inequalities in~\eqref{EqBoundsLimitC} follow directly from Lemma~\ref{LemBounds} and~(i). Furthermore,
the compact convergence of $C(l,t)$ (and $\tilde{C}(l,t)$) on $[0,T)$ and the fact that $C(l,\cdot)$ solves the Matrix Differential Equation~\eqref{EqDiffEqLimit} implies 
\begin{align*}
C(t)=\lim\limits_{l \rightarrow \infty} C(l,t) &= \lim\limits_{l \rightarrow \infty} \int_0^t 
	\big(C\Lambda^{-1}C+C\tilde{C}C-\alpha \Sigma\big)(l,s) ds + \lim\limits_{l \rightarrow \infty} C(l,0) \\
	&=  \int\limits_0^t \big(C\Lambda^{-1}C+C\tilde{C}C-\alpha \Sigma\big)(s) ds +C(0)
\end{align*}
and hence that $C$ solves the Differential Equation~\eqref{EqDiffEqLimit}
with boundary condition $\lim_{s \rightarrow T-} c_{\min}(s) =\infty$ (as $\lim_{s \rightarrow T-} p(s) =\infty$).
\end{enumerate}
\end{proof}

\subsection{Proof of Theorem~\ref{TheoremAdm}}

We start by proving the following bounds for $X^*(l)$ and $X^*$.

\begin{lem} \label{LemX*neu}
Let $t \in [0,T)$, $x \in \mathds{R}^n$ be the portfolio position at time $t$. 
\begin{enumerate}
\item[(i)]
Let $l>l_0$ for $l_0$ as in Equation~\eqref{lNull} and $s\in [t,T)$. Then  (cf.~Notation~\ref{NotationTau}~(i)),
\begin{align}
X^*(l,s)^\top\Lambda X^*(l,s) 
& \leq 
x^\top \Lambda x \exp\Big(-2 \int_{t}^{s} p(l,u) du \Big) 
	\prod\limits_{t \leq \tau_i\leq s} \frac{q(l,\tau_i)}{p(l,\tau_i)}  \label{IneqXl} \\
&  \leq x^\top \Lambda x \exp\big(\theta (s-t)\big) 
	\frac{\big(T-s+\frac{2 \lambda_{\max}}{2 l+\theta \lambda_{\max}}\big)^2}
	{\big(T-t+\frac{2 \lambda_{\max}}{2 l+\theta \lambda_{\max}}\big)^2}	
	\prod\limits_{t \leq \tau_i\leq s} \frac{q(l,\tau_i)}{p(l,\tau_i)} \quad \text{a.s.,} \label{IneqXl2}
\end{align}
where $p$ and $q$ are as in Equations~\eqref{plt} and~\eqref{qlt}
respectively as in~\eqref{pqlttheta0}.
\item[(ii)]
Let $s\in [t,T)$. Then 
\begin{equation}
X^*(s)^\top\Lambda X^*(s) 
\leq x^\top \Lambda x \exp\big(\theta (s-t)\big) 
	\frac{\big(T-s\big)^2}
	{\big(T-t\big)^2}	
	\prod\limits_{t \leq \tau_i\leq s} \frac{q(\tau_i)}{p(\tau_i)} \quad \text{a.s.,} \notag
\end{equation}
where $p$ and $q$ are as in Equations~\eqref{pt} and~\eqref{qt}
respectively as in~\eqref{ptqtalpha0}.
\end{enumerate}
\end{lem}

\begin{proof}
We prove (i); (ii) follows by exactly the same line of reasoning with the respective bounds.
Let $i \in \mathds{N}$. Observe that on $\{\tau_i <T\}$
\[
\frac{\partial}{\partial s}\big( X^*(l,s)^\top \Lambda X^*(l,s) \big)
 = -2  X^*(l,s)^\top C(l,s) X^*(l,s) \notag \\
 \leq -2 p(l,s) X^*(l,s)^\top \Lambda X^*(l,s) \notag
\]
for $s \in [\tau_i,\tau_{i+1} \wedge T)$ by~Theorem~\ref{TheoremDiffBounds}.
Gronwall's inequality implies
\begin{equation*}
X^*(l,s)^\top \Lambda X^*(l,s) \leq X^*(l,\tau_i)^\top \Lambda X^*(l,\tau_i) \exp\Big(-2 \int_{\tau_i}^{s} p(l,r) dr \Big), 
\end{equation*}
in particular
\begin{equation}
X^*(l,(\tau_{i+1}\wedge T)-)^\top \Lambda X^*(l,(\tau_{i+1}\wedge T)-) 
\leq X^*(l,\tau_{i})^\top \Lambda X^*(l,\tau_{i}) 
 \exp\Big(-2 \int_{\tau_{i}}^{\tau_{i+1}\wedge T} p(l,r) dr \Big) . \label{Ineqt_i}
\end{equation}
This implies (cf.~Inequalities~\eqref{IneqStattBellman}  and~\eqref{IneqAbsch}  in the proof of Lemma~\ref{LemXl*})
\begin{align}
X^*(l,\tau_{i+1})^\top \Lambda X^*(l,\tau_{i+1})
&\leq
	 \frac{1}{p(l,\tau_{i+1})} X^*(l,\tau_{i+1}-)^\top 
	 C(l,\tau_{i+1}) X^*(l,\tau_{i+1}-)  \notag \\
&\leq \frac{q(l,\tau_{i+1})}{p(l,\tau_{i+1})}
	X^*(l,\tau_{i+1}-)^\top  \Lambda X^*(l,\tau_{i+1}-)\quad \text{a.s.}
 	\label{Ineqt-}
\end{align}
Using Inequalities \eqref{Ineqt_i} and \eqref{Ineqt-}, we obtain Inequality~\eqref{IneqXl}
inductively as before. Inequality~\eqref{IneqXl2}
follows from
\begin{equation*}
\frac{1}{T-r+\frac{2 \lambda_{\max}}{2 l+\theta \lambda_{\max}}} -\frac{\theta}{2} \leq p(l,r)
\end{equation*}
which is a direct consequence of Corollary~\ref{CorScalarRiccati}.
\end{proof}

The main step in the proof of the admissibility of $u^*$ is to show that the liquidation constraint holds (cf.~Definition~\ref{DefAdmStr}~(iv)). This is (in particular) accomplished in the following proposition.

\begin{prop} \label{PropLiquidation}
Let $t \in[0,T)$ and $x \in \mathds{R}^n$ be the portfolio position at time $t$. 
\begin{enumerate}
\item[(i)]
$X^*(l,\cdot) \overset{l\rightarrow \infty}{\longrightarrow} X^*(\cdot)$  a.s.~compactly on $[t,T)$.
\item[(ii)]
$l \cdot \|X^*(l,T)\|_2^2 \overset{l\rightarrow \infty}{\longrightarrow}  0$ a.s.~and in $L^1$
and
$X^*(l,T) \overset{l\rightarrow \infty}{\longrightarrow} X^*(T)=\lim\limits_{s\rightarrow T-} X^*(s)= 0$ a.s.
\end{enumerate}
\end{prop}

\begin{proof}
\begin{enumerate}
\item[(i)]
The spectral norm $\|\cdot\|_{2,2}$ is equivalent to the matrix maximum norm, and therefore
the element-wise convergence results from Theorem~\ref{TheoremLimitValueFunc}~(i) transfer
to the corresponding results for the spectral norm.

Let $t\leq T' <T$. On $\{ \tau_i <T'\}$, $X^*$ and $X^*(l)$
solve the respective ordinary differential equations
\[
X' = - \xi^*( \cdot,X) = - \Lambda^{-1} C X, \quad
X' = - \xi^*(l,\cdot,X) = - \Lambda^{-1} C(l) X
\] 
on the interval $[\tau_i,\tau_{i+1}\wedge T')$. We prove that
the assertion follows from the continuous dependence of solutions of ordinary differential equations on the right hand side and initial values. To this end, we first require some preliminary observations.

For $s \in [t,T']$ and $x,y \in \mathds{R}^n$, we have
\begin{equation*}
\|\Lambda^{-1} C(s)x-\Lambda^{-1} C(s)y\|_2  \leq \max\limits_{s\in [t,T']} \|Q(s)\|_{2,2} \|x-y\|_2 
=:L \|x-y\|_2
\end{equation*}
(cf.~Theorem~\ref{TheoremLimitValueFunc}~(ii)), i.e., for all $s \in [t,T']$, $\xi^*(s,\cdot)$ is Lipschitz continuous on $\mathds{R}^n$ with Lipschitz
constant $L=L(T')$ independent of $s$.
Furthermore, there exits a constant $K_1\geq 1$ such that for $s \in [t,T']$,
$
\|\bar{C}(s)C(s)\|_{2,2} \leq K_1.
$
We now show by induction on $i\in \mathds{N}$ that for all $\epsilon >0$, there exits an $l_i > l_0$ such that
$l_{i} \geq l_{i-1}$ and for all $l \geq l_i$, $s \in [t,\tau_{i} \wedge T']$, 
$
\| X^*(l,s) - X^*(s) \|_2 < \epsilon.
$

The assertion is clear for $i=0$. Let $i>0$ and $\epsilon >0$. By the induction hypothesis, there exists
$l_{i-1}>l_0$ such that for $l>l_{i-1}$,
\begin{equation} \label{AbschContDep1}
\| X^*(l,\tau_{i-1}) - X^*(\tau_{i-1})\|_2 <\epsilon\frac{ e^{-L(T'-t)}}{6 K_1 }.
\end{equation}
Note that on $\{\tau_{i-1}\geq T'\}$ the induction step is trivial. We therefore fix some 
$
\omega \in \{\tau_{i-1}<T'\}.
$
Let now $l_{i} \geq l_{i-1}$ such that for $l>l_{i}$, $s \leq \tau_i\wedge T'$ 
(recall the uniform convergence of $(C(l,s))_l$ on $[t,T']$, Theorem~\ref{TheoremLimitValueFunc}~(i))
\begin{equation} \label{AbschContDep2}
\|\Lambda^{-1} C(l,s) -\Lambda^{-1} C(s) \|_{2,2}
\leq \|\Lambda^{-1} \|_{2,2} \| C(l,s) - C(s) \|_{2,2} < \epsilon \frac{ e^{-L(T'-t)}}{6 (T'-t)K_1 }
\end{equation}
and
\begin{equation} \label{IneqCl}
\| \bar{C}(l,s)C(l,s)-\bar{C}(s)C(s)\|_{2,2} <  \frac{\epsilon}{3K^2}
\end{equation}
for $K$ as in Lemma~\ref{LemXl*}~(ii).
By the continuous dependence of solutions of systems of ordinary differential equations 
on the right hand side and initial values,
we have for $s \in [\tau_{i-1},\tau_{i}\wedge T')$ (by Inequalities~\eqref{AbschContDep1},
and~\eqref{AbschContDep2}),
\[
\| X^*(l,s,\omega) - X^*(s,\omega)\|_2 
\leq \bigg( \epsilon \frac{ e^{-L(T'-t)}}{6 K_1}
+(T'-t) \epsilon \frac{ e^{-L(T'-t)}}{6 (T'-t)K_1}\bigg) e^{L(T'-t)} \\
 =\frac{\epsilon}{3 K_1},
\]
in particular 
\begin{equation} \label{IneqX*K1}
\| X^*(l,(\tau_{i}\wedge T')-) - X^*((\tau_{i}\wedge T')-)\|_2 
\leq \frac{\epsilon}{3 K_1}.
\end{equation}
We can conclude by using the Inequalities~\eqref{IneqCl}
and~\eqref{IneqX*K1}:
\begin{align}
& \|X^*(l,\tau_{i}(\omega) \wedge T',\omega) - X^*(\tau_{i}\wedge T',\omega)\|_2 \notag \\*
&\qquad = \big\|X^*(l,(\tau_{i}(\omega) \wedge T')-,\omega) 
	- \bar{C}(l,\tau_{i}(\omega) \wedge T') C(l,\tau_{i}(\omega) \wedge T') 
	X^*(l,(\tau_{i}(\omega) \wedge T')-,\omega) \notag \\*	
&\qquad \qquad  - X^*((\tau_{i}(\omega) \wedge T')-,\omega) 
  + \bar{C}(\tau_{i}(\omega) \wedge T')C(\tau_{i}(\omega) \wedge T') 
	X^*((\tau_{i}(\omega) \wedge T')-,\omega) \|_2 \notag  \\
&\qquad \leq  \underbrace{\| X^*(l,(\tau_{i}(\omega) \wedge T')-,\omega) 
	- X^*((\tau_{i}(\omega) \wedge T')-,\omega) \big\|_2}_
	{\leq \epsilon/(3 K_1)\leq \epsilon /3 \text{ by Inequality~\eqref{IneqX*K1}}} \notag \\* 
&\qquad \qquad + \underbrace{ \|X^*(l,(\tau_{i}(\omega) \wedge T')-,\omega)\|_2 
	\| \bar{C}(l,\tau_{i}(\omega) \wedge T')C(l,\tau_{i}(\omega) \wedge T')
	-\bar{C}(\tau_{i}(\omega) \wedge T')C(\tau_{i}(\omega) \wedge T')\|_{2,2}}_
	{\leq \epsilon/3 \text{ by Inequality~\eqref{IneqCl}}}
	\notag \\
&\qquad  \qquad  + \underbrace{\|\bar{C}(\tau_{i}(\omega) \wedge T')C(\tau_{i}(\omega) \wedge T')\|_2  
	\| X^*(l,(\tau_{i}(\omega) \wedge T')-,\omega) 
	- X^*((\tau_{i}(\omega) \wedge T')-,\omega)\|_2}_{<\epsilon/3 
	\text{ by Inequality~\eqref{IneqX*K1}}} \notag \\*
& \qquad < \epsilon \notag
\end{align}
as required. 
\item[(ii)]
For fixed $\omega \in \Omega$, we have by Lemma~\ref{LemX*neu}~(i) that
\begin{equation} 
l \cdot \|X^*(l,T,\omega)\|_2^2 \leq \frac{1}{\lambda_{\min}} 
	x^\top \Lambda x \exp\big(\theta (T-t)\big) 
	\frac{l \big(\frac{2 \lambda_{\max}}{2 l+\theta \lambda_{\max}}\big)^2}
	{\big(T-t+\frac{2 \lambda_{\max}}{2 l+\theta \lambda_{\max}}\big)^2}	
	\prod\limits_{t \leq \tau_i\leq s} \frac{q(l,\tau_i)}{p(l,\tau_i)}.\label{IneqXlTLimit0}
\end{equation}
Furthermore,
\begin{equation*}
\mathbb{E}\Big[\prod\limits_{t\leq \tau_i\leq T} \frac{q(l,\tau_{i})}{p(l,\tau_{i})}\Big] <K_2 <\infty
\end{equation*}
for some constant $K_2$ independent of $l$; thus for almost all $\omega \in \Omega$, 
there exists a constant $K_3(\omega)$ such that
\begin{equation*}
\prod\limits_{t\leq \tau_i(\omega)\leq T} \frac{q(l,\tau_{i}(\omega))}{p(l,\tau_{i}(\omega))} <K_3(\omega).
\end{equation*}
Therefore, Inequality~\eqref{IneqXlTLimit0} implies
\begin{equation*}
\mathbb{E}\big[l \cdot \|X^*(l,T,\omega)\|_2^2\big] \overset{l\rightarrow \infty}{\longrightarrow} 0
\quad
\text{and}
\quad
l \cdot \|X^*(l,T,\omega)\|_2^2 \overset{l\rightarrow \infty}{\longrightarrow} 0 \quad \text{a.s.} 
\end{equation*}

Finally, Lemma~\ref{LemX*neu}~(ii) implies that
$\lim\limits_{s \rightarrow T-} \| X^*(s)\|_2 =0$ a.s.,
finishing the proof.
\end{enumerate}
\end{proof}

We are now able to prove that $u^*$ is indeed an admissible liquidation strategy. The 
main step towards this goal is accomplished by Proposition~\ref{PropLiquidation}~(ii).
It remains thus to show that $u^*$ fulfills the moment conditions in Definition~\ref{DefAdmStr}~(ii).

\begin{proof}[Proof of Theorem~\ref{TheoremAdm}]
Definition~\ref{DefAdmStr}~(i) and~(iii) are clear and (iv)~follows from 
Proposition~\ref{PropLiquidation}~(ii).

Furthermore, we have
\begin{equation*}
\mathbb{E} \Big[\int_t^T\ \|\xi(s,X^*(s))\|_2^4 ds \Big] 
	\leq \mathbb{E} \Big[\int_t^T\ q(s)^4 \|X^*(s)\|_2^4 ds \Big] <\infty
\end{equation*}
by Lemma~\ref{LemXl*}~(ii) and Proposition~\ref{PropLiquidation}~(i).
Finally, $\bar{C}(s) \leq n C(s)^{-1}$ by Corollary~\ref{CorMatIneq}
and thus, as $\frac{q(s)}{p(s)}$ admits a continuous extension on $[t,T]$, there exists a constant $\bar{K}$ independent of $s$ such that
\begin{equation} \label{IneqbarCC}
\| \bar{C}(s) C(s) \|_{2,2}^8 \leq n^8 \|C^{-1}\|_{2,2}^8 \|C\|_{2,2}^8 \leq 
\frac{n^8 \lambda_{\max}^8 q(s)^8}{\lambda_{\min}^8 p(s)^8}
\leq \bar{K}. 
\end{equation}
Using Lemma~\ref{LemXl*}~(ii) and Proposition~\ref{PropLiquidation}~(i) again, we can deduce from Inequality~\eqref{IneqbarCC} that
\begin{equation*}
\mathbb{E} \Big[\int_t^T\ \|\eta(s,X^*(s-)\|_2^8 ds \Big] 
	\leq \mathbb{E} \Big[\int_t^T\ \| \bar{C}(s) C(s) \|_2^8 \|X^*(s-)\|_2^8 ds \Big] < \infty.
\end{equation*}
\end{proof}

\subsection{Proof of Theorem~\ref{TheoremOptLiq}} \label{SecProofMain}

We can directly deduce compact convergence of the optimal trading intensity in the primary venue from Theorem~\ref{TheoremLimitValueFunc}~(i) and Proposition~\ref{PropLiquidation}~(i).

\begin{cor} \label{Corxil*limit}
Let $t \in [0,T)$ and $x \in \mathds{R}^n$ be the portfolio position at time $t$. Then
\begin{equation*}
\xi^*(l,\cdot,X^*(l,\cdot)) \longrightarrow \xi^*(\cdot,X^*(\cdot)) \quad \text{a.s. compactly on } [t,T)
\end{equation*}
as $l \rightarrow \infty$.
\end{cor}

This enables us to finally prove the main result of the article.

\begin{proof}[Proof of Theorem~\ref{TheoremOptLiq}]
We fix $t \in [0,T)$ and $x\in \mathds{R}^n$. 
Note first that we have
\begin{equation} \label{IneqLeq}
v(t,x) \geq \lim\limits_{l\rightarrow \infty} \tilde{v}(l,t,x).
\end{equation}
For the converse inequality, let 
$A:=\{ \pi(T)=\pi(t) \}$
be the set of scenarios without any dark pool execution in $[0,T]$ and 
$K:(l_0,\infty)\times \Omega \longrightarrow [0,\infty]$ be the following cost function:
\[ 
K(l,\omega) := \int_t^T \left(\xi^*(l,s,X^*(l,s,\omega))^\top \Lambda \xi^*(l,s,X^*(l,s,\omega)) 
	+\alpha X^*(l,s,\omega)^\top \Sigma X^*(l,s,\omega)\right)ds 
	 +l  \|X^*(l,T,\omega)\|_2^2. \notag
\]
Then $\mathbb{P}[A]>0$ and for $\omega \in A$, $K(l,A):=K(l,\omega)$ is independent of the specific scenario $\omega$
almost surely. By optimality of $u^*(l)$ (Theorem~\ref{TheoremOptStrl}), 
$K(l,A)$ is an upper bound for $K(l,\cdot)$ almost surely.
As $\lim_{l\rightarrow \infty} \tilde{v}(l,t,x)$ is bounded and 
$\mathbb{P}[A]>0$, there exists a constant $K$ such that for all $l>l_0$,
$K(l,A) \leq K$.
By the dominated convergence theorem this implies
\[
\lim\limits_{l\rightarrow \infty} \tilde{v}(l,t,x)
= \lim\limits_{l\rightarrow \infty}\mathbb{E}[  K(l)] \notag\\ 
=\mathbb{E}[ \lim\limits_{l\rightarrow \infty} K(l)]. 
\]
By Proposition~\ref{PropLiquidation} and Corollary~\ref{Corxil*limit},
the limit in the last expression exists, so Fatou's lemma yields
\begin{align}
& \lim\limits_{l\rightarrow \infty} \tilde{v}(l,t,x) \notag \\*
& \qquad \geq \mathbb{E} \bigg[ \int_t^T  \lim\limits_{l\rightarrow \infty}  \Big( \xi^*(l,s,X^*(l,s))^\top 
	\Lambda \xi^*(l,s,X^*(l,s))  + \alpha X^*(l,s)^\top \Sigma X^*(l,s) \Big) ds
	+ \lim\limits_{l\rightarrow \infty}l \cdot \|X^*(l,T)\|_2^2\bigg] \notag \\
&\qquad = \mathbb{E}\bigg [ \int_t^T \left( \xi^*(s,X^*(s))^\top 
	\Lambda \xi^*(s,X^*(s)) 
	+\alpha X^*(s)^\top \Sigma X^*(s)\right)ds \bigg] \notag \\
&\qquad \geq v(t,x). \label{IneqAndereRichtung}
\end{align}
The Inequalities~\eqref{IneqLeq} and~\eqref{IneqAndereRichtung} establish that $u^*$ solves
the Optimization Problem~\eqref{EqValueFct} and that the value function is given by $v$.
For uniqueness, let $u=(\xi,\eta),\tilde{u}=(\tilde{\xi},\tilde{\eta}) \in \mathbb{A}(t,x)$
and $\mu\in (0,1)$. We define the convex combination $\bar{u}=(\bar{\xi},\bar{\eta})$:
\[
\bar{\xi}(s) =\mu \xi(s) + (1-\mu) \tilde{\xi}(s), \quad
\bar{\eta}(s) =\mu \eta(s) + (1-\mu) \tilde{\eta}(s)
\]
for $s \in [t,T)$. Thus,
$
X^{\bar{u}}(s) = \mu X^u(s) + (1-\mu) X^{\tilde{u}}(s)
$
and $\bar{u} \in \mathbb{A}(t,x)$.
Notice that 
\begin{equation} \label{ImplicationEqual}
\mathbb{P}\otimes\mathbb{\lambda}\big[u \not= \tilde{u}\big]>0
\quad \text{implies} \quad \mathbb{P}\otimes\mathbb{\lambda}\big[\xi \not= \tilde{\xi}\big] >0
\end{equation}
as else
$
\mathbb{P}[ \lim_{s\rightarrow T-} X^u(s)
	\not=  \lim_{s\rightarrow T-} X^{\tilde{u}}(s)]>0,
$
a contradiction to Definition~\ref{DefAdmStr}~(iv).
Hence, 
\begin{align}
J(t,x,\bar{u}) &= \mathbb{E} \Big[ \int_t^T 
		f\big(\bar{\xi}(r),  X^{\bar{u}}(r) \big) dr \Big] \notag\\
&\leq \mathbb{E} \Big[ \int_t^T 
		\mu f\big( \xi(r) , X^u(r)\big) + (1-\mu) f \big(\tilde{\xi}(r), X^{\tilde{u}}(r) \big) dr \Big] 
		\label{ConvexityForProof}\\
& = \mu J(t,x,u) + (1-\mu) J(t,x, \tilde{u}), \notag
\end{align}
where Inequality~\eqref{ConvexityForProof} follows from the convexity of $f$. We have equality 
in Inequality~\eqref{ConvexityForProof}
if and only if $u = \tilde{u}$ $\mathbb{P}\otimes\mathbb{\lambda}$ - a.s. by strict convexity of
$f$ in the first argument and~\eqref{ImplicationEqual}.
\end{proof}

\section{Proofs of the results of Section~\ref{SecPortfolioProp}} \label{SecAppendixProp}

\subsection{Proofs of the results of Section~\ref{SecContPropOne}}

We first require the following elementary result.

\begin{lem} \label{LemDiscrPropeindim}
Let $0 < a < b$, $x>0$. Then
\begin{equation*} 
0>\frac{d}{d x} \frac{\sinh(a x)}{\sinh(b x)} 
>(a - b)  \frac{\sinh(a x)}{\sinh(b x)} .
\end{equation*}
\end{lem}

\begin{proof}
Note first that
\begin{equation}
\frac{d}{dx} \frac{\sinh(a x)}{\sinh(b x)} =
	\frac{a \cosh(a x) \sinh(b x) - b \cosh(b x) \sinh(a x)}
	{\sinh^2(b x)}. \label{IneqDiscrLemmaProp1*}
\end{equation}
The result follows from elementary calculus (see~\cite{Kratz2011} for details).

\end{proof}

\begin{proof}[Proof of Proposition~\ref{PropPropertiesn=1Cont}]
We let $t\in [0,T)$ and compute
\begin{align*}
\frac{\partial}{\partial \theta} C(t;\theta) &= \frac{\Lambda \theta 
	\coth \big( \frac{\tilde{\theta}}{2}(T-t) \big)}{2 \tilde{\theta}} - \frac{\Lambda \theta (T-t)}
	{4 \sinh^2 \big(\frac{\tilde{\theta}}{2}(T-t)\big)} - \frac{\Lambda}{2} \\
& \leq \frac{\Lambda \theta \big( \cosh \big(\frac{\tilde{\theta}}{2}(T-t)\big)
	\sinh \big(\frac{\tilde{\theta}}{2}(T-t)\big) - \frac{\tilde{\theta}}{2}(T-t) 
 	- \sinh^2 \big(\frac{\tilde{\theta}}{2}(T-t)\big) \big)}
 	{2 \tilde{\theta} \sinh^2\big( \frac{\tilde{\theta}}{2}(T-t) \big)} <0 
\end{align*}
for $\theta>0$ since
\begin{align*}
&\cosh \big(\frac{\tilde{\theta}}{2}(T-t)\big)
	\sinh \big(\frac{\tilde{\theta}}{2}(T-t)\big) - \frac{\tilde{\theta}}{2}(T-t) 
 	- \sinh^2 \big(\frac{\tilde{\theta}}{2}(T-t)\big) \\*
&\qquad = \sinh \big(\frac{\tilde{\theta}}{2}(T-t)\big) \Big(
	\cosh \big(\frac{\tilde{\theta}}{2}(T-t)\big)  
 	- \sinh \big(\frac{\tilde{\theta}}{2}(T-t)\big) \Big) - \frac{\tilde{\theta}}{2}(T-t) \\
&\qquad = \Big(\frac{1-\exp(-\tilde{\theta}(T-t) )}{2} \Big)- \frac{\tilde{\theta}}{2}(T-t)  < 0
\end{align*}
(note that $\frac{1}{2}\big( 1-\exp(-2 x) \big)-x < 0$ for $x>0$). This
establishes the first and the second 
assertion directly;
the third assertion follows from the first equality in Equation~\eqref{EqSingleAssetTrajectory}.
For the proof of~(iv), we note first that
\begin{equation} \label{EqNoExContn=1}
\mathbb{E} [X^*(t;\theta)] = \mathbb{P}[\pi(t)=0] \cdot \tilde{X}(t;\theta)
	= \frac{ \sinh \big(\frac{\tilde{\theta}}{2}(T-t)\big) 
	\exp \big(-\frac{\theta}{2} t\big) }{ \sinh \big(\frac{\tilde{\theta}}{2}T \big) } x.
\end{equation}
We compute for $\theta>0$,
\begin{align*}
\frac{\partial}{\partial \theta} \mathbb{E} [X^*(t;\theta)] &= 
	\frac{x}{\sinh^2 \big(\frac{\tilde{\theta}}{2}T \big)} 
		\Big( \frac{\theta (T-t)}{2 \tilde{\theta}} 
	\sinh \big(\frac{\tilde{\theta}}{2}T\big) 
	\cosh \big(\frac{\tilde{\theta}}{2}(T-t)\big)
	\exp\big(-\frac{\theta}{2} t \big) \\*
&\qquad	- \frac{t}{2}
	\sinh \big(\frac{\tilde{\theta}}{2}T\big)
	\sinh \big(\frac{\tilde{\theta}}{2}(T-t)\big)
	\exp\big(-\frac{\theta}{2} t \big) 
	 -\frac{\theta T}{2 \tilde{\theta}} 
	\cosh \big(\frac{\tilde{\theta}}{2}T\big)
	\sinh \big(\frac{\tilde{\theta}}{2}(T-t)\big)
	\exp\big(-\frac{\theta}{2} t \big) \Big) \\
&< \frac{\theta \exp\big(-\frac{\theta}{2} t \big) x}
	{2 \tilde{\theta}\sinh^2 \big(\frac{\tilde{\theta}}{2}T \big)} 
		\Big(  (T-t) 
	\sinh \big(\frac{\tilde{\theta}}{2}T\big) 
	\cosh \big(\frac{\tilde{\theta}}{2}(T-t)\big) - T
	\cosh \big(\frac{\tilde{\theta}}{2}T\big)
	\sinh \big(\frac{\tilde{\theta}}{2}(T-t)\big) \Big) <0
\end{align*}
by Lemma~\ref{LemDiscrPropeindim} (cf.~also Equation~\eqref{IneqDiscrLemmaProp1*}), finishing the proof of~(iv).

We have
\[
\mathbb{E} [X^*(t;\theta)^2] = \mathbb{P}[\pi(t)=0] \cdot \tilde{X}(t;\theta)^2
	= \frac{ \sinh \big(\frac{\tilde{\theta}}{2}(T-t)\big)}
	{ \sinh \big(\frac{\tilde{\theta}}{2}T \big) } x^2.
\]
This term is differentiable and strictly decreasing in $\theta$ by Lemma~\ref{LemDiscrPropeindim} 
(note that $\tilde{\theta}$ is strictly increasing in $\theta$). Thus, by Fubini's theorem,
\[
\frac{\partial}{\partial \theta}
	\mathbb{E} \Big[\int\limits_0^T X^*(t;\theta)^2 dt \Big] 
	= \frac{\partial}{\partial \theta} \int\limits_0^T \mathbb{E} [X^*(t;\theta )^2] dt 
	= \int\limits_0^T \frac{\partial}{\partial \theta} \mathbb{E} [X^*(t;\theta)^2] dt <0,
\]
establishing~(v).

Finally, we note that
\[
\mathbb{E} [\xi^*(t,X^*(t;\theta);\theta)^2] 
	= \mathbb{P}[\pi(t)=0] \cdot \frac{C(t;\theta)^2}{\Lambda^2}\tilde{X}(t;\theta)^2
	= \frac{C(t;\theta)^2}{\Lambda^2} \mathbb{E} [X^*(t;\theta)^2]
\]
by~Equation~\eqref{EqNoExContn=1}.
This term is differentiable and strictly decreasing in $\theta$ as both terms are positive and 
strictly increasing in $\theta$. Similarly as before, we deduce~(vi).
\end{proof}

\begin{proof}[Proof of Proposition~\ref{PropPropertiesn=1ContImpact}]
(i) follows directly from the cost functional $J$. For (ii), we compute
\begin{equation*}
\frac{\partial}{\partial \Lambda} \xi^*(t;\Lambda) =  \frac{1}{2} \frac{\partial \tilde{\theta}}{\partial \Lambda} \Big(
	\frac{\cosh\big(\frac{\tilde{\theta}}{2}(T-t)\big)
	\sinh \big(\frac{\tilde{\theta}}{2}(T-t)\big) -\tfrac{\tilde{\theta}}{2}(T-t)}
	{\sinh^2 \big(\frac{\tilde{\theta}}{2}(T-t)\big)} \Big) <0
\end{equation*}
as $\frac{\partial \tilde{\theta}}{\partial \Lambda} <0$. Monotonicity of $\tilde{X}$ follows as in the proof of Proposition~\ref{PropPropertiesn=1Cont}~(iii). A similar calculation yields 
\[
\frac{\partial}{\partial(\alpha \Sigma)} \xi^*(t;\alpha\Sigma) >0
\] 
as $\frac{\partial \tilde{\theta}}{\partial(\alpha \Sigma)} >0$, finishing the proof.
\end{proof}

\subsection{Proofs of the results of Section~\ref{SubSecContMuliAsset}}

\begin{proof}[Proof of Proposition~\ref{PropOptDPOrder}]
We have $v(t,x-\eta e_i)= (x-\eta e_i)^\top C(t)(x-\eta e_i)$ which can easily be seen to be minimized by $\eta_1^*(t,x)$.
\end{proof}

\begin{proof}[Proof of Proposition~\ref{PropPropertiesCharactWell}]
We prove the case $x_1,x_2>0$, $\rho<0$, i.e., $x$ is well diversified. Let $u^*$ be the optimal strategy for the initial portfolio position $(x_1,-x_2)^\top$. For $(x_1,x_2)^\top$, we define the strategy $u \in \mathbb{A}(t,(x_1,x_2)^\top)$ in such a way that for $i=1,2$, $s\geq t$, $X^{u}_i(s) \geq 0$ and $|X^{u}_i(s)| = |X^{u^*}_i(s)|$;
this is achieved by changing the signs of the trading intensities and by adjusting the dark pool orders appropriately if necessary. In particular, we have $|\xi_i(s, X^{u}(s))| = |\xi^*_i(s, X^{u^*}(s))|$ and both strategies yield the same impact costs. On the other hand, the risk costs of $u$ are strictly smaller as $\rho<0$. Hence,
\[
v(t,(x_1,-x_2)^\top) =J(t,(x_1,-x_2)^\top,u^*) >J(t,(x_1,x_2)^\top,u) \geq v(t,(x_1,x_2)^\top)
\]
as desired. The remaining cases follow accordingly.
\end{proof}

\begin{proof}[Proof of Proposition~\ref{PropPropertiesValueWellMonotone}]
We prove the case $x_1,x_2>0$, $\rho < \tilde{\rho} < 0$ and proceed similarly as in the proof of Proposition~\ref{PropPropertiesCharactWell}. Let $\tilde{u}$ be the optimal strategy for $\tilde{\rho}$. For $\rho$, we define the strategy $u \in \mathbb{A}(t,x)$ in such a way that for $i=1,2$, $s\geq t$,
$X^{u}_i(s) \geq 0$ and $|X^{u}_i(s)| = |X^{\tilde{u}}_i(s)|.$
As $\rho < \tilde{\rho} < 0$, this yields
\[
v(t,x;\rho) \leq J(t,x,u;\rho) < J(t,x,\tilde{u};\tilde{\rho})= v(t,x;\tilde{\rho}).
\]
The remaining cases follow in the same way.
\end{proof}

For the proof of Proposition~\ref{PropMatrixEntries}, we first require the following symmetry results of the value function. 

\begin{lem} \label{LemPropSymm}
Let $t \in [0,T)$, $x\in \mathds{R}^2$ and $\rho \in [-1,1]$. Then,
$v(t,x;\rho) =v(t,-x,;\rho)$ and $v(t,(x_1,x_2)^\top;\rho) =v(t,(x_1,-x_2)^\top;-\rho)$.
\end{lem}

\begin{proof}
We have $J(t,x,u)=J(t,-x,-u)$ and hence the first assertion follows. The second assertion follows from $J(t,(x_1,x_2)^\top,(u_1,u_2)^\top;\rho) =J(t,(x_1,-x_2)^\top,(u_1,-u_2)^\top;-\rho)$.
\end{proof}

\begin{proof}[Proof of Proposition~\ref{PropMatrixEntries}]
\begin{enumerate}
\item[(i)]
The first assertion follows directly from $C(t)>0$. The second assertion follows from~Lemma~\ref{LemPropSymm} as for $i=1,2$,
$
c_{i,i}(t,\rho)=v(t,e_i;\rho) =v(t,e_i;-\rho)=c_{i,i}(t,-\rho).
$
We directly deduce the third assertion as
\[
c_{1,1}(t;\rho)+c_{2,2}(t;\rho)+ 2 c_{1,2}(t,\rho)=v(t,(1,1)^\top;\rho) =v(t,(1,-1)^\top;-\rho)=c_{1,1}(t;\rho)+c_{2,2}(t;\rho) - 2 c_{1,2}(t;-\rho).
\]
Finally, it follows for $\rho<0$ ($\rho>0$) by Proposition~\ref{PropPropertiesCharactWell} that
\[
c_{1,1}(t;\rho)+c_{2,2}(t;\rho)+2 c_{1,2}(t;\rho)= v(t,(1,1)^\top;\rho)
<(>)v(t,(1,1)^\top;-\rho)=c_{1,1}(t;\rho)+c_{2,2}(t;\rho) - 2 c_{1,2}(t;\rho)
\]
and therefore $c_{1,2}(t;\rho)<0$ ($c_{1,2}(t;\rho)>0$).
\item[(ii)]
For the monotonicity of $c_{1,1}(t;\cdot)$, we let $\rho < \tilde{\rho} < 0$ and proceed similarly as in the proofs of Propositions~\ref{PropPropertiesCharactWell} and~\ref{PropPropertiesValueWellMonotone}. Let $\tilde{u}$ be the optimal strategy for $\tilde{\rho}$. For $\rho$, we define the strategy $u \in \mathbb{A}(t,e_1)$ in such a way that for $i=1,2$, $s\geq t$, $X^{u}_i(s) \geq 0$ and $|X^{u}_i(s)| = |X^{\tilde{u}}_i(s)|$.
As $\rho < \tilde{\rho} < 0$, this yields $J(t,e_1,u;\rho) \leq J(t,e_1,\tilde{u};\tilde{\rho})$ with equality if and only if $X^{\tilde{u}}_2(s)=0$ a.s. However, we have $\xi_2(t,e_1)=1/\lambda_2 c_{1,2}(t;\tilde{\rho})<0$ and hence $X^{\tilde{u}}_2>0$ in some neighborhood of $t$ with positive probability. Thus,
\[
c_{1,1}(t;\rho)=v(t,e_1;\rho) \leq J(t,e_1,u;\rho) <J(t,e_1,\tilde{u};\tilde{\rho})=v (t,e_1,\tilde{u};\tilde{\rho})
=c_{1,1}(t;\tilde{\rho}).
\] 
The monotonicity of $c_{2,2}$ and the case $\rho>0$ follow by the same line of reasoning.

Before we proceed, we remark that all symmetry properties and the monotonicity of $c_{1,1}$ and $c_{2,2}$ in $\rho$ also hold for $C(l)$ ($l\geq l_0$) with exactly the same proofs. We now prove monotonicity of $c_{1,2}(l,t;\rho)$ ($l \geq l_0$) first; monotonicity of $c_{1,2}(t;\rho)$ then follows directly from the fact that $\lim_{l \rightarrow \infty} c_{1,2}(l,t;\rho)=c_{1,2}(t;\rho)$. A straightforward computation confirms that $c_{1,2}(l)$ fulfills the following scalar initial value problem
\begin{equation} \label{EqInitialValuec12}
\frac{\partial}{\partial t} c_{1,2}(l,t;\rho) = c_{1,2}(l,t;\rho) \Big( \frac{c_{1,1}(l,t;\rho)}{\lambda_1} +
\frac{c_{2,2}(l,t;\rho)}{\lambda_2} + \theta_1 + \theta_2 \Big) - \alpha \sigma_1 \sigma_2 \rho, 
\quad c_{1,2}(l,T;\rho)=0.
\end{equation}
By the continuous differentiable dependence of $c_{1,2}$ on the parameter $\rho$, we can exchange differentiation with respect to $t$ and $\rho$ and obtain the following initial value problem for $\tfrac{\partial c_{1,2}}{\partial \rho}$:
\[
\frac{\partial}{\partial t} \frac{\partial c_{1,2}}{\partial \rho}(l,t;\rho) = \frac{\partial c_{1,2}}{\partial \rho}(l,t;\rho) f(l,t;\rho) - g(l,t;\rho), 
\quad \frac{\partial c_{1,2}}{\partial \rho}(l,T;\rho)=0
\]
for
\begin{align*}
f(l,t;\rho)&:= \frac{c_{1,1}(l,t;\rho)}{\lambda_1} +
\frac{c_{2,2}(l,t;\rho)}{\lambda_2} + \theta_1 + \theta_2, \\
g(l,t;\rho)&:= \alpha \sigma_1 \sigma_2 - c_{1,2}(l,t;\rho)
\Big( \frac{1}{\lambda_1} \frac{\partial c_{1,1}(l,t;\rho)}{\partial \rho}+
\frac{1}{\lambda_2} \frac{\partial c_{2,2}(l,t;\rho)}{\partial \rho}  \Big);
\end{align*}
this implies
\[
\frac{\partial c_{1,2}}{\partial \rho}(l,t;\rho) = \exp\Big(\int_T^t f(l,s;\rho) ds\Big)
\int_t^T g(l,r;\rho) \exp\Big(\int_r^T f(l,s;\rho) ds\Big) dr > 0
\]
as $g(l,t;\rho) \geq 0$ by~(i) and the monotonicity of $c_{1,1}(l,t;\cdot)$ and $c_{2,2}(l,t;\cdot)$.
\end{enumerate}
\end{proof}

\begin{proof}[Proof of Proposition~\ref{PropPropertiesNeverShort}]
\begin{enumerate}
\item[(i)]
We prove the assertion for $x_1,x_2>0$ and $\rho<0$. 
Let $j \in \mathds{N}$. We assume that $X^*_i(\tau_j)>0$ ($i=1,2$) on $\{\tau_j <T\}$ (cf.~Notation~\ref{NotationTau}). We compare the initial value problem for the controlled process on $[\tau_j,\tau_{j+1}\wedge T)$ (cf.~Equations~\eqref{EqOptXin=2}) with the case $\rho=0$; as $c_{1,2}(s;\rho)<c_{1,2}(s;0)=0$ and $0<c_{i,i}(s;\rho) < c_{i,i}(s;0)$ by Proposition~\ref{PropMatrixEntries}, $X^*_i(s;\rho) \geq X^*_i(s;0)$ for $s \in [\tau_j,\tau_{j+1}\wedge T)$. For $\rho=0$ the two components of $X^*$ evolve independently according to the results of Section~\ref{SecContPropOne}. In particular, the optimal asset position remains positive in $[\tau_j,\tau_{j+1}\wedge T)$. It follows that $X^*_i(s;\rho)>0$ for all $s < \tau_{j+1}\wedge T$ and that $X^*_i(\tau_{j+1}-)>0$ on $\{\tau_{j+1} <T\}$.
Applying Equations~\eqref{EqOptEtan=2} and Proposition~\ref{PropMatrixEntries}, we obtain on $\{ \tau_{j+1}<T\}$,
\[
X_1^*(\tau_{j+1}-)-\eta_1(\tau_{j+1},X^*(\tau_{j+1}-))
=-\frac{c_{1,2}(\tau_{j+1};\rho)}{c_{1,1}(\tau_{j+1};\rho)} X^*_2(\tau_{j+1}-) >0
\]
and the respective result for $X_2^*$. The assertion now follows by induction on $j$. The remaining cases follow accordingly.
\item[(ii)]
We assume $x_1,x_2,\rho>0$. It is clear that $\tau>t$ a.s. The result follows from Equations~\eqref{EqOptEtan=2} and Proposition~\ref{PropMatrixEntries} as for $t\leq s < \tau$,
\[
\eta_1(s,X^*(s))=X_1^*(s) + \frac{c_{1,2}(s;\rho)}{c_{1,1}(s;\rho)} X_2^*(s)>X_1^*(s).
\]
\end{enumerate}
\end{proof}

\begin{proof}[Proof of Proposition~\ref{PropStratMon}]
The assertions follow directly from Equations~\eqref{EqOptXin=2} and~\eqref{EqOptEtan=2} by applying Proposition~\ref{PropMatrixEntries}.
\end{proof}

For the proof of Proposition~\ref{PropWrongDirection}, we require the following results about the monotonicity of $c_{i,i}$ and $c_{1,2}$ in $t$.

\begin{lem} \label{LemMontonCt}
For $i=1,2$ and $t \in [0,T)$, $c_{i,i}(t)$ is increasing in $t$. $c_{1,2}(t)$ is decreasing in $t$ if $\rho >0$ and increasing in $t$ if $\rho<0$.
\end{lem}

\begin{proof}
For fixed 
$l \geq l_0$,
we consider the Initial Value Problem~\eqref{EqODE1}; 
its solution $C(l)$ satisfies
$C(l)' \geq Q(l)'$, where $Q(l)$ solves the initial value problem 
$Q'=1/ \lambda_{\min} Q^2 - \alpha \sigma_{\min} I$, 
$Q(T)=l I$ (cf.~Notation~\ref{NotaMatrix}). As in the proof of Proposition~\ref{PropAdml}, $Q(l)$ can be computed explicitly with 
$Q(l)'(t)>0$ for $l$ large enough, say $l\geq l_1$; in particular $c_{i,i}'(l,t)>0$ and therefore $c_{i,i}'(t)\geq 0$.

For the monotonicity of $c_{1,2}$ we assume $\rho>0$ and consider the Initial Value Problem~\eqref{EqInitialValuec12} for $c_{1,2}(l,t)$. Note first that for all $l$, 
$c_{1,2}'(l,T)= - \alpha \rho \sigma_1 \sigma_2 <0$. Furthermore,
\[
c_{1,2}''(l,t) = c_{1,2}'(l,t) \Big( \frac{c_{1,1}(l,t)}{\lambda_1} +
\frac{c_{2,2}(l,t)}{\lambda_2} + \theta_1 + \theta_2 \Big) +\underbrace{ c_{1,2}(l,t) \Big(\frac{c_{1,1}'(l,t)}{\lambda_1} +
\frac{c_{2,2}'(l,t)}{\lambda_2} \Big)}_{\geq 0 \text{ for } l \geq l_1}.
\]
This implies $c_{1,2}'(l,t)<0$ for  $l \geq l_1$ and hence $c_{1,2}'(t)\leq 0$. The case $\rho<0$ follows accordingly.
\end{proof}

\begin{proof}[Proof of Proposition~\ref{PropWrongDirection}]
We assume $\rho>0$ and that the dark pool order for the first asset is executed at time $\tau_1$, i.e., 
$X^*_1(\tau_1) = - \frac{c_{1,2}(\tau_1)}{c_{1,1}(\tau_1)} X_2^*(\tau_1)$; we further assume $X_2^*(\tau_1)>0$. We have  $\xi_1^*(\tau_1)=0$ and $\xi_2^*(\tau_1) >0$. As $c_{1,2}(s)/c_{i,i}(s)$ is decreasing in $s$ by Lemma~\ref{LemMontonCt}, we have $\xi_2^*(s) >0$ and $\xi_1^*(s)<0$ (in particular $\eta_2^*(s) >0$ and $\eta_1^*(\tau_1)<0$) until the next jump time of $\pi$. The result follows inductively as the number of jumps is almost surely finite. The proof for the remaining cases is analog.
\end{proof}

\begin{proof}[Proof of Proposition~\ref{PropMontonValueLam}]
The first assertion follows directly from the definition of the cost functional $J$. For the second assertion, let $\lambda_i < \tilde{\lambda}_i$.  Then,
\begin{align*}
\frac{v(t,x;\lambda_i)}{\lambda_i} 
	&= \mathbb{E} \Big[ \int_t^T  \Big(\xi^*_i(s)^2 + \sum_{j\not=i} \frac{\lambda_j}{\lambda_i} \xi_j^*(s)^2 
+ \frac{\alpha}{\lambda_i}X^*(s)^\top \Sigma X^*(s) \Big)ds \Big] \\
	&\geq \mathbb{E} \Big[ \int_t^T  \Big(\xi^*_i(s)^2 + \sum_{j\not=i} \frac{\lambda_j}{\tilde{\lambda}_i} \xi_j^*(s)^2 
+ \frac{\alpha}{\tilde{\lambda}_i}X^*(s)^\top \Sigma X^*(s) \Big)ds \Big] \geq \frac{v(t,x;\tilde{\lambda}_i)}{\tilde{\lambda}_i}.
\end{align*}
\end{proof}

\begin{proof}[Proof of Proposition~\ref{PropPropertiesCross}]
\begin{enumerate}
\item[(i)]
We prove the case $x_1,x_2>0$, $\rho<0$, i.e., $x$ is well diversified, and modify the proof of Proposition~\ref{PropPropertiesCharactWell}. Let $u^*$ be the optimal strategy for the initial portfolio position $(x_1,-x_2)^\top$. For $(x_1,x_2)^\top$, we can define a strategy $u \in \mathbb{A}(t,(x_1,x_2)^\top)$ with the following properties: for $s \leq \tau_1$ (where $\tau_1$ is the first jump time of $\pi$)
\begin{align}
|\xi_i(s)|=|\xi_i^{u^*}|,\quad &|X_i^u(s)|\leq|X_i^{u^*}(s)|, \label{IneqX*Cross}\\
\sgn(\xi_1^{u^*}(s))=\sgn(\xi_2^{u^*}(s)) &\Rightarrow \sgn(\xi_1(s))=\sgn(\xi_2(s)),\notag \\
\sgn(X_1^{u^*}(s))=\sgn(X_2^{u^*}(s)) &\Rightarrow \sgn(X^u_1(s))=\sgn(X^u_2(s)).\notag
\end{align}
$u$ needs to be defined carefully by considering all possible combinations of the signs of the trading intensities; note that it can be necessary to change the signs of the positions.
We adjust the dark pool orders in such a way that 
$X^u_i(\tau_1)\in\{ \gamma X_i^{u^*}(\tau_1), - \gamma X_i^{u^*}(\tau_1)\}$ for some $\gamma \leq 1$ (cf.~the inequality in~\eqref{IneqX*Cross}). For $s\geq \tau_1$, we proceed similarly as before by defining $\xi_i(s)=\gamma \xi_i^{u*}(s)$ (or $\xi_i(s)=-\gamma \xi_i^{u*}(s)$) and the dark pool orders as before.

This ensures that the impact costs of $u$ are less or equal than the impact costs of $u^*$ while the risk costs are strictly smaller; hence the assertion follows inductively. The remaining cases follow by the same line of reasoning.
\item[(ii),] (iii)
As in the proof of~Proposition~\ref{PropMatrixEntries}~(i), we obtain that $\sgn(c_{1,2}(t;\rho,\lambda_{1,2}))=\sgn(\lambda_{1,2})=\sgn(\rho)$. The assertions follow from the fact that $\eta^*$ is as in~Equations~\eqref{EqOptEtan=2} (also in the presence of cross price impact).
\end{enumerate}
\end{proof}

\begin{proof}[Proof of Proposition~\ref{PropMontonIntensity}]
Let $\theta_i \leq \tilde{\theta}_i $. For any matrix $C>0$, we have $\tilde{C}(\theta_i) \geq \tilde{C}(\tilde{\theta}_i)$ (cf.~Equation~\eqref{EqTildeC}). Similarly as in the proof of Theorem~\ref{TheoremDiffBounds}, we obtain that the respective solutions of the Initial Value Problem~\eqref{EqODE1} fulfill $C(l,t;\tilde{\theta}_i)\leq C(l,t;\theta_i)$ for $l$ large enough. The assertion follows by taking the limit $l \rightarrow \infty$.
\end{proof}
%

\appendix

\section*{Appendix}

\section{Riccati matrix differential equations} \label{SecAppRiccati}

In this section, we state a well-known comparison result about matrix Riccati Equations in the form in which we apply it in the proof of Theorem~\ref{TheoremDiffBounds}. A standard textbook is the one by~\cite{Reid1972}. A proof for the specific form of the theorem can, e.g., be found in~\cite{Kratz2011}.

\begin{theorem} \label{TheRiccIneq}
Let $A(t),B_P(t),C_P(t),B_Q(t),C_Q(t) \in \mathds{R}^{n\times n}$ be piecewise continuous on $\mathds{R}$. Furthermore, let $B_P(t),C_P(t),B_Q(t),C_Q(t)$ ($t \in \mathds{R}$) and 
$S_P,S_Q \in \mathds{R}^{n\times n}$
be symmetric. 
Let $t_0 > t_1 \geq -\infty$ and 
\begin{equation} \notag
S_Q \leq S_P, \quad 0 \leq B_Q(\cdot) \leq B_P(\cdot), \quad C_P(\cdot) \leq C_Q(\cdot)
\end{equation}
on $(t_2,t_0]$. Assume that the initial value problem
\[
P' = -A^\top P - P A - P B_P P + C_P, \quad
P(t_0) = S_P
\]
possesses a solution $P$ on $(t_1,t_0]$. Then the initial value problem
\[
Q' = -A^\top Q - Q A - Q B_Q Q + C_Q, \quad
Q(t_0) = S_Q
\]
possesses a solution $Q$ on $(t_1,t_0]$ and
$
P(t) \geq Q(t) \quad \text{on } (t_1,t_0].
$
\end{theorem}

We apply the theorem to scalar Riccati equations with constant coefficients and obtain a useful lower bound for their solution.

\begin{cor} \label{CorScalarRiccati}
Let $y$ be the solution of the scalar initial value problem
$y'= y^2 + a y - b$, $y(T) = c$, where $a,b \geq 0$, $c >0$, $b< c^2+ac$ and $d:=a^2/4 + b>0$.
Then for $t \in (-\infty, T]$, 
\begin{equation} \notag
y(t) \geq \frac{1}{T-t+\frac{1}{c+a/2}}-\frac{a}{2}.
\end{equation}
\end{cor}

\begin{proof}
As $d>0$, we have that the solution $z$ of the initial value problem
$
z'= z^2 - d,$ $z(T) = c+ a/2
$
fulfills
\begin{equation} \notag
z(t) \geq \frac{1}{T-t+\frac{1}{c+a/2}} \quad \text{on }(-\infty,T]
\end{equation}
(cf.~Theorem~\ref{TheRiccIneq}; compare $z$ with the solution of $f'=f^2$, $f(T) = c+ \frac{a}{2}
$). The assertion follows directly from the fact that $y(t)=z(t)-a/2$.
\end{proof}

\bibliographystyle{abbrvnat}
\bibliography{bibliography}

\end{document}